\documentclass[11pt]{article}
\usepackage[utf8]{inputenc}
\usepackage{mathtools}
\usepackage{amsmath}
\usepackage{algorithm}
\usepackage{algpseudocode}
\usepackage{amsfonts,amsthm,boxedminipage,color,url}
\usepackage{fullpage}
\usepackage{amssymb}
\usepackage{natbib}
\usepackage{soul}
\usepackage{enumitem}
\usepackage{booktabs}
\usepackage{multirow}
\usepackage[pdfstartview=FitH,colorlinks,linkcolor=blue,filecolor=blue,citecolor=blue,urlcolor=black]{hyperref}
\usepackage{hyperref} 
\usepackage[labelfont=bf]{caption}
\usepackage{aliascnt,cleveref}
\usepackage{authblk}
\usepackage{accents}
\usepackage{tikz}
\usetikzlibrary{calc, arrows.meta, positioning, patterns}
\usepackage{pgfplots}
\pgfplotsset{width=10cm,compat=1.9}
\usepgfplotslibrary{external}
\usepackage{bbm}
\usepackage{float}

\newtheorem{theorem}{Theorem}
\newtheorem{lemma}{Lemma}
\newtheorem{corollary}{Corollary}
\newtheorem{claim}{Claim}
\newtheorem{observation}{Observation}
\newtheorem{proposition}{Proposition}
\newtheorem{definition}{Definition}

\newtheorem{fact}{Fact}
\newtheorem{mainresult}{Main Result}
\newtheorem*{theorem*}{Theorem}

\newcommand{\eps}{\varepsilon}

\newcommand{\reals}{\mathbb{R}}

\newcommand{\A}{\mathbb{A}}
\newcommand{\I}{\mathcal{I}}
\newcommand{\IS}{\mathcal{A}}
\newcommand{\feq}{\bar{f}}
\newcommand{\ceq}{\bar{c}}
\DeclareMathOperator*{\argmax}{arg\,max}

\newcommand{\instance}{\langle n, f, c \rangle}
\newcommand{\instanceI}{I=\langle n, f, c \rangle}
\newcommand{\inst}[3]{\langle #1, #2, #3 \rangle}
\newcommand{\indicator}[1]{\mathbbm{1}[{#1}]}

\newcommand{\demand}[2]{\mathcal{D}_{{#1},{#2}}}
\newcommand{\demSigP}{\demand{\sigma}{p}}
\newcommand{\supply}[2]{\mathcal{D}_{{#1},{#2}}}
\newcommand{\supSigP}{\supply{\sigma}{p}}
\newcommand{\BR}[3]{\demand{#1}{#2}^{#3}}
\newcommand{\BRI}[2]{\BR{#1}{#2}{I}}
\newcommand{\BrSigAlpha}{\BRI{\sigma}{\alpha}}

\newcommand{\DISJ}[1]{\mathsf{DISJ}_{#1}}
\newcommand{\tc}{\tilde{c}}
\newcommand{\tf}{\tilde{f}}
\newcommand{\talpha}{\tilde{\alpha}}
\newcommand{\hf}{\hat{f}}
\newcommand{\hc}{\hat{c}}
\newcommand{\halpha}{\hat{\alpha}}
\newcommand{\tI}{\tilde{I}}
\newcommand{\InstTildeI}{\tI = \inst{n}{f}{\tc}}
\newcommand{\InstFTildeI}{\tI = \inst{n}{\tf}{c}}
\newcommand{\hI}{\hat{I}}
\newcommand{\InstHatI}{\hI = \inst{n+1}{\hf}{\hc}}

\newcommand{\eqRevInst}{\langle n,\feq,\ceq \rangle}
\newcommand{\perturbedInst}{\langle n,\feq_k,\ceq \rangle}
\newcommand{\familyI}{\I = \{ \perturbedInst \}_{k=1}^{2^n-1}}

\definecolor{darkpastelgreen}{rgb}{0.01, 0.75, 0.24}
\definecolor{darkorchid}{rgb}{0.6, 0.2, 0.8}
\crefname{prop}{property}{properties}

\author{
Paul Dütting\thanks{Google Research, Zurich, Switzerland. Email: \texttt{duetting@google.com}} 
\;\;\;\; Michal Feldman\thanks{Tel Aviv University, Israel. Email: \texttt{mfeldman@tauex.tau.ac.il}}
\;\;\;\; Yoav Gal-Tzur\thanks{Tel Aviv University, Israel. Email: \texttt{yoavgaltzur@mail.tau.ac.il}}
\;\;\;\; Aviad Rubinstein\thanks{Stanford University, USA. Email: \texttt{aviad@cs.stanford.edu}}
}
\date{\today}

\title{
When Contracts Get Complex: Information-Theoretic Barriers\thanks{This is the full version of a SODA 2026 paper. A preliminary version of this work was circulated under the title ``The Query Complexity of Contracts.''
This project has been partially funded by the European Research Council (ERC) under the European Union's Horizon 2020 program (grant agreement No.~866132), by the European Union's Horizon Europe Program (grant agreement No.~101170373), by an Amazon Research Award, by the Israel Science Foundation Breakthrough Program (grant No.~2600/24), and by a grant from TAU Center for AI and Data Science (TAD), and by the NSF-BSF (grant number 2020788).
Aviad Rubinstein is supported by NSF CCF-2112824, and a David and Lucile Packard Fellowship.}
}

\begin{document}

\maketitle

\begin{abstract}
In the combinatorial-action contract model (D\"utting et al., FOCS'21) 
a principal delegates the execution of a complex project to an agent, who can choose any subset from a given set of actions.
Each set of actions incurs a cost to the agent,
given by a set function $c$, and induces an expected reward to the principal, given by a set function $f$. 
To incentivize the agent, the principal designs a contract that specifies the payment upon success, with the optimal contract being the one that maximizes the principal's utility.

It is known that with access to value queries no constant-approximation is possible for submodular $f$ and additive $c$. A fundamental open problem is: does the problem become tractable with demand queries?
We answer this question to the negative, by establishing that finding an optimal contract for submodular $f$ and additive $c$ requires exponentially many demand queries. 
We leverage the robustness of our techniques to extend and strengthen this result to different combinations of submodular/supermodular $f$ and $c$; 
while allowing the principal to access $f$ and $c$ using arbitrary communication protocols.

Our results are driven by novel equal-revenue constructions when one of the functions is additive, immediately implying value query hardness. 
We then identify a new property --- sparse demand --- which allows us to strengthen these results to demand query hardness.
Finally, by augmenting a perturbed version of these constructions with one additional action, thereby making both functions combinatorial, we establish exponential communication complexity.  
\end{abstract}

\section{Introduction}

Contract design is one of the pillars of economic theory (cf., the 2016 Nobel Prize in Economics for Hart and Holmström \cite{Nobel2016}), of which the hidden-action principal-agent model is a cornerstone.
Recently, this area has seen growing attention from the theoretical computer science and algorithmic game theory communities \cite{contractsSurvey}, with combinatorial contract design being a natural focal point.

A prime example of a combinatorial contract setting is the combinatorial-action model of D\"utting, Ezra, Feldman, and Kesselheim \cite{dutting2022combinatorial}.
In this model, a principal (she) delegates the execution of a project, which can either succeed or fail, to an agent (he). 
The agent may pick any subset of a known set of $n$ actions. 
The choice of actions $S$ determines the success probability of the project via a (monotone) success probability function $f:2^{[n]}\rightarrow [0,1]$. The principal enjoys a (normalized) reward of $r = 1$ when the project succeeds, and her reward is zero otherwise. 
Observe that $f$ can also be interpreted as the principal's expected reward, and so we also sometimes refer to it as the reward function. 
Each choice of actions incurs a cost for the agent. The cost of a set of actions is given by a (monotone) cost function $c: 2^{[n]} \rightarrow \reals_{\geq 0}$.

The principal cannot observe which actions are taken by the agent. Thus, in order to incentivize the agent to exert costly effort, the principal offers a (linear) contract, which specifies a payment in case the project succeeds, which we denote by $\alpha \in [0,1]$.
Given a contract, the agent picks his \textit{best-response}, a set of actions that maximizes his utility, defined to be the expected payment $\alpha \cdot f(S)$ minus cost $c(S)$. 
The goal is to compute an optimal contract --- the one maximizing the principal's utility, defined as the expected reward $f(S)$ minus expected payment $\alpha \cdot f(S)$.
The gold standard is an algorithm that runs in time polynomial in $n$, with appropriate query access to the set functions $f$ and $c$. (Note that an explicit description of combinatorial
functions typically requires exponentially many bits.)

While prior work has mostly focused on the case where $f$ is combinatorial and $c$ is additive, the ultimate goal would be a theory that covers scenarios where both $f$ and $c$ are combinatorial.
Of particular interest are submodular and supermodular functions, capturing scenarios where the agents' actions exhibit diminishing or increasing returns. 
It is known that with supermodular $f$ and submodular $c$, an optimal contract can be found with polynomially-many value queries \cite{dutting2024combinatorial}. 
In this work, we show that computing the optimal contract under any other combination of submodular/supermodular $f$ and $c$ faces strong information-theoretic barriers.

\subsection{Our Contribution}
We establish hardness results in two models: one based on \textit{query} complexity, and the other on \textit{communication} complexity. A summary of our results appears in \Cref{tab:overview}.

\begin{table}[t]
  \centering
  \begin{tabular}{|p{2.7cm}|p{4cm}|p{4cm}|p{4cm}|}
    \hline
    & & & \\[-.8em]
    & \textbf{submodular cost} & \textbf{additive cost} & \textbf{supermodular cost} 
    \\[0.3em]
    \hline  
    & & & \\[-.8em]
    \textbf{submodular} & \textit{OPT}: CC+BR & \textit{OPT}: Demand query & \textit{OPT}: CC+BR \\
    \textbf{reward} &  hardness (Thm.\ \ref{thm:cc_submod_f_c}) & hardness (Thm.\ \ref{thm:expDemandQueries}) & hardness (Thm.\ \ref{thm:cc_submod_f_supmod_c}) \\[0.3em]
    & \textit{Approx.:} CC &  &  \\
    & hardness (Thm.\ \ref{thm:cc_inapprox_sub_sub}) &  &  \\
    \hline  
    & \multicolumn{2}{|p{3.8cm}|}{} & \\[-.8em] 
    \textbf{additive} & \multicolumn{2}{|p{3.8cm}|}{} & \textit{OPT}: Supply query \\
    \textbf{reward} & \multicolumn{2}{|p{3.8cm}|}{\multirow{5}{*}{\textit{OPT:} poly time (with value queries) \cite{dutting2024combinatorial}}} & hardness (Thm.\ \ref{thm:expDemandQueries_SupMod_c}) \\[0.3em]
    \cline{1-1}\cline{4-4}  
    & \multicolumn{2}{|p{3.8cm}|}{} & \\[-.8em] 
    \textbf{supermodular} & \multicolumn{2}{|p{3.8cm}|}{} & \textit{OPT}: CC+BR \\
    \textbf{reward} & \multicolumn{2}{|p{3.8cm}|}{} & hardness (Thm.\ \ref{thm:cc_supmod_f_c}) \\[0.3em]
    & \multicolumn{2}{|p{3.8cm}|}{} & \textit{Approx.}: CC \\
    & \multicolumn{2}{|p{3.8cm}|}{} & hardness (Thm.\ \ref{thm:cc_inapprox_sup_sup}) \\
    \hline
  \end{tabular}
  \caption{Overview of our results for computing the optimal contract (\textit{OPT}), 
  and for approximating it to within any finite multiplicative factor (\textit{Approx.}). 
    When either $f$ (reward) or $c$ (cost) is additive, we show hardness of the optimal contract in the appropriate query-complexity model.
  For fully combinatorial $f$ and $c$, we consider two communication-based models:
  standard communication complexity ({CC}) and the stronger model with access to polynomially many best-response queries ({CC+BR}).
  We note that the approximability of the optimal contract for supermodular $c$ (even when $f$ is additive) remains an open problem.
  } 
  \label{tab:overview}
\end{table}

\paragraph{Query Complexity Hardness.}
Our first set of results establishes exponential query complexity lower bounds when (i) $f$ is submodular and $c$ is additive or (ii) $f$ is additive and $c$ is supermodular, even with access to value and demand/supply oracles (see \Cref{sec:model} for definitions). 

Value and demand queries are standard in the literature, and supply queries are the natural analog of demand queries in the case where $c$ is combinatorial (rather than $f$). All three query types have natural economic interpretations.
Notably, when $f$ is combinatorial and $c$ is additive (resp., $f$ is additive and $c$ combinatorial), a single demand (resp., supply) query can return a set of actions that maximizes the agent's utility for a given contract. 
Value queries, which are strictly weaker than demand/supply queries, do not possesses this property \cite{blumrosen2005computational}.

\begin{mainresult}[\Cref{thm:expDemandQueries,thm:expDemandQueries_SupMod_c}]\label{thm:main1}
In any of the settings below, any algorithm that computes the optimal contract requires exponentially-many demand queries to $f$ or supply queries to $c$: 
\begin{itemize}\setlength\itemsep{-0.25em}
    \item $f$ is submodular and $c$ is additive.
    \item $c$ is supermodular and $f$ is additive.
\end{itemize}
\end{mainresult}

Prior to this work, it was known that when $c$ is additive, an optimal contract can be computed with polynomially-many value queries, when $f$ is gross-substitutes (a strict subclass of submodular) \cite{dutting2022combinatorial}. 
It was also known that computing an optimal contract with submodular $f$ (in fact, budget additive $f$) is \textsf{NP}-hard \cite{dutting2022combinatorial}. More recently, it was shown that for submodular $f$ and additive $c$ it is hard to approximate the optimal contract to within any constant factor using value queries alone (assuming \textsf{P} $\neq$ \textsf{NP}) \cite{ezra2023Inapproximability}. 
In both these cases, the hardness stems from the hardness of answering a demand query with value queries. Whether or not the problem remains hard with access to demand queries remained a major open question (explicitly raised in \cite{dutting2022combinatorial} and reiterated by \cite{deo2024supermodular}).

Our result resolves this question in the negative, and thus shows an intrinsic hardness of the contracting problem with submodular $f$ and additive $c$.

\paragraph{Communication Complexity Hardness.}

Our second set of results establishes exponential communication lower bounds in a natural model where $f$ and $c$ are held by two different parties.
We observe that while the case in which $c$ (or $f$) is additive admits a trivial solution, the optimal principal's utility cannot be approximated to within any factor when $f$ and $c$ are both submodular or are both supermodular (see \Cref{sec:CC_inapprox}). For the remaining case of submodular rewards and supermodular costs we show that exponential communication is required to find an optimal contract (this will be implied by our \Cref{thm:main2} below). 

Finally, we consider a combined model, where the parties also gain access to a ``best-response'' oracle -- which returns the agent's favorite set for a given contract.
Our motivation for studying this model is that it generalizes both the demand/supply query model and the communication model. We show that exponential communication is required to compute the optimal contract for any combination of submodular/supermodular $f$ and $c$, except for $f$ supermodular and $c$ submodular (which admits an efficient algorithm), strengthening \Cref{thm:main1} for these cases.

\begin{mainresult}[\Cref{thm:cc_submod_f_c,thm:cc_supmod_f_c,thm:cc_submod_f_supmod_c}]\label{thm:main2}
In any of the settings below, communicating exponentially-many bits is required to compute the optimal contract, even with poly-many best-response queries:
\begin{itemize}\setlength\itemsep{-0.25em}
    \item When $f$ and $c$ are submodular.
    \item When $f$ and $c$ are supermodular.
    \item When $f$ is submodular and $c$ is supermodular.
\end{itemize}
\end{mainresult}

To the best of our knowledge, our results are the first communication complexity results in the contracts literature. The strong impossibilities in the pure communication model highlight the power of best-response oracles when $f$ and $c$ are both submodular or both supermodular, yet we show that computing an optimal contract remains unattainable even with this type of oracle.
We believe that both our model which assumes that $f$ and $c$ are held by two different parties and the techniques developed for showing the communication lower bounds (see \Cref{sec:techniques}), could be stepping stones for further research into the communication complexity of combinatorial contracts. 
In particular, determining the communication complexity for strict subclasses of submodular/supermodular is a natural direction.

\paragraph{Tightness of our Results.} 
We note that our hardness results are tight in the following sense. \cite{dutting2022combinatorial} give a FPTAS for general (monotone) $f$ and $c$ with access to a best-response oracle, assuming all numbers can be represented with $\textsf{poly}(n)$ bits. Our hardness results utilize instances which can also be represented in this way, and thus provide matching lower bounds even for more restrictive classes of $f$ and $c$ and even with arbitrary protocols.

In a model where the representation size is unbounded, we observe that the FPTAS given by \cite{multimulti} for general (monotone) $f$ and additive $c$ can also be used to solve instances with subadditive $c$ (with best-response queries) (see \Cref{sec:FPTAS_subadditive_c}). 
Thus, for this family of instances, our results are tight also in this generalized model.
We leave the question of whether instances with supermodular $c$ (even with additive $f$) are amenable to approximation, even with best-response queries, as an intriguing open problem.

\subsection{Our Techniques}
\label{sec:techniques}
Our results are driven by structural insights regarding settings where one of the set functions is additive and the other is combinatorial (either submodular or supermodular). 
For ease of presentation, we present our techniques for the case of submodular $f$ (and additive $c$) leading to \Cref{thm:main1} (first bullet) and \Cref{thm:main2} (first bullet). The arguments for supermodular $c$ (and additive $f$) are analogous.

We identify two key properties of instances with submodular $f$ and additive $c$, which can be used in a black-box manner to establish 
lower bounds on query and communication complexity.
\begin{enumerate}[label=\textbf{P.\arabic*}]
    \setlength\itemsep{-0.25em}
    \item\label{prop:eqRev} Equal-Revenue: There are exponentially-many distinct contracts, each incentivizing the agent 
    to take a different set of actions, 
    but all are optimal, yielding the same principal's utility.
    \item\label{prop:demand} Sparse Demand: The reward function $f$ satisfies the property that, for any price vector $p$, the number of sets $S$ that approximately maximize $f(S) - \sum_{i \in S} p_i$ is small.
\end{enumerate}
Property \ref{prop:eqRev} enables a standard ``hide a special set'' argument, leading to value query hardness. We establish that together with \ref{prop:demand}, this argument can be extended to imply demand query hardness. 
Additionally, an instance that exhibits these two properties can be augmented with a single additional action such that (i) both $f$ and $c$ are submodular, and (ii) the resulting instance admits a reduction from set disjointness (\Cref{def:disjointness}) --- a well-known benchmark problem in communication complexity.

\paragraph{Establishing~\ref{prop:eqRev}.}
To establish that an instance with this property exists, we consider actions with exponential costs, where $c_i = 2^{i-1}$ is the cost of action $i$. 
We then identify each integer $t \in \{0,\dots,2^n-1\}$, with a set of actions $S_t$, that corresponds to its binary representation. Observe that the (additive) cost of a set is exactly its index.

We show that in this setting, constructing an equal revenue instance boils down to satisfying a specific relation between two consecutive breakpoints $\alpha_t, \alpha_{t+1}$ of the agent's piece-wise linear utility function (see Equation~\eqref{eq:RR_t}).
In this instance, the value associated with a set is $f_t = f(S_t) = \frac{1}{1-\alpha_t}$. 
We show that the sequence of $\alpha_t$'s satisfying the emerging recurrence relation is such that the discrete derivative, $f_{t+1}-f_t$, is a decreasing function of $t$, implying the submodularity of $f$.

\paragraph{Establishing~\ref{prop:demand}.}

Using a delicate counting argument, we show that the same construction used for the equal-revenue property also satisfies sparse demand. 

More precisely, we show that for a small enough $\sigma >0$, for any price vector $p$ the number of sets which maximize the demand $f(S) - \sum_{i \in S} p_i$ up to an additive factor of $\sigma$ is $O(n^2)$.

We denote the collection of sets that approximately maximize the demand by $\demSigP$, and prove that for any $p$, if a set $S_t \in \demSigP$ contains some action $i \in S_t$, then any other set $S_{t'} \in \demSigP$ which is significantly cheaper to incentivize than $S_t$, must also contain $i$.

A sketch of the argument is as follows:
$S_t \in \demSigP$ implies that $p_i < f(S_t)-f(S_t \setminus \{i\}) + \sigma$.
We upper bound this term by utilizing
\ref{prop:eqRev}: As every set is incentivizable, there exists a contract $\alpha(S_t \setminus \{i\})$, for which the agent picks the set $S_t \setminus \{i\}$.
Equivalently, for this contract the marginal reward of action $i$ is smaller than its marginal cost: $\alpha(S_t\setminus \{i\})\cdot  (f(S_t)-f(S_t \setminus \{i\})) < c_i$, and thus $p_i < \frac{c_i}{\alpha(S_t\setminus \{i\})} + \sigma$.
To complete the argument we observe that for any set $S_{t'}$ such that $i \notin S_{t'}$, and $\alpha(S_{t'}\cup \{i\})$ is much smaller than $\alpha(S_t\setminus \{i\})$, it holds that
$p_i < \frac{c_i}{\alpha(S_t\setminus \{i\})} + \sigma
<
\frac{c_i}{\alpha(S_{t'}\cup \{i\})} - \sigma
<
f(S_{t'} \cup \{i\}) - f(S_{t'}) - \sigma
$,
where the last inequality follows from the fact that $\alpha(S_{t'}\cup \{i\})$ incentivizes the set $S_{t'} \cup \{i\}$. This implies that $S_{t'} \notin \demSigP$, as $S_{t'} \cup \{i\}$ is better by at least $\sigma$.

Now, for any price vector $p$, and any action $i$, let $r_i$ be the largest index such that $S_{r_i} \in \demSigP$ and $i \in S_{r_i}$. By definition, any set with higher a index than $r_i$ cannot be approximately optimal and contain $i$. 
On the other hand, by the argument above, any approximately optimal set with index smaller than some threshold $l_i<r_i$ must contain $i$. See \Cref{fig:intervals} for an illustration.
We then utilize this observation to identify every set $S_t$ in the approximate demand with the minimal action $i$ such that $l_i \le t \le r_i$; and show that an action can only be identified with $O(n)$ different sets.
Since there are $n$ actions in total, this yields an upper bound of $O(n^2)$ on the size of the approximate demand.

\subsection{Related Work}

Query and communication complexity has been extensively studied in the mechanism design literature  \cite[e.g.,][]{nisan2006communication,commDobz,settlingCC}. We provide an overview of related work in  algorithmic contract design, with a special focus on combinatorial contracts.

\paragraph{Algorithmic Contract Design and Combinatorial Contracts.}

There are several ways in which contracting problems can be combinatorial; this includes settings in which the agent has many actions \cite{dutting2022combinatorial}, the principal interacts with more than one agent \cite{babaioff2006combinatorial,babaioff2012combinatorialJET,duetting2022multi}, or there are many outcomes \cite{dutting2021complexity}.

The combinatorial-action contracting problem that we study in this work was introduced in~\cite{dutting2022combinatorial}. They show that the optimal contract is tractable with value queries when $f$ is gross-substitutes, and $\textsf{NP}$-hard to compute when $f$ is submodular (both with additive cost). 
They also give a weakly poly-time FPTAS for any (monotone) $f$ and $c$ with access to best-response queries.
More recently, \cite{multimulti} 
give a strongly poly-time FPTAS for any (monotone) $f$ and additive $c$ (with access to value and demand queries). 
Building upon this model, ~\cite{deo2024supermodular} and 
\cite{dutting2024combinatorial} give an algorithm to enumerate all the breakpoints in the agent's piece-wise linear utility given access to a best-response oracle, whose running time is polynomial in the number of breakpoints.
Together with the observation that for supermodular $f$ and submodular $c$ a best-response query can be computed with poly-many value queries, and an argument that shows that there are at most $n$ breakpoints, this gives an efficient algorithm (in the value query model). The tractability frontier was recently extended by \cite{FY25}, who give a poly-time algorithm in the value-query model for the case where $f$ is ultra (a strict super-class of gross-substitutes) and $c$ is additive.
The work of \cite{ezra2023Inapproximability} shows an inapproximability result with value queries when $f$ is submodular and $c$ is additive, assuming $\textsf{P} \neq\textsf{NP}$. 
In a related single-agent model, \cite{contractsSequential} consider an agent that takes actions sequentially, observing the (non-binary) outcome after each step. 
The agent then picks the outcome to present to the principal, and according to which he will get paid.
For independent actions, they give a poly-time algorithm for the optimal linear contract, and for the optimal general contract if the number of outcomes is a constant. For an arbitrary number of outcomes they show that computing the optimal general contract is $\textsf{NP}$-hard.

The work of \cite{contractsInspection} extends the classic hidden-action model by allowing the principal to inspect which action the agent actually took, at some combinatorial cost. The inspection introduces negative incentives: the principal can verify compliance and withhold payment if the agent deviates, but the inspection itself is costly.
For submodular inspection costs, they give an efficient algorithm that computes the optimal contract and randomized inspection scheme using value queries. If inspection costs are XOS (a strict super-class of submodular), they show that exponentially-many value and demand queries may be required.

Another combinatorial setting is when the principal is faced with \textit{multiple agents} and may hire any subset of those. 
The work of~\cite{babaioff2006combinatorial,babaioff2012combinatorialJET} considers a setting where each agent may or may not exert effort, and a Boolean function maps the individual outcomes to success or failure of the project. Two natural examples are the OR and AND Boolean functions, whose corresponding computational complexities are studied in \citep{babaioff2006combinatorial,EmekF12}.
In follow-up work, \cite{babaioff2009free,babaioff2006mixed} study mixed strategies and free-riding in this model.
In contrast with these works, which considered a success probability function encoded as a read-once network, the more recent work of \cite{duetting2022multi} explores reward functions from the complement-free hierarchy, where a function $f$ maps each set of agents that exert effort to a success probability. 
They achieve $O(1)$-approximation to the optimal utility of the principal with poly-many value and demand oracle calls, when rewards are XOS. If the reward function is submodular, they show that poly-many value queries are sufficient for a constant-factor approximation.
The same multi-agent model was also considered by \cite{deo2024supermodular}, who show hardness of approximation when the reward function is supermodular. A pair of recent papers \cite{FeldmanTPS25} and \cite{Aharoni0T25} explores other objectives and budgeted settings, and establishes connections between these settings.

The work of \cite{multimulti} generalizes both the combinatorial-action and the multi-agent model. They show that when $f$, which maps the collection of actions taken by the agents to a success probability, is submodular and $c$ is additive, an approximately optimal contract can be computed in poly-time, with access to value and demand queries. They also strengthen the impossibility of \cite{duetting2022multi} for the multi-agent model with binary actions, by showing that with submodular rewards no algorithm that uses polynomially many value and demand queries can achieve a better than constant-approximation. \cite{cacciamani2024multi} study a related (but non-succinct) 
multi-agent multi-action model, and point to the power of randomized contracts. Another recent work studies multi-agent settings, in which actions choose from a continuum of actions \cite{DasarathaGS25}. Finally, \cite{AlonCCELT25} consider a model in which there are multiple projects, but each agent can work on at most one project, and each project's success probability depends on the agents that work on it.

An alternative multi-agent model was previously considered by \cite{CastiglioniM023}. In their model, the principal can observe each agent's individual outcome and make the contract contingent on this information. The work of \cite{hann2024optimality} considers a similar environment, but is focused on maximizing the principal's reward under budget constraints. 

\paragraph{Further Directions in Algorithmic Contract Design.}
Several studies explore the tradeoffs between simple and optimal contracts. The seminal paper in this direction is \cite{carroll2015robustness}. Another influential paper is \cite{dutting2019simple}. Additional work in this direction includes \cite{dutting2022combinatorial,CarrollW22,DaiToikka22,AnticG23,Kambhampati23,peng2024optimal}. 
Another important line of work, initiated by \cite{ho2014adaptive}, has explored contracting from an online/offline learning perspective. 
This includes work on contract design with bandit feedback  \cite{ZhuBYWJJ23,BacchiocchiC0024,ivanov2024principal}, and work that assumes expert advice \cite{DuettingGuruganeshSchneiderWang23,duetting2025pseudodimensioncontracts}.
The literature on contracts with typed agents  \cite{guruganesh2021contracts, GuruganeshSW023,CastiglioniM021, castiglioni2022designing,alon2021contracts,AlonDLT23,CastiglioniCLXZ25}
studies the combination of moral hazard (hidden action) and screening (hidden types). Recent work has explored the joint design of information structures and contracts \cite{BabichenkoEtAl2024,CastiglioniC25}. Finally, \cite{dutting2024ambiguous,DFR25} study structural and algorithmic properties of ambiguous contracts.

\subsection{Organization}

The paper is organized as follows:
In \Cref{sec:model}, we introduce the model. 
In \Cref{sec:eqRev}, we construct an equal-revenue instance for submodular $f$ and additive $c$, adhering to~Property~\ref{prop:eqRev}. In
\Cref{sec:demandHardness}, we show that this instance also possesses~Property~\ref{prop:demand}, leading to demand query hardness. 
In \Cref{sec:CC_submod_fc}, we present our communication complexity model, augmented with a best-response oracle, and show how to (i) perturb the hard instance for submodular $f$ and additive $c$ rendering $c$ strictly submodular and (ii) add an additional action to prove an exponential lower bound for the case of submodular $f$ and submodular $c$. 

\section{Model and Preliminaries}\label{sec:model}

\paragraph{The Optimal Contract Problem.}
A principal (she) delegates the execution of a task to an agent (he).
The delegated task may have one of two outcomes: success, in which case the principal receives a reward of $1$, and failure, for which she receives $0$. This choice of rewards is without loss of generality. 
The agent may perform any subset of a known set of actions, denoted $[n] = \{1, \ldots, n\}$, where each set of actions incurs him a cost and yields a success probability for the project.

The \emph{cost} of a set of actions is given by the set function $c:2^{[n]} \to \reals_{\ge 0}$, the success probability is given by $f:2^{[n]} \to [0,1]$\footnote{For simplicity, we consider reward functions $f$ such that the image of $f$ is $[A,B]$. One can scale down and shift $f$, 
along with the costs $c$, without affecting our results.}.
Observe that $f(S)$ also specifies the expected reward of the principal when $S$ is performed, and we often refer to $f$ as the (expected) \emph{reward} function. 

The agent's actions cannot be observed by the principal, only the outcome, so the principal may only incentivize the agent using an outcome-contingent payment scheme.
That is, the \emph{contract} specifies the non-negative\footnote{We impose the standard limited liability constraint, meaning that payments can only go from the principal to the agent.} 
transfer from the principal to the agent for each possible outcome, namely two numbers, $t(\text{failure})$ and $t(\text{success})$.

Given a contract, the agent responds with a set of actions $S \subseteq [n]$ that maximize his utility.
We take the perspective of the principal and seek a contract $t$ which maximizes her expected utility, given the agent's best-response.
In this setting, it is without loss of generality that the optimal contract satisfies $t(\text{failure})=0$ \cite{dutting2022combinatorial} and can thus be expressed using a single parameter $t(\text{success}) = \alpha \in [0,1]$ (the payment will never exceed the reward of the principal). 
Hereafter we will refer to $\alpha$ as the contract.

We denote the utility functions of the agent and the principal  from a contract $\alpha \in [0,1]$ and a chosen set of actions $S$ by
$$
u_a(\alpha, S) = \alpha f(S) - c(S) \quad \mbox{ and }
\quad u_p(\alpha, S) = (1-\alpha) f(S),
$$
respectively.
We slightly abuse notation and denote the agent's and principal utilities for contract $\alpha$, while also accounting for the agent's best-response, by
$$
u_a(\alpha) = u_a(\alpha, S_\alpha) = \max_{S\subseteq [n]} (\alpha f(S) - c(S)) \quad \mbox{ and } \quad
u_p(\alpha) =  u_p(\alpha, S_\alpha) = (1-\alpha) f(S_\alpha),
$$
respectively, where $S_\alpha$ denotes the agent's best response for contract $\alpha$. 

As standard in the literature, we assume the agent is breaking ties in favor of the principal. 
So for any two sets $S,S' \subseteq [n]$ maximizing the agent's utility for a given contract $\alpha$, $S_\alpha$ is taken to be the set with the higher value of $f$. 

\begin{definition}[Optimal Contract Problem]\label{def:instance}
    An instance of the optimal contract problem is a triplet $\instance$, specifying the (number of) actions, reward function, and costs.
    The solution to the problem is a contract $\alpha^\star \in [0,1]$ which maximizes the principal's reward, given that the agent responds optimally. Namely,
    \begin{align*}
        \alpha^\star \in \argmax_{\alpha \in [0,1]} (1-\alpha)f(S_\alpha),
    \end{align*}
    where
    $S_\alpha = \argmax_{S \subseteq [n]} (\alpha f(S) - c(S))$, breaking ties in favor of the principal. 
\end{definition}

\paragraph{Structure of the Optimal Contract.}
For a fixed set of actions $S$, the agent's utility, $u_a(\alpha, S)= \alpha f(S) - c(S)$, is an affine function of $\alpha$. As the agent always picks the set $S_\alpha$ that maximizes his utility, the function $u_a(\alpha) = u_a(\alpha, S_\alpha) =  \max_{S \subseteq [n]} u_a(\alpha, S)$ is piece-wise linear, monotonically non-decreasing and convex (see illustration in \Cref{fig:agentPrinUtil}).
The ``breakpoints" of $u_a(\alpha)$ are those values of $\alpha \in [0,1]$ in which the agent's best response changes, i.e. values of $\alpha$ such that $S_\alpha \ne S_{\alpha - \eps}$ for any $\eps>0$. 
It is easy to see (and has been observed by \cite{dutting2022combinatorial}) that the optimal contract must reside at such a breakpoint, since any other contract that incentivizes the same set of actions must be costlier for the principal.
These breakpoints are also referred to as \emph{critical values}. 

\begin{definition}[Incentivizable Set of Actions]\label{def:incentivizableSets}
    A set of actions $S' \subseteq [n]$ is incentivizable if there exists a contract $\alpha \in [0,1]$ such that $S'$ maximizes the agent's utility, after accounting for tie breaking. Namely, there exists $\alpha$ such that
    $$
    S' \in \argmax_{S \subseteq [n]} \left(\alpha f(S) - c(S)\right).
    $$
\end{definition}

\begin{definition}[Critical Value]\label{def:criticalValue}
    $\alpha' \in [0,1]$ is a critical value if it is the minimal payment which incentivizes some set $S$. Namely, there exists an incentivizable set $S'$ such that
    $S' \in \argmax_{S \subseteq [n]} u_a(\alpha',S)$ and for any $\eps >0$, $S' \notin \argmax_{S \subseteq [n]} u_a(\alpha'-\eps,S)$.
\end{definition}

For a given instance $\instance$, we denote by $\IS_{\instance}$ (and omit the subscript when it is clear from context), the collection of incentivizable sets. Often we will index the sets in this collection according to the order of their corresponding critical values. Namely, if the critical values critical values $\alpha < \alpha'$ incentivize the sets $S_t$ and $S_{t'}$ respectively, then $t < t'$.

\begin{observation} [\cite{dutting2022combinatorial}]\label{obs:critVals}
    For any cost function $c:2^{[n]}\to \reals_{\ge 0}$,
    and every reward function $f:2^{[n]} \to [0,1]$, there exist $k \le 2^n -1$ critical values, $0=\alpha_0 < \alpha_1 < \hdots < \alpha_k \le 1$, such that 
    \begin{enumerate}
        \setlength\itemsep{-0.25em}
        \item Each critical value $\alpha_t$ incentivizes a different set $S_t$.
        \item For every contract $x \in [0,1]$, $S_x = S_{\alpha_{t(x)}}$, where $t(x) = \argmax_{t}\{\alpha_t \mid \alpha_t \le x\}$.
        \item For every $0 \le t < t' \le k$, $f(S_t) < f(S_{t'})$, and $c(S_t) < c(S_{t'})$. 
        \item The optimal contract satisfies $\alpha^\star \in \{\alpha_0,\alpha_1,\ldots,\alpha_k\}$. 
    \end{enumerate}
    \end{observation}

Because any critical value $\alpha$ is, by definition, a contract for which the agent's best response changes, the agent is indifferent between (at least) two sets of actions when facing a contract $\alpha$. 
This implies the following observation.
\begin{observation}[\cite{dutting2022combinatorial}]\label{obs:critValStructure}
Let $\alpha$ be a critical value, and
let $S_\alpha$ be the agent's best response at $\alpha$.
The critical value that immediately succeeds $\alpha$ is $\min\{\alpha', 1\}$, where
$$
\alpha' = \min_{S \subseteq [n], f(S) > f(S_\alpha)} \frac{c(S) - c(S_\alpha)}{f(S) - f(S_\alpha)}
$$
\end{observation}

\paragraph{Combinatorial Set Functions and Oracles.}
For any combinatorial set function $v:2^{[n]} \to \reals$ we denote by $v(i \mid S)$ the \emph{marginal value} of action $i \in [n]$ given the set $S \subseteq [n]$, i.e. $v(i \mid S) = v(S \cup \{i\}) - v(S)$. Unless otherwise mentioned, we assume set functions are monotone, i.e., $v(S) \subseteq v(T)$ whenever $S \subseteq T$. For brevity, we will mostly omit this assumption.
In this work we consider submodular and supermodular set functions, which are characterized by decreasing and increasing marginal returns, respectively.
\begin{definition}\label{def:submodularity}
    $v:2^{[n]} \to \reals$ is \emph{submodular} if for every two sets $S,T \subseteq [n]$ such that $S\subseteq T$ and every $i \in [n] \setminus T$, $v(i \mid S)\ge v(i \mid T)$.
    A function $v$ is \emph{supermodular} if and only if $-v$ is submodular.
\end{definition}

Representing a combinatorial set function may require exponential size in the number of actions, $n$. Thus, we assume that the principal accesses $v$ via an oracle. We consider four types of queries:
\begin{itemize}
    \setlength\itemsep{-0.25em}
    \item \emph{Value queries:} The oracle accepts a set of actions $S \subseteq [n]$ and returns $v(S)$.
    \item \emph{Demand queries:} The oracle accepts a price vector $p \in \reals^n_{+}$ and returns the set $S \subseteq [n]$ that maximizes 
    $v(S) - \sum_{i \in S} p_i$.
    Ties are broken in favor of set with the higher value of $v$.
    \item \emph{Supply queries:} The oracle accepts a price vector $p \in \reals^n_{+}$ and returns the set $S \subseteq [n]$, that maximizes 
    $\sum_{i \in S} p_i - v(S)$.
    Ties are broken in favor of the set with the higher value of $v$.
    \item \emph{Best-response queries:} Given two set functions, $f:2^{[n]} \to \reals$ and $c:2^{[n]} \to \reals$, the oracle accepts a contract $\alpha \in [0,1]$ and returns the set of actions $S \subseteq [n]$, that maximizes $\alpha \cdot f(S)-c(S)$. Namely, it returns the set that maximize the agent's utility for a contract $\alpha$ when rewards are given by $f$ and costs are given by $c$. Ties are broken in favor of the set with higher value of $f$.
\end{itemize}
Observe that when $c$ is additive, with costs $\{c_i\}_{i \in [n]}$, a demand query to $f$ with $p_i = c_i / \alpha$ gives the agent's best response for contract $\alpha$. This is because a set $S$ maximizes $f(S) - \sum_{i \in S}c_i/\alpha$ if and only if $\alpha f(S) - \sum_{i \in S} c_i$. 
Thus, when $c$ is additive, the set that maximizes the agent's utility can be computed using a single demand query to $f$. 
The same holds for the additive $f$ case, with respect to a supply query to $c$.

\begin{observation}
    When $c$ is additive, with costs $\{c_i\}_{i \in [n]}$, a demand query to $f$ with prices $p_i = \frac{c_i}{\alpha}$ returns the agent's best response for contract $\alpha$.
    When $f$ is additive, with rewards $\{f_i\}_{i \in [n]}$, a supply query to $c$ with prices $p_i =  \alpha f_i$ returns the agent's best response for contract $\alpha$.
\end{observation}
Following this observation, and drawing from the literature on combinatorial auctions, we sometimes refer to the collection of sets that maximize the agent's utility as the agent's \emph{demand}.

\paragraph{Computational Model.} For ease of presentation, we describe our results for a model where each number takes $O(1)$ space, and basic operations involving numbers take $O(1)$ time.
In \Cref{apx:rounded} we show how to round our construction so that all numbers can be represented with poly-many bits (in the number of actions). So all our results continue to hold in a model in which the representation size of a number $x$ is $O(\log x)$.

\section{An Equal Revenue Instance for Submodular $f$ and Additive $c$}\label{sec:eqRev}
In this section we define a combinatorial equal-revenue instance --- in which all non-empty sets of actions are incentivizable via distinct contracts, and each of them yields the same utility for the principal. 
We construct such an instance when $f$ is submodular and $c$ is additive. An analogous construction for additive $f$ and supermodular $c$ appears in \Cref{sec:eqRev_SupMod_c}. 

\begin{definition}[Equal-Revenue Instance]\label{def:eqRev}
An instance $\instance$, is called \emph{equal-revenue} if 
\begin{enumerate}
    \setlength\itemsep{-0.25em}
    \item Every non-empty set of actions is incentivizable, i.e., $\IS_{\instance} = \IS = \{S_1, \dots, S_{2^n-1}\} =  2^{[n]} \setminus \emptyset$. 
    \item There are $2^n-1$ distinct optimal contracts $\{\alpha_{t}\}_{t=1}^{2^n-1}$, where $\alpha_t$ incentivizes $S_t$ and for all $t$, $u_p(\alpha_t) = (1-\alpha_t)f(S_t) = 1$
\end{enumerate}
\end{definition}

\begin{theorem}\label{thm:eqRev}
For every $n \in \mathbb{N}$, there exists an equal-revenue optimal contract instance, $\instance$, such that $f:2^n \to \reals_+$ is submodular and $c:2^n \to \reals_+$ is additive.
\end{theorem}
We dedicate the rest of the section to construct an equal revenue instance as per \Cref{thm:eqRev}.
In order to specify our construction, we associate every subset of actions $S \subseteq [n]$ with an index $t \in \{0,1,...,2^n-1\}$, which is the decimal value of the $n$-bit characteristic vector of $S$. We order the sets by those indexes $S_0 = \emptyset, S_1=\{1\}, S_2=\{2\}, S_3=\{1,2\},S_4=
\{3\}, S_5=\{1,3\},\ldots,S_{2^n-1} = \{1,...,n\}$. 
We set the cost for each action $i \in [n]$, to be $c_i = 2^{i-1}$. Note that the cost of the set $S_t$ is exactly its index, $t$, as, $c(S_t) = \sum_{i=1}^{n} I[i\in S_t] \cdot 2^{i-1} = t$.
Thus, the difference in costs between two consecutive sets $S_t$ and $S_{t-1}$ is exactly $c(S_t) - c(S_{t-1}) = t - (t-1) = 1$.
Let $\alpha_t$ be the critical value that incentivizes $S_t$. 
An equal revenue instance must satisfy the following requirements:
\begin{itemize}
    \setlength\itemsep{-0.25em}
    \item The revenue from each critical value is exactly 1. Namely, for any $t$, $(1-\alpha_t)f(S_t) = 1$. Let $f_t := f(S_t)$, this is equivalent to 
    $f_t = \frac{1}{1-\alpha_t}$.   
    \item By \Cref{obs:critValStructure}, for any $t \in \{0,\dots,2^n-2\}$, it holds that $\alpha_{t+1} = \frac{c(S_{t+1})-c(S_t)}{f(S_{t+1})-f(S_t) }=\frac{1}{f_{t+1}-f_t}$. 
\end{itemize}
By combining the two requirements, we get
\begin{equation}\label{eq:RR_t}
\alpha_{t+1} = \frac{(1-\alpha_{t+1})(1-\alpha_t)}{\alpha_{t+1}-\alpha_t}.    
\end{equation}
Rearranging gives the following recurrence relation, with the initial condition that $\alpha_0 = 0$:

\begin{equation}\label{eqn:alphaRR}
\alpha_0 = 0; \;\;\;\; \alpha_{t+1} = \alpha_t + \frac{1}{2}\left(\sqrt{4\alpha_t^2 - 8\alpha_t + 5} - 1\right) \qquad t \in \{0,...,2^n-2\}.    
\end{equation}
An illustration of the instance is depicted in \Cref{fig:agentPrinUtil}.

\definecolor{ao}{rgb}{0.0, 0.5, 0.0}
\begin{figure}[t]
    \centering
    \begin{tikzpicture}
    \begin{axis}[
        width=\textwidth,
        height=5cm,
        xlabel={\small{Contract}},
        ylabel={\small{Expected Utility}},
        xtick={0, 0.618, 0.747, 0.807, 0.843, 0.867, 0.884, 
        0.897, 1},
        xticklabels={
        0, $\alpha_1$, $\alpha_2$, $\alpha_3$, 
        {$\alpha_4$}, 
        {}, 
        {$\alpha_6$}, 
        {}, 
        1},
        ytick={1},
        yticklabels={\footnotesize{1}},
        xticklabel style={
        font=\footnotesize, 
        rotate=0, 
        yshift={-2pt}
        },
        xmajorgrids=true,
        grid style=dashed,
        xmin=0, xmax=1, ymin=0, ymax=3,
        legend pos=north west,
        line width=1pt
    ]

    \pgfplotsset{cycle list={{ao},{blue}}}
    \foreach \cost/\xstart/\xend/\rew in {
        0/0/0.618/1, 
        1/0.618/0.747/2.618, 
        2/0.747/0.807/3.956, 
        3/0.807/0.843/5.195, 
        4/0.843/0.867/6.381, 
        5/0.867/0.884/7.534, 
        6/0.884/0.897/8.664, 
        7/0.897/1/9.778
    } {
        \addplot[domain=\xstart:\xend, color=ao] {\rew * x - \cost};
        \addplot[domain=\xstart:\xend, color=blue] {\rew - \rew * x};
    }
    \addplot[domain=0:1, opacity=0.5, dashed, line width=0.25pt] {1};

    \legend{\small{Agent}, \small{Principal}}
    \end{axis}
    \end{tikzpicture}
    \caption{The utilities of the players for $f$ and $\{c_i\}_{i\in [n]}$ as defined above, for the case of $n=3$. Observe that each of the critical values yields the principal the same expected utility of 1.}
    \label{fig:agentPrinUtil}
\end{figure}

Below, we establish the basic properties of our construction. Most importantly
we show that $f$ is submodular (\Cref{prop:fSubmodular}), and that each of the $2^n-1$ values of $\alpha_t$ defined above is a distinct critical value, that gives the principal the same utility (\Cref{prop:CVs}).
Our first lemma  shows that the $\alpha_t$ values
as per \Cref{eqn:alphaRR} are distinct, form an increasing sequence, and are bounded in $[0,1)$.
Our second lemma establishes a
key property of the set of rewards, namely that the discrete derivative $f_{t+1}-f_t$ is negative. 
The proofs of these two lemmas are deferred to \Cref{apx:missingProof}.

\begin{lemma}\label{lem:alphaMon}
    For every $t \in \{0,\dots,2^n-2\}$, $0 \le \alpha_t < \alpha_{t+1} < 1$.
\end{lemma}

\begin{lemma}\label{lem:fMarginalDecrease}
     $f_{t+1} - f_t$ is a decreasing function of $t$, for any $t \ge 0$.
\end{lemma}

Next we show that all $\alpha_t$ are critical values, and yield the same principal utility of $1$.
\begin{proposition}\label{prop:CVs}
    The set of critical values is $0 = \alpha_0 < \alpha_1 < \ldots < \alpha_{2^n-1} < 1$ and for any $t \in \{0, \ldots, 2^n-2\}$, $S_t$ is the agent's best response for contract $\alpha \in [\alpha_t, \alpha_{t+1})$. Additionally, each of the critical values yields the same utility for the principal.
\end{proposition}
\begin{proof}
    Observe that for $\alpha_0=0$, $S_0=\emptyset$ is the agent's best response, as any other set has positive costs.
    Assume that $S_t$ is the agent's best response at $\alpha_t$, for some $t\ge 0$. By \Cref{obs:critVals}, the agent's best response can only change to a set with a higher reward, i.e., $S_{t+k}$ for some $k>1$. 
    By \Cref{obs:critValStructure}, the next critical value satisfies
\begin{eqnarray*}
    \alpha =  
    \min_{k \in \{1,\ldots,2^n-t-1\}} \frac{c_{t+k} - c_{t}}{f_{t+k}-f_{t}} 
    =
    \min_{k \in \{1,\ldots,2^n-t-1\}}
    \frac{k}{f_{t+k}-f_{t}}
    =
    \min_{k \in \{1,\ldots,2^n-t-1\}}
    \frac{k}{\sum_{i=0}^{k-1}f_{t+i+1}-f_{t+i}},
\end{eqnarray*}
    where the first equality follows from the definition of $\{c_i\}_{i \in [n]}$ and the second from telescoping sum.
    By \Cref{lem:fMarginalDecrease}, it holds that for any $k$,
    $\sum_{i=0}^{k-1}f_{t+i+1}-f_{t+i} \le k \cdot (f_{t+1}-f_t)$. Thus, 
    $$
    \alpha \ge \min_{k \in \{1,\ldots,2^n-t-1\}}
    \frac{k}{k(f_{t+1}-f_{t})} = \frac{1}{f_{t+1}-f_{t}},
    $$
    and the next critical value will be $\frac{1}{f_{t+1}-f_{t}} = \frac{c_{t+1}-c_t}{f_{t+1}-f_{t}}$ with the agent best response being $S_{t+1}$.
    Finally, observe that for any $t \ge 0$
    $$
    \frac{1}{f_{t+1}-f_t} 
    = \frac{(1-\alpha_{t+1})(1-\alpha_t)}{\alpha_{t+1}-\alpha_t}
    = \alpha_{t+1},
    $$
    where the first equality follows from the definition of $f_t$ and the second from \Cref{eq:RR_t}.

    The fact that for $t \ge 0$, every contract $\alpha_t$ yields the same utility for the principal follows immediately from the definition of $f_t$,
    $
    u_p(\alpha_t) = (1-\alpha_t)f_t = 1,
    $
    which concludes the proof.
\end{proof}
We conclude by showing that the reward function we described is monotone and submodular.
\begin{proposition}\label{prop:fSubmodular}
    The function $f$ defined above is monotonically non-decreasing and submodular.
\end{proposition}

\begin{proof}
    Monotonicity follows immediately from \Cref{lem:alphaMon}, as $\alpha_t < \alpha_{t+1}$ and $f_t= \frac{1}{1-\alpha_t}$ for any $t$.

    To show submodularity, we prove a slightly stronger claim: for any two sets $S_t, S_{t'} \subseteq [n]$ such that $t<t'$ and any action $i \notin S_t$, it holds that $f(i \mid S_t) \ge f(i\mid S_{t'})$.
    Observe that submodularity follows:
    Since $S_t \subseteq S_{t'}$ implies $t<t'$, for any $i \in [n] \setminus S_{t'}$ we get $f(i \mid S_t) \ge f(i\mid S_{t'})$.
    
    To see why the stronger claim holds, observe that if $i \in S_{t'}$, then by the monotonicity of $f$, $f(i \mid S_{t})\ge 0 =f(i \mid S_{t'})$. Otherwise, $i \notin S_{t'}$, and we have $f( i \mid S_{t'} ) = f_{t'+2^{i-1}} - f_{t'}$.
    In addition, since $i \notin S_{t}$, it also holds that $f( i \mid S_t ) = f_{t+2^{i-1}} - f_t$. 
    For $k\in \{t,t'\}$, we have by telescoping sum that
    $f_{k+2^{i-1}} - f_{k} = \sum_{j=1}^{2^{i-1}} (f_{k+j} - f_{k+j-1}).$
    Therefore, it suffices to show that for every $j \in \{1,\ldots,2^{i-1}\}$, it holds that $f_{t+j} - f_{t+j-1} \ge f_{t'+j} - f_{t'+j-1}$.
    As $t' > t$, this follows immediately from \Cref{lem:fMarginalDecrease}.
\end{proof}

\section{Query Complexity Hardness for Submodular $f$ and Additive $c$}\label{sec:demandHardness}

In this section we give a hardness result for submodular rewards in the demand query model. The analogous result for additive $f$ and supermodular $c$ with supply queries appears in \Cref{sec:demandHardness_SupMod_c}.

\begin{theorem}\label{thm:expDemandQueries}
    When $f$ is submodular and $c$ is additive, any algorithm that computes the optimal contract requires exponentially-many demand queries to $f$.
\end{theorem}

To prove \Cref{thm:expDemandQueries} we proceed as follows. In \Cref{sec:eqToVal}, we show, using a standard ``hide a special set'' argument, that any equal-revenue instance, as per \Cref{def:eqRev}, can be utilized to show hardness in the value query model.
In \Cref{sec:SparsetoDemand}, we continue by extending the hardness results to the stronger demand query model.
Our key insight is a notion of ``sparse demand'', which, together with the equal-revenue property, enables a black box reduction from value-query hardness to demand-query hardness.
In \Cref{sec:eqRevWithSparse}, we use a delicate counting argument to show that the construction per \Cref{sec:eqRev} has sparse demand, and conclude that there are instances in which the principal must use exponentially-many demand queries to find the optimal contract.

\subsection{From Equal Revenue to Value Query Hardness}\label{sec:eqToVal}

\definecolor{ao}{rgb}{0.0, 0.5, 0.0}

\newcommand{\agentcolor}{ao}
\begin{figure}
    \centering
    \begin{tikzpicture}
    \begin{axis} 
    [   width=\textwidth,
        height=5cm,
        xtick={0.47, 0.5, 0.752, 0.81},  
        xticklabels={$\alpha'_{k}$, $\alpha_{k}$, $\alpha_{k+1}$, $\alpha'_{k+1}$}, 
        ymajorgrids=true,
        grid style=dashed,
        line width=1pt,
        ymajorticks=false,
        ylabel={\small{Expected Utility}},
        legend pos=north west
    ]
    
    \addplot[ domain=0.3:0.5, color=\agentcolor,]{0.2*x-0.05};
    \addplot[ domain=0.5:0.75, color=\agentcolor,]{0.4*x-0.15};
    \addplot[ domain=0.75:0.9, color=\agentcolor,]{0.6*x-0.3};

    \addplot[dashed, dash pattern=on 6pt off 3pt, line width=2pt, domain=0.4:0.85, color=\agentcolor]{0.415*x - 0.15};

    \addplot[color=gray, very thin, dashed] coordinates {(0.47, 0) (0.47, 0.3)};
    \addplot[color=gray, very thin, dashed] coordinates {(0.5, 0) (0.5, 0.3)};
    \addplot[color=gray, very thin, dashed] coordinates {(0.752, 0) (0.752, 0.3)};
    \addplot[color=gray, very thin, dashed] coordinates {(0.81, 0) (0.81, 0.3)};

    \legend{\small{Agent}}

    \end{axis}
    \end{tikzpicture}
    \caption{The agent's utility in an equal-revenue instance $\langle n, \feq, \ceq \rangle$ (solid), and the utility for the perturbed instance $\langle n, \feq_k, \ceq \rangle$ (dashed). Note that the perturbed critical values satisfy $\alpha'_k < \alpha_k$ and $\alpha_{k+1} < \alpha'_{k+1}$.}
    \label{fig:perturbation}
\end{figure}

In the following we show how, given an equal revenue instance $\eqRevInst$, one can define a family of optimal contract instances $\familyI$, such that (i) for any $k$, $\feq_k$ is monotone and submodular, 
and (ii) any algorithm that computes the optimal contract for any instance in $\I$ must use exponentially-many value queries.

Every instance $\perturbedInst \in \I$ is almost identical to the original equal revenue instance $\eqRevInst$, except that one of the sets gets a small additive bonus of $\eps$ to its reward. 
We show that when properly picking $\eps$, $\feq_k$ maintains the key properties of $\eqRevInst$. Namely, it is monotone and submodular (see \Cref{prop:feq_kSubmod}), and $\perturbedInst$ has exponentially-many critical values (see \Cref{prop:feq_k_CVs}).
The exponential lower bound on the number of value queries follows from a straight-forward ``hide a special set'' argument, as we show in \Cref{prop:valueQueryBound}.

Given an equal-revenue instance $\eqRevInst$, denote the collection of incentivizable sets with $\IS = \{S_t\}_{t=1}^{2^n-1}$, let $\eps>0$ be such that,
\begin{align}\label{eq:eps_0}
0 < 
\eps < \min \Bigg\{ 
& \min_{\substack{S \subseteq T \subseteq [n], \\ i \in [n]\setminus T}} (\feq(i \mid S) - \feq(i \mid T)), \nonumber \\
& \min_{t \in \{2,\dots,2^n-1\}} \left(\frac{\ceq(S_t) - \ceq(S_{t-1})}{\alpha_{t-1}} - (\feq(S_t) - \feq(S_{t-1}))\right), \\
& \min_{t \in \{1,\dots,2^n-1\}} \left(\feq(S_{t})-\feq(S_{t-1})\right) \Bigg\}.
\nonumber
\end{align}
The three terms in \Cref{eq:eps_0} correspond to the three properties $\perturbedInst$ should satisfy: Submodularity of $\feq_k$, exponential number of critical values, and monotonicity of $\feq_k$, respectively.
In \Cref{cla:eps_0_positive}, deferred to \Cref{apx:proofs_family_equalRev}, we show that the right-hand side of \Cref{eq:eps_0} is positive.
We define $\feq_k$ to be identical to $\feq$ everywhere, except it gives an additive ``bonus" of $\eps$ to the set $S_k$. Namely,
$$
\feq_k(S_t) = 
\begin{cases}
\feq(S_t) + \eps & t=k, \\
\feq(S_t) & t \ne k.
\end{cases}
$$
This completely defines the collection of instances $\familyI$, as they all share the same costs $\ceq$.
The following propositions describe the basic properties of a perturbed instance $\perturbedInst$. The proofs are deferred to \Cref{apx:proofs_family_equalRev}.
\begin{proposition}\label{prop:feq_kSubmod}
    For any $k$, let $\langle n, \feq_k, \ceq \rangle$ be a perturbed instance with $\eps$ satisfying \Cref{eq:eps_0}.
    Then $\feq_k$ is non-negative, monotone and submodular.
\end{proposition}

\begin{proposition}\label{prop:feq_k_CVs}
    For any $k$, let $\langle n, \feq_k, \ceq \rangle$ be a perturbed instance with $\eps$ satisfying \Cref{eq:eps_0}.
    Then $\langle n, \feq_k, \ceq \rangle$, has $2^n-1$ distinct critical values $\{\alpha'_t\}_{t=1}^{2^n-1}$, and $\alpha'_k$ is the unique optimal contract.
\end{proposition}

As shown in the proof of \Cref{prop:feq_k_CVs}, the critical values of the instance $\perturbedInst$ remain identical to the equal revenue instance, expect for $\alpha_k$ and $\alpha_{k+1}$, which are slightly perturbed (see \Cref{fig:perturbation} for an illustration).
For ease of presentation, we will continue to use $\alpha_k$, even when considering instances defined w.r.t. $\feq_k$, although the values $\alpha'_k $ and $\alpha_k$ are not exactly equal.

We show that when faced with an instance from the above collection, the principal must make exponentially-many value queries in expectation in order to find the optimal contract. 

\begin{proposition}\label{prop:valueQueryBound}
Let $I^* = \langle n, \feq^*, \ceq \rangle \in \I = \{\langle n, \feq_k, \ceq \}_{k=1}^{2^n-1}$.
Any algorithm that is only given access to a value oracle for $f^*$, must perform an exponential number of value queries (in expectation) to find the optimal contract.
\end{proposition}
\begin{proof}
    By Yao's principle, it is enough to show that when we draw $I^* =\langle n, \feq^*, \ceq \rangle \in \I$ uniformly at random, any deterministic algorithm requires an exponential number of value queries in expectation over the choice of $I^*$.

    First, observe that every value query $\feq^*(S_t)$ can lead to two different answers. If $\feq^* = \feq_{t}$, then $\feq^*(S_t) = \feq(S_t) + \eps$, otherwise $\feq^*(S_t) = \feq(S_t)$.
    
    Thus, without loss of generality, every deterministic algorithm makes a \emph{fixed} series of value queries and stops whenever the answer is $\feq^*(S_t) = \feq(S_t) + \eps$. If it stops beforehand, and makes less than $2^n-2$ value queries, there exist two sets $S_t,S_{t'}$ which the algorithm did not query and are both consistent with the oracle answers given. Thus, the algorithm cannot distinguish between the case in which $\alpha_t$ is optimal and the case in which $\alpha_{t'}$ is.
    It follows that any deterministic algorithm that makes $q$ value queries only stops for $q$ different choices of $\feq^*$. 
    
    Fix a deterministic algorithm for the optimal contract. Since $\feq^*$ is drawn uniformly at random, with probability $1/2$, the first $2^{n-1}$ queries, $S_{t_1},\ldots,S_{t_{2^{n}-1}}$, satisfy $\feq^*(S_{t_i}) = \feq(S_{t_i})$, and the algorithm does not stop.
    Thus, the expected number of queries of any deterministic algorithm is at least $\frac{1}{2} \cdot 2^{n-1} = 2^{n-2}$.
\end{proof}

\subsection{Extending Hardness to Demand Queries}\label{sec:SparsetoDemand}
In the following section we extend the argument from the previous part to show an impossibility result in the demand query model.
To this end, we introduce the property of sparse demand. We say that a set function $f$ possesses sparse demand if for any sufficiently small $\sigma>0$, and any price vector $p$ there are poly-many sets which maximize the (quasi-linear) utility with respect to $f$ and $p$, up to an additive factor of $\sigma$.
We show that if an equal revenue instance has a reward function with sparse demand, then finding an optimal contract for the family of instances defined in \Cref{sec:eqToVal}, $\familyI$, is hard even with access to demand queries.

\begin{definition}\label{def:sparseDemand}[$\sigma$-Sparse Demand]
For $\sigma >0$, we say that $f: 2^{[n]} \rightarrow \reals_+$ has $\sigma$-sparse demand for every price vector $p$, if $|\demSigP|$ is polynomial in $n$, where $\demSigP$ is the collection of sets of actions that maximize the quasi-linear utility, w.r.t. $f$ and $p$, up to an additive factor of $\sigma$:
$$
\demSigP = \left\{ S\subseteq [n] ~\middle|~ \max_{T\subseteq [n]} \left(f(T)-\sum_{i\in T} p_i\right) - \left(f(S) - \sum_{i\in S} p_i\right) \le \sigma \right\}.
$$
We refer to $\demSigP$ as the $\sigma$-approximate demand. 
Observe that for any $0 < \sigma' \le \sigma$, $\sigma$-sparse demand implies $\sigma'$-sparse demand, as $\demand{\sigma'}{p} \subseteq \demSigP$.
\end{definition}
When the specific value of the sparseness parameter $\sigma>0$ is not required, we omit it and say that $f$ has \textit{sparse demand}.
We now turn to show a demand query lower bound for computing the optimal contract for the family $\familyI$.
\begin{theorem}\label{thm:DemandHardnessForPerturbed}
    Let $\eqRevInst$ be an equal revenue instance such that $\feq$ has sparse demand. Let $\familyI$ be the perturbed family of instances as defined in \Cref{sec:eqToVal}.
    Any algorithm that finds the optimal contract for any instance in $\familyI$, must make  exponentially-many demand queries in expectation.
\end{theorem}

To establish \Cref{thm:DemandHardnessForPerturbed}, we prove that any algorithm, even if augmented with full knowledge of the ``original'' equal revenue instance, $\eqRevInst$, still requires an exponential number of demand queries to find the optimal contract.
The idea is that the sparse demand property allows this augmented algorithm to simulate a demand query to the ``real'' $\feq_k$ with poly-many value queries, implying that, in this case, demand queries are no stronger than value queries (up to a polynomial factor).
The result then follows from fact that any algorithm (even an augmented one), requires exponentially-many value queries, as shown in \Cref{prop:valueQueryBound}.

\begin{proof}[Proof of \Cref{thm:DemandHardnessForPerturbed}]
    Let $\langle n,\feq,\ceq\rangle$ be an equal revenue instance where $\feq$ is submodular and has $\sigma$-sparse demand.
    Recall the family of perturbed instances, $\familyI$, which is defined in \Cref{sec:eqToVal}. Namely, $\feq_k(S_t) = \feq(S_t)$ for any $t \ne k$ and for $t=k$, $\feq_k(S_k) = \feq(S_k) + \eps$, for some $\eps>0$ which satisfies \Cref{eq:eps_0} and also $\eps \le \sigma$.
    
    Consider a specialized algorithm $\A$, which finds the optimal contract for all instances in $\I$ and is augmented with full knowledge of $\feq$, the reward function of the ``original" equal revenue instance.
    Clearly, showing that $\A$ requires exponentially many demand queries in expectation, is sufficient.
    
    Let $\langle n,\feq^*,\ceq\rangle \in \I$, be a perturbed instance drawn uniformly at random.
    First, observe that $\A$ can compute a demand query for $\feq^*$ using poly-many value queries (and perhaps exponentially-many computational steps):
    For any price vector $p$, if $S_p$ should be returned by the demand oracle, then $\feq^*(S_p) - p(S_p) \ge \feq^*(S)- p(S)$ for any $S \subseteq [n]$, which implies that,
    $$
    \feq(S_p) + \eps - p(S_p) \ge \feq^*(S_p) - p(S_p) \ge \feq^*(S)- p(S) \ge \feq(S)- p(S),
    $$
    and as such, $S_p \in \demand{\eps}{p}$. 
    Thus, for any price vector $p$, every set in the demand with respect to $\feq^*$ must also belong to $\demand{\eps}{p}$ with respect to $\feq$.
    
    Because the approximate demand $\demand{\eps}{p}$ does not depend on the identity of $\feq^*$, it can be computed without any value queries -- but perhaps using a super-polynomial number of computational steps.
    Thus, to solve the demand query we can pick the best set in the collection $\demand{\eps}{p}$ using $|\demand{\eps}{p}| \le |\demSigP|$ value queries. As $\feq$ has $\sigma$-sparse demand, the number of queries is polynomial.
    Thus, if $\A$ is able to find the optimal contract with poly-many demand queries, it can also do it with poly-many value queries, due to the above argument. This will contradict \Cref{prop:valueQueryBound}, and the claim follows.
\end{proof}

\subsection{An Equal Revenue Instance with Sparse Demand}\label{sec:eqRevWithSparse}

In what follows we show that the equal revenue instance of \Cref{sec:eqRev} has a reward function that admits sparse demand, which, together with \Cref{thm:DemandHardnessForPerturbed}, implies \Cref{thm:expDemandQueries} --- the main result of this section. Recall the equal revenue instance, $\instance$, presented in \Cref{sec:eqRev}. 
For the purpose of proving that $f$ admits sparse demand we reiterate two important properties of this instance:
\begin{enumerate}
    \setlength\itemsep{-0.25em}
    \item The cost of each action $i \in [n]$ is $c_i = 2^{i-1}$. The cost of a set $S \subseteq [n]$ is $ c(S)= \sum_{i \in S} 2^{i-1}$.
    \item Let $S_t \subseteq [n]$ be the set whose cost is $t$,\footnote{Note that there exists only one such set, the set whose characteristic vector encodes the integer $t$.} and let $\alpha_t$ be the critical value which incentivizes the set $S_t$. 
    Since $\instance$ is equal-revenue, all sets $S_1,\dots,S_{2^n-1}$ have distinct critical values $0<\alpha_1 < \dots < \alpha_{2^n-1}<1$, each yields a utility of 1 to the principal.
\end{enumerate}

\begin{proposition}[Establish Sparse Demand]\label{prop:BoundApproxDemand}
    Let $f$ be as defined in \Cref{sec:eqRev} and let $0 < \sigma < \min_{1 \le l < h} \frac{1}{2}\left(\frac{1}{\alpha_l} - \frac{1}{\alpha_h}\right)$, then
    $f$ has $\sigma$-sparse demand. 
\end{proposition}

In the remainder of the section we prove \Cref{prop:BoundApproxDemand} by showing that for any price vector $p \in \reals_{+}^n$, the size of the approximate demand set for $f$, $\demSigP$, is $O(n^2)$.
We use ambiguity intervals (\Cref{def:amgInt}) to prove this upper bound. 
The key idea is this: 
if action $i$ appears in some approximately optimal set $S$ with respect to $p$, then $i$ must also appear in any approximately optimal set whose cost is significantly smaller than $c(S)$.
Hence, if $S_t$ is the most expensive approximately optimal set containing $i$, the ambiguity interval $[c(S_t) - 2^i, c(S_t)]$ captures the collection of sets in which action $i$ may or may not appear. Indeed, we show that any approximately optimal set with cost below $c(S_i) - 2^i$ must include $i$ (\Cref{lem:RLintervals}), while any set with cost above $c(S_t)$ excludes $i$, by definition. 
This structure allows us to map each approximately optimal set $S$ to the minimal action $i$, such that $c(S)$ lies in $i$’s ambiguity interval. By proving that each ambiguity interval can host at most $O(n)$ sets (\Cref{lem:minAmb}), we obtain the desired $O(n^2)$ bound.

\begin{definition}[Ambiguity Interval]\label{def:amgInt}
    Fix an action $i\in [n]$, a price vector $p \in \reals_{+}^n$ and $\sigma > 0$.
    The ambiguity interval with respect to $i$ and $p$ is $[l_{i,p},r_{i,p}] \subseteq [0,2^n-1]\cap \mathbb{N}$, where 
    \begin{align*}
        r_{i,p} &= \max \{t\in [2^n-1] \mid S_t \in \demSigP, \; i \in S_t\} \quad\text{ and }\quad l_{i,p} = \max\{r_{i,p} - 2^i, 0\}.
    \end{align*}
    If action $i$ does not belong to any set in $\demSigP$, we define $r_{i,p} = l_{i,p} = 0$. For brevity, we omit the subscript $p$ when clear from the context.
\end{definition}

\newcommand{\minAmb}[1]{m(#1)}
\begin{definition}[Minimal Ambiguous Action]\label{def:minAmgAction}
    For any price vector $p \in \reals_{+}^n$ and $\sigma >0$, the minimal ambiguous action of an approximately optimal set $S_t \in \demSigP$ is $\minAmb{S_t} = \min\{i \in [n] \mid t \in [l_{i,p}, r_{i,p}] \}$, with $\minAmb{S_t} = n+1$ if $t$ does not belong to any such interval.
\end{definition}

\begin{figure}
    \centering
    \begin{tikzpicture}
        \coordinate (li) at ($(4cm,0)$) {};
        \draw ($(li)+(0,5pt)$) -- ($(li)-(0,5pt)$);
        \node at ($(li)+(0,3ex)$) {$l_i$};

        \coordinate (ri) at ($(8cm,0)$) {};
        \draw ($(ri)+(0,5pt)$) -- ($(ri)-(0,5pt)$);
        \node at ($(ri)+(0,3ex)$) {$r_i$};
        
        \node at ($(6cm,0)+(0,4ex)$) {ambiguous};
        \node at ($(2cm,0)+(0,4ex)$) {\Large{$\substack{S_t \in \demSigP \\ \text{implies }i \in S_t}$}};
        \node at ($(10cm,0)+(0,4ex)$) {\Large{$\substack{S_t \in \demSigP \\ \text{implies }i \notin S_t}$}};

        \draw[thick,arrows=->] ($(0,0)$) -- ($(12cm,0)$);
    \end{tikzpicture}

    \vspace{.5cm}
    
    \begin{tikzpicture}
        \coordinate (l3) at ($(1cm,0)$) {};
        \draw ($(l3)+(0,5pt)$) -- ($(l3)-(0,5pt)$);
        \node at ($(l3)+(0,3ex)$) {$l_3$};
        
        \coordinate (l1) at ($(2.5cm,0)$) {};
        \draw ($(l1)+(0,5pt)$) -- ($(l1)-(0,5pt)$);
        \node at ($(l1)+(0,3ex)$) {$l_1$};

        \coordinate (r1) at ($(3.5cm,0)$) {};
        \draw ($(r1)+(0,5pt)$) -- ($(r1)-(0,5pt)$);
        \node at ($(r1)+(0,3ex)$) {$r_1$};

        \coordinate (r3) at ($(5cm,0)$) {};
        \draw ($(r3)+(0,5pt)$) -- ($(r3)-(0,5pt)$);
        \node at ($(r3)+(0,3ex)$) {$r_3$};

        \coordinate (l2) at ($(8cm,0)$) {};
        \draw ($(l2)+(0,5pt)$) -- ($(l2)-(0,5pt)$);
        \node at ($(l2)+(0,3ex)$) {$l_2$};

        \coordinate (r2) at ($(9.5cm,0)$) {};
        \draw ($(r2)+(0,5pt)$) -- ($(r2)-(0,5pt)$);
        \node at ($(r2)+(0,3ex)$) {$r_2$};

        \draw[thick,arrows=->] ($(0,0)$) -- ($(12cm,0)$);
        \draw[ultra thick,blue] (l3) -- (r3);
        \draw[ultra thick,magenta] (l1) -- (r1);
        \draw[ultra thick,darkpastelgreen] (l2) -- (r2);

        \draw[thick, arrows=->] ($(4.25,0) - (0,3ex)$) -- ($(4.25,0)$);
        \node at ($(4.25,0) - (0,4.5ex)$) {$S_{t'}$};
    \end{tikzpicture}
    
    \caption{Illustration of ambiguity intervals (top), and minimal ambiguous actions for a fixed price vector $p$ (bottom).
    The minimal ambiguous action of a set $S_t$ for $t \in [l_3, l_1) \cup (r_1,r_3]$, is $i(S_t) = 3$. By \Cref{lem:RLintervals} for $S_{t'}$ such that $t'\in (r_1,r_3]$ as in the figure, $S_{t'}$ does not contain action 1, must contain action 2, and may or may not contain action 3.}
    \label{fig:intervals}
\end{figure}

In the following lemma we establish a key property of the ambiguity intervals (as defined in \Cref{def:amgInt}), which hold for any price vector and any $\sigma$ satisfying the conditions in the lemma: For any $S_t \in \demSigP$ such that $t < l_{i,p}$ it must hold that $i \in S_t$, where as $t > r_{i,p}$ implies $i \notin S_t$. For $t \in [l_{i,p}, r_{i,p}]$ there is ambiguity as to whether $i$ is in $S_t$ or not (hence the name). See \Cref{fig:intervals} for an illustration.
 
\begin{lemma}\label{lem:RLintervals}
    For any price vector $p \in \reals_{+}^n$, any $0 < \sigma < \min_{1 \le l < h}\frac{1}{2}\left(\frac{1}{\alpha_l} - \frac{1}{\alpha_h}\right)$, any action $i \in [n]$ and any $S_t \in \demSigP$, if $t > r_{i,p}$, then $i \notin S_t$. If $t < l_{i,p}$, then $i \in S_t$.
\end{lemma}
\begin{proof}
    The first claim is immediate from the definition of $r_{i,p}$.
    To prove the second claim, fix some $\sigma > 0$ as above and $p \in \reals_{+}^n$. We show that any two sets $S_t,S_{t'} \subseteq [n]$ such that $i \in S_{t'}\setminus S_t$ and $t < t' - 2^i$, cannot be simultaneously be approximately optimal (i.e., $\{S_t,S_{t'}\} \not\subseteq \demSigP$).
    Since by definition $S_{r_{i,p}} \in \demSigP$, replacing $t'$ with $r_{i,p}$, proves the claim for any $t < r_{i,p} - 2^i = l_{i,p}$.

    For any set $S$, let $\alpha_S$ be the critical value that incentivizes $S$.
    Fix $S_t, S_{t'}$ and $i$ as above. 
    As $S_{t'} \setminus \{i\}$ is the agent's best response for $\alpha_{S_{t'}\setminus \{i\}}$, we have
    $$
    \alpha_{S_{t'}\setminus\{i\}}f(i \mid S_{t'}\setminus\{i\}) \le c_i = 2^{i-1}
    \quad
    \Longrightarrow
    \quad
    f(i \mid S_{t'}\setminus\{i\}) \le \frac{2^{i-1}}{\alpha_{S_{t'}\setminus\{i\}}}.
    $$
    Thus, if $p_i > \frac{2^{i-1}}{\alpha_{S_{t'}\setminus\{i\}}}  + \sigma$ we get that $S_{t'} \notin \demSigP$ and the claim holds.
    Similarly, $\alpha_{S_t \cup \{i\}} f(i \mid S_t) \ge c_i = 2^{i-1}$, and if $p_i < \frac{2^{i-1}}{\alpha_{S_t \cup \{i\}}} - \sigma$, then $S_t \notin \demSigP$ and the claim holds.
    Thus, it suffices to show that,
    $$
    \frac{2^{i-1}}{\alpha_{S_t \cup \{i\}}} - \sigma \geq \frac{2^{i-1}}{\alpha_{S_{t'}\setminus\{i\}}}  + \sigma.
    $$
    First, observe that by assumption $t' - 2^{i-1} > t + 2^{i-1}$, so the index of the set $S_{t'} \setminus \{i\}$ is greater than the index of the set $S_t \cup \{i\}$ and by \Cref{lem:alphaMon}, $\alpha_{S_{t'}\setminus \{i\}} > \alpha_{S_t\cup \{i\}}$. 
    By our choice of $\sigma$,
    \begin{align*}
        \sigma 
        &<
        \min_{1 \le l < h}
        \frac{1}{2}\left(\frac{1}{\alpha_l} - \frac{1}{\alpha_h}\right) 
        \le 
        \frac{1}{2}\left(\frac{2^{i-1}}{\alpha_{S_t \cup \{i\}}} - \frac{2^{i-1}}{\alpha_{S_{t'}\setminus\{i\}}}\right),
    \end{align*}
    and the claim follows.
\end{proof}

Next, for any $i^* \in [n+1]$,  we establish an upper bound on the number of sets $S_t \in \demSigP$ for which $i^*$ is the minimal ambiguous action (see \Cref{def:minAmgAction}), i.e., $\minAmb{S_t} = i^*$.

\begin{lemma}\label{lem:minAmb}
Let $\sigma > 0$ satisfy the conditions of \Cref{lem:RLintervals}. For any $i^* \in \{1,...,n,n+1\}$, there can be at most $4i^*$ sets $S_t \in \demSigP$ with $\minAmb{S_t}=i^*$.
\end{lemma}

\begin{proof}
    Fix price vector $p \in \reals^n_+$ and $\sigma >0$ as above.
    First, we prove the statement for the case $i^*=n+1$.
    Any set $S_t \in \demSigP$ for which $\minAmb{S_t}=n+1$ belongs to an interval in which there is no ambiguity about the actions it contains. There are at most $n+1$ such intervals (as there are at most $n$ ambiguity intervals), and, as there's no ambiguity, at most one set per interval.

    For $i^* \in \{1,\dots,n\}$, let $t < t'$ be two indexes such that $S_t, S_{t'} \in \demSigP$ and $\minAmb{S_t} = \minAmb{S_{t'}} = i^*$.
    We give a counting argument for why there cannot be more than $4i^*$ sets of the above form. 
    First, we show that as we increase the index from $t$ to $t'$, the set $S_{t'}$ does not introduce any new actions $i$ such that $i<i^*$. 
    Then, we turn to consider the characteristic vectors of $S_t$ and show that for any suffix (the $i^* + 1$ least significant bits), there may only exists a single prefix (the $n-i^*-1$ most significant bits), such that $t \in [l_{i^*}, r_{i^*}]$.    

    Fix $i < i^*$. Since $\minAmb{S_t}=i^*$ it must be that $t \notin [l_i,r_i]$, and if $t < r_i$, then $t < l_i$. 
    Together with \Cref{lem:RLintervals} we have that $i \in S_t$ if and only if $t < r_i$, and similarly for $t'$.
    Thus, $S_t$ and $S_{t'}$ disagree on action $i$ if only if $l_{i^*} \le t < r_i < t' \le r_{i^*}$, and in this case $i \in S_t$ and $i \notin S_{t'}$ (see illustration in \Cref{fig:smallActionsBound}). 
    Let $S_{t_1},\dots,S_{t_k} \in \demSigP$ be the sets with $\minAmb{S_{t_j}}=i^*$, ordered by index, the above implies that 
    $(S_{t_1} \cap [i^*-1]) \supseteq \dots \supseteq (S_{t_{k}} \cap [i^*-1])$. 
    The number of distinct sets in this chain is at most $|S_{t_1} \cap [i^*-1]| + 1\le i^*$. We conclude that there are at most $4i^*$ different sets of the form $S_t \cap [i^*+1]$.
    
    To complete the argument, fix a set $X \subseteq [i^*+1]$. We show that there can be one set $S_t$ such that $\minAmb{S_t}=i^*$, and $X =S_t \cap [i^*+1]$. To see that, aiming for contradiction, assume there are two different sets $S_t$ and $S_{t'}$ which satisfy the above. Because $S_t$ and $S_{t'}$ only disagree on actions greater than $i^*+1$, the difference $|t-t'|$ is at least $2^{i^*+2-1} = 2^{i^*+1}$, so they cannot both satisfy $t,t' \in [l_{i^*},r_{i^*}]$, as $r_{i^*} - l_{i^*} = 2^{i^*}$. This contradicts the fact that $\minAmb{S_t}=\minAmb{S_{t'}}=i^*$, and the claim immediately follows.
\end{proof}

\begin{figure}
    \centering
    
    \begin{tikzpicture}
        \coordinate (li*) at ($(1cm,0)$) {};
        \draw ($(li*)+(0,5pt)$) -- ($(li*)-(0,5pt)$);
        \node at ($(li*)+(0,3ex)$) {$l_{i^*}$};
        
        \coordinate (li) at ($(3cm,0)$) {};
        \draw ($(li)+(0,5pt)$) -- ($(li)-(0,5pt)$);
        \node at ($(li)+(0,3ex)$) {$l_i$};

        \coordinate (ri) at ($(5cm,0)$) {};
        \draw ($(ri)+(0,5pt)$) -- ($(ri)-(0,5pt)$);
        \node at ($(ri)+(0,3ex)$) {$r_i$};

        \coordinate (ri*) at ($(9cm,0)$) {};
        \draw ($(ri*)+(0,5pt)$) -- ($(ri*)-(0,5pt)$);
        \node at ($(ri*)+(0,3ex)$) {$r_{i^*}$};
      
        \draw[thick,arrows=->] ($(0,0)$) -- ($(12cm,0)$);
        \draw[ultra thick,blue] (li*) -- (ri*);
        \draw[ultra thick,magenta] (li) -- (ri);

        \draw[thick, arrows=->] ($(2,0) - (0,3ex)$) -- ($(2,0)$);
        \node at ($(2,0) - (0,6ex)$) {\Large{$\substack{S_t \\ i \in S_t}$}};
        \draw[thick, arrows=->] ($(7,0) - (0,3ex)$) -- ($(7,0)$);
        \node at ($(7,0) - (0,6ex)$) {\Large{$\substack{S_{t'} \\ i \notin S_{t'}}$}};
    \end{tikzpicture}
    
    \caption{Illustration of the proof of \Cref{lem:minAmb}. Let $S_t$ and $S_{t'}$ have the same minimal ambiguous action $i^*$ and $t'>t$. If both $t$ and $t'$ are on the same side of the ambiguity interval of $i < i^*$, then $S_t$ and $S_{t'}$ agree on action $i$. Otherwise, as in the figure, it must be that $i \in S_t$ and $i \notin S_t$. }
    \label{fig:smallActionsBound}
\end{figure}

We are now ready to prove \Cref{prop:BoundApproxDemand}.

\begin{proof}[Proof of \Cref{prop:BoundApproxDemand}]
    \Cref{lem:minAmb} immediately implies the claim, as 
    \[
    |\demSigP| = \sum_{i^* \in \{1,\dots,n,n+1\}} |\{S \in \demSigP \mid \minAmb{S} = i^*\}| \le \frac{4(n+1)(n+2)}{2} = O(n^2) \qedhere
    \]
\end{proof}

\section{Communication Complexity Hardness for Submodular $f$ and $c$}\label{sec:CC_submod_fc}

In this section, we exhibit another use of combinatorial equal-revenue instances with sparse demand.
Namely, we demonstrate how they can be utilized to establish communication complexity lower bounds, in a contracting model where the $f$ is held by one party, the $c$ is held by another party, and both parties have access to a best-response oracle.
Specifically, we show how to modify an equal revenue construction with submodular $f$ and additive $c$, to establish an exponential lower bound on the communication required to compute the optimal contract when $f$ and $c$ are submodular. We extend these results to other classes in \Cref{sec:cc_lower_bound_supmod_f_c} and \Cref{sec:cc_lower_bound_submod_f_supmod_c}.

\subsection{Model and Preliminaries}\label{sec:CC_model}
\paragraph{The Optimal Contract Problem for Two Parties.}

We consider the problem of computing the optimal contract when the input is divided between two parties.
Recall that an optimal contract problem, as per \Cref{def:instance}, is a 3-tuple $\instance$, where for each subset of actions $f:2^{[n]} \to \reals_{+}$ specifies its reward, and $c:2^{[n]} \to \reals_{+}$ specifies its cost.
In the model presented below, two parties hold the two combinatorial functions which define the instance. Namely, Alice get as input a description $f$ and Bob gets a description of $c$. Additionally, each of the parties has access to a best-response oracle (see \Cref{sec:model}), which given a contract $\alpha$ returns the set that maximizes the agent's utility, $\alpha f(S) - c(S)$.
We are interested in measuring the minimal number of bits Alice and Bob need to communicate in order to compute the optimal contract , while only allowing poly-many oracles calls.
Observe that whenever $c$ or $f$ can be represented with $poly(n)$ bits (e.g., an additive function), 
the problem reduces to the query-complexity model discussed in previous sections, as one party can send the function in its entirety to the other.
In this section we will focus on the case where $f$ and $c$ submodular and are thus not necessarily succinctly representable.

\paragraph{{Vanilla} Communication Complexity and Set Disjointness.}
We describe the two-party communication complexity model, introduced by \cite{yao1979some} and the set disjointness problem.
Let $X$ and $Y$ be two finite sets. Alice and Bob wish to compute a function $\varphi: X \times Y \to Z$, where Alice receives some $x \in X$ as input, and Bob receives $y \in Y$.
For this purpose, Alice and Bob may communicate bits back and forth according to a protocol, which maps the input of a party and the history of messages to a new message. 
Moreover, Alice and Bob share random bits, which they can use to execute a randomized protocol.
A protocol \textit{computes} $\varphi$ with error $\delta$ if for every pair $(x,y)$, the probability that the output is $\varphi(x,y)$ is at least $1-\delta$.
The cost of a protocol is the number of bits exchanged on the worst-case input $(x,y)$, and the communication complexity of $\varphi$ is the cost of the cheapest protocol that computes $\varphi$ with error\footnote{This is without loss of generality, as for any constants $\eps,\eps' \in (0,1/2)$, the error can be reduced from $\eps$ to $\eps'$ by repeating the protocol a constant number of times and output the majority.} $1/3$.

One of the most studied problems in communication complexity theory is the set disjointness problem (often denoted $\DISJ{k}$), see survey \cite{sherstov2014communication}.
In this problem each party gets a subset of $\{1,\dots,k\}$, and the parties need to determine if they have no common elements. Namely,
\begin{definition}[The $\DISJ{k}$ Problem]\label{def:disjointness}
    The problem of set disjointness with parameter $k$ is computing for every $A \subseteq [k]$ and $B \subseteq [k]$ the function
    $$
    \DISJ{k}(A,B) = \begin{cases}
        1 & \text{if } A \cap B = \emptyset, \\
        0 & \text{otherwise.} 
    \end{cases}
    $$
\end{definition}
\begin{fact}[\cite{kalyanasundaram1992probabilistic, Razborov92}]
    Any protocol that solves $\DISJ{k}$ requires communicating $\Omega(k)$ bits.
\end{fact}

\paragraph{Adding Best-Response Oracle.}

In the context of an optimal contract problem $\instance$, we augment the communication complexity model above by allowing each of the parties (Alice who gets $f$ and Bob who gets $c$) to make poly-many queries to a best response oracle.
To establish our results, we generalize the concept of Sparse Demand (see \Cref{def:sparseDemand}), to settings when both $f$ and $c$ are combinatorial.

\begin{definition}\label{def:sparseBestREsponse}[$\sigma$-Sparse Best-Response]
For $\sigma >0$, instance $\instanceI$ has $\sigma$-sparse best-response if for and every contract $\alpha \in [0,1]$, $|\BrSigAlpha|$ is polynomial in $n$, where $\BrSigAlpha$ is the collection of sets of actions that maximize the agent's utility for contract $\alpha$, up to an additive factor of $\sigma$:
$$
\BrSigAlpha = \left\{ S\subseteq [n] ~\middle|~ \max_{T\subseteq [n]} (\alpha f(T) - c(T)) - \left(\alpha f(S) - c(S)\right) \le \sigma \right\}.
$$
We refer to $\BrSigAlpha$ as the $\sigma$-approximate best-response to $\alpha$.
When the instance $I$ is clear from context, we sometimes omit it and use $\demand{\sigma}{\alpha}$.
Observe that for any $0 < \sigma' \le \sigma$, $\sigma$-sparse best-response implies $\sigma'$-sparse best-response, as $\BRI{\sigma'}{\alpha} \subseteq \BrSigAlpha$.
\end{definition}
When the specific value of the sparseness parameter $\sigma>0$ is not required, we omit it and say that the instance $\instanceI$ has \textit{sparse best-response}.

\subsection{CC Lower Bound for Submodular $f$ and $c$}\label{subsec:cc_lower_bound_submod_fc}

We now prove the main result of this section, an exponential lower bound on the communication required to compute the optimal contract when both $f$ and $c$ are submodular, even if poly-many best-response queries are allowed.
\begin{theorem}[CC lower bound, submodular $f$ and $c$]\label{thm:cc_submod_f_c}
    When $f$ and $c$ are submodular, any protocol which computes the optimal contract and makes $poly(n)$ best-response queries, requires $\Omega(2^n/\sqrt{n})$ bits of communication.
\end{theorem}

The sketch of the proof is as follows:
We start from an equal-revenue instance with sparse demand, for example the one described in \Cref{sec:eqRev}, and add small perturbations to the additive cost function $c$, such that the new function, $\tc$, is strictly submodular.
The perturbations are small enough so that the new instance $\tilde{I} = \inst{n}{f}{\tc}$ possesses the following properties:
\begin{enumerate}
    \item There are $2^n-1$ critical values, where 
    for any $t \in [2^n-1]$, the set $S_t$ that encodes the integer $t$ is incentivized by the contract $\talpha_t$.
    \item The utility the principal gets from each of these contracts is approximately $1$, i.e., this is an ``almost'' equal revenue instance.
    \item The new instance $\tI$ has sparse best-response.
\end{enumerate}

Then, we introduce a new action, numbered $n+1$. Its marginal rewards (respectively, costs) are defined using a binary vector $x_f$ (resp., $x_c$) of size $\binom{n}{n/2}$, indicating ``special'' sets of size $n/2$ among actions $\{1,\dots,n\}$. This defines a family of instances, all based on $\InstTildeI$ and differ only in the values of $x_f$ and $x_c$.
We show that if there exists a special set with respect to both costs and rewards, i.e. $x_f \cap x_c \ne \emptyset$, then one of these sets, $S \subseteq [n]$, is such that $S \cup \{n+1\}$ is optimal. Conversely, if $x_f \cap x_c = \emptyset$, a set of the form $S \cup \{n+1\}$ cannot be incentivized by the principal.

We augment Alice and Bob with knowledge of the ``base instance'' $\tI$, and show, using the sparse best-response property, that a best-response oracle can be simulated with polynomial communication.
Thus, any efficient protocol executed by Alice and Bob and uses (poly-many) oracle calls can be executed efficiently by augmented Alice and Bob, with no oracle calls.
Nonetheless, augmented Alice and Bob still need to know whether $x_f \cap x_c \ne \emptyset$ in order to find the optimal contract, i.e., solve $\DISJ{\binom{n}{n/2}}$. This requires exponential communication, implying \Cref{thm:cc_submod_f_c}.

Our construction relies on an equal-revenue instance $\instanceI$ with sparse demand, as the one presented in \Cref{sec:eqRev}. The observation below follows from \Cref{obs:critValStructure} and \Cref{lem:f_strictly_submodolar}:
\begin{observation}\label{obs:equalrev_subModf_prop}
    Any equal revenue instance $\instanceI$ with submodular $f$ an additive $c$ must satisfy:
    (i) $c$ and $f$ are strictly monotone, and (ii) $f$ is strictly submodular.
\end{observation}

Denote the set of incentivizable sets with $\IS = \{S_1,\dots,S_{2^n-1}\}$. Recall that $c(S_{t-1})<c(S_t)$ and $f(S_{t-1})<f(S_t)$, for any $t$.
We denote $S_0 = \emptyset$ and assume $c$ and $f$ are normalized such that $c(\emptyset)=f(\emptyset)=0$.
The perturbations introduced to the costs are defined with respect to 
$\phi_f = \min_{t} \{f(S_{t+1})-f(S_{t})\} > 0$ and
$\phi_c = \min_{t} \{c(S_{t+1})-c(S_{t})\} > 0$, the discrete derivatives of $f$ and $c$ with respect to $t$. 
By the above, these are well defined for an equal-revenue instance.

\subsubsection{Making $c$ strictly submodular}
The perturbed cost function is defined as follows, $\tc(S)=c(S) - \delta \cdot |S|^2$, where $\delta$ satisfies
\begin{align}\label{eq:delta_for_tc}
0 < 
\delta < \min \Bigg\{ 
& \frac{\phi_c}{n^2}, \alpha_1 \cdot \frac{f_1}{n^2}, \frac{1}{n^2}(1-\alpha_{2^n-1})(f(S_{2^n-1})-f(S_{2^n-2})), \\
& \min_{t \in \{1,\dots,2^n-2\}} (\alpha_{t+1}-\alpha_t)\cdot \frac{1}{n^2} \cdot \left(\frac{(f(S_t)-f(S_{t-1}))(f(S_{t+1})-f(S_t))}{(f(S_{t+1})-f(S_{t-1}))}\right) \nonumber \bigg\}
\end{align}

The next two propositions establish the basic properties of the cost function $\tc$. Namely, that for a proper choice of $\delta$, $\tc$ is non-negative, monotone, submodular, and additionally, the instance $\InstTildeI$ has $2^n-1$ critical values.
We delegate their proofs to \Cref{sec:Omitted_CC_submod_fc}.

\begin{proposition}\label{prop:ctilde_basix_props}
For $0 < \delta < \frac{\phi_c}{n^2}$, $\tc$ is non-negative, strictly monotone and strictly submodular.
\end{proposition}

\begin{proposition}\label{prop:exp_many_cvs_maintained}
Let $\delta > 0$ be such that \Cref{eq:delta_for_tc} holds, then, $\inst{n}{f}{\tc}$ has $2^n-1$ critical values.
\end{proposition}

\subsubsection{Sparse Best-Response}
\begin{proposition}\label{prop:sparse_demand_submod_f_c}
    If the instance $\instanceI$ has $\sigma$-sparse demand for some $\sigma > 0$, then for any $0 < \delta < \frac{\sigma}{2n^2}$, the perturbed instance $\InstTildeI$ has $(\sigma/2)$-sparse best-response.
\end{proposition}
\begin{proof}
\newcommand{\sa}{S_\alpha}
\newcommand{\tsa}{\tilde{S}_\alpha}
Fix a contract $\alpha$ and some $0 < \delta < \frac{\sigma}{2n^2}$. We show that 
$\BR{\sigma/2}{\alpha}{\tI} \subseteq \BRI{\sigma/2+n^2\delta}{\alpha} \subseteq \BRI{\sigma}{\alpha}$, which implies the claim.
Let $\sa$ be the set of actions which maximizes $\alpha f(S) - c(S)$ and let $\tsa$ be the set of which maximizes $\alpha f(S) - \tc(S)$.
Observe that for any $S \in \BR{\sigma/2}{\alpha}{\tI}$
it holds that,
\begin{align*}
    \alpha f(S) - \tc(S) 
    &\ge \alpha f(\tsa) - \tc(\tsa) - \sigma/2 \\
    &\ge \alpha f(\sa) - \tc(\sa) - \sigma/2 \\
    &\ge \alpha f(\sa) - c(\sa) - \sigma/2,
\end{align*}
where the last inequality follows from the fact that $\tc(\sa) = c(\sa) - \delta|\sa|^2$.
On the other hand, 
$$\alpha f(S) - \tc(S) = \alpha f(S) - c(S) + \delta|S|^2 \le \alpha f(S) - c(S) + \delta n^2,$$ 
we conclude that 
\[
\alpha f(S) - c(S) \ge \alpha f(\sa) - c(\sa) - (\sigma/2 + \delta n^2),
\]
which implies that $S \in \BRI{\sigma/2+n^2\delta}{\alpha}$.
By our choice of $\delta$, it holds that $\sigma/2+n^2\delta \le \sigma$, and thus $\BRI{\sigma/2+n^2\delta}{\alpha} \subseteq \BRI{\sigma}{\alpha}$, which concludes the proof.
\end{proof}

\subsubsection{Adding the ($n+1$)th action}
In order to introduce the marginal rewards and costs of the ($n+1$)th action, we consider a small quantity $z$, with respect to which marginals will be defined.
We pick $z$ so that the new instance with $n+1$ actions has submodular costs and rewards, has sparse best-response, and that the optimal contract problem can be reduced to $\DISJ{\binom{n}{n/2}}$.
We define $z = \min \{\phi_{\tc},\psi_{\tc},\phi_f,\psi_f, \zeta, \sigma\}$, as the minimum of these positive quantities:
\begin{itemize}
    \setlength\itemsep{-0.25em}
    \item $\phi_{\tc} = \min_{\tc} \tc(S_{t+1})-\tc(S_{t}) > 0$, note that for any $S \subsetneq T \subseteq [n]$, $\tc(T)-\tc(S) \ge \phi_{\tc}$),
    \item $\phi_f = \min_{t} f(S_{t+1})-f(S_{t}) > 0$, note that for any $S \subsetneq T \subseteq [n]$, $f(T)-f(S) \ge \phi_{f}$),
    \item $\psi_{\tc} = \min_{S\subseteq T, i \notin T} \tc(i \mid S)-\tc(i \mid T) > 0$,
    \item $\psi_f = \min_{S\subseteq T, i \notin T} f(i \mid S)-f(i \mid T) > 0$,
    \item $\zeta = \frac{\delta \cdot (1-\alpha_{2^n-1}) \phi_f}{16\cdot n^2 \cdot f_{2^n-1}}> 0$.
    \item $\sigma>0$, the sparseness parameter for which $\InstTildeI$ has sparse demand, see \Cref{prop:sparse_demand_submod_f_c}.
\end{itemize}

We first show that the perturbed instance (with only $n$ actions), $\InstTildeI$, is ``almost'' equal revenue, that is, we bound the principal's utility from any critical value $\talpha_1, \dots \talpha_{2^n-1}$.
\begin{proposition}\label{prop:eps_for_revenue_bound}
For any $0 < z \le \zeta = \frac{\delta \cdot (1-\alpha_{2^n-1}) \phi_f}{16\cdot n^2 \cdot f_{2^n-1}}$, for any $t \in [2^n-1]$, 
$$
(1-\talpha_t)f(S_t) \in \left[1 - \frac{z(1-\alpha_{2^n-1})}{16}, 1 + \frac{z(1-\alpha_{2^n-1})}{16} \right]
$$
\end{proposition}
\begin{proof}
    Fix $t$. Denote $f_t = f(S_t)$. By the definition of $\tc$ we have,
    \begin{align*}
        (1-\talpha_t)f(S_t)
        &\in
        \left[(1-\alpha_t)f_t - \frac{\delta n^2}{f_t-f_{t-1}}f_t, (1-\alpha_t)f_t + \frac{\delta n^2}{f_t-f_{t-1}}f_t \right] \\
        &=
        \left[1 - \frac{\delta n^2}{f_t-f_{t-1}}f_t, 1 + \frac{\delta n^2}{f_t-f_{t-1}}f_t \right] \\
        &\subseteq
        \left[1 - \frac{\delta n^2}{\phi_f}f_{2^n-1}, 1 + \frac{\delta n^2}{\phi_f}f_{2^n-1} \right] \\
        &\subseteq
        \left[1 - \frac{z(1-\alpha_{2^n-1})}{16}, 1 + \frac{z(1-\alpha_{2^n-1})}{16} \right],
    \end{align*}
    where the last transition follows plugging in the definition of $\zeta$.
\end{proof}
We are ready to define the new instance $\InstHatI$, where 
\begin{align*}
\hf(S)&=f(S\setminus\{n+1\})+\indicator{n+1 \in S}\cdot\hf(n+1 \mid S), \text{ and } \\
\hc(S)&=\tc(S\setminus\{n+1\})+\indicator{n+1 \in S}\cdot\hc(n+1 \mid S),
\end{align*}

where the marginal rewards and costs of action $n+1$ are defined with respect to two binary vectors $x_f,x_c \in \{0,1\}^{\binom{n}{n/2}}$, as follows:
$$
\hf(n+1 \mid S_t) = 
\begin{cases}
    z/4 & |S_t| < \frac{n}{2} \\
    z/4 & |S_t| = \frac{n}{2} \land S_t \in x_f \\
    0 & |S_t| = \frac{n}{2} \land S_t \notin x_f \\
    0 & |S_t| > \frac{n}{2}
\end{cases}
\qquad 
\hc(n+1 \mid S_t) = 
\begin{cases}
    z/2 & |S_t| < \frac{n}{2} \\
    z/2 & |S_t| = \frac{n}{2} \land S_t \notin x_c \\
    (\talpha_t\cdot z)/4 & |S_t| = \frac{n}{2} \land S_t \in x_c \\
    (\talpha_1\cdot z)/8 & |S_t| > \frac{n}{2}
\end{cases}.
$$

\begin{proposition}\label{prop:monotone_submod_hf_hc}
    Both $\hf$ and $\hc$ are monotone and submodular.
\end{proposition}
The proof of \Cref{prop:monotone_submod_hf_hc} is delegated to \Cref{sec:Omitted_CC_submod_fc}.

\begin{observation}\label{obs:n+1Incentivizable}   
    If a set of the form $S_t \cup \{n+1\}$, for some $S_t \subseteq [n]$, can be incentivized by the principal, then $|S_t| = n/2$ and $S_t \in x_c \cap x_f$.
\end{observation}
\begin{proof}
Given some set $S_t$, we consider $\hf(n+1 \mid S_t) - \hc(n+1 \mid S_t)$, which upper bounds the marginal contribution of $n+1$ to the agent's utility for any contract $\alpha \in [0,1]$. 
We show that unless $S_t \in x_c \cap x_f$, this quantity is negative, which implies the claim. 

    If $|S_t|<n/2$, the marginal contribution of $n+1$ is $z/4 - z/2 < 0$.
    
    If $|S_t| > \frac{n}{2}$, the marginal contribution of $n+1$ is $0 - \frac{\talpha_1\cdot z}{8} < 0$.
    
    If $|S_t|=n/2$ and $S_t \notin x_f$, the marginal contribution of $n+1$ is at most $0 - \frac{\talpha_t\cdot z}{4} < 0$.
    
    If $|S_t|=n/2$ and $S_t \notin x_c$, the marginal contribution of $n+1$ is at most $z/4 - z/2 < 0$.

    This concludes the proof.
\end{proof}

\begin{proposition}\label{prop:n+1_is_optimal}
    Let $S_t$ be the set with the minimal index such that $S_t \in x_c \cap x_f$. 
    The contract $\halpha_t = \frac{\hc(S_t)-\hc(S_{t-1})}{\hf(S_t)-\hf(S_{t-1})}$ incentivizes the set $S_t \cup \{n+1\}$ and
    the principal's revenue from $\halpha_t$, $u_p(\halpha_t,S_t\cup\{n+1\})$, exceeds the revenue from any incentivizable set of action of the form $S_{t'} \subseteq [n]$.
\end{proposition}
\begin{proof}
To see why $\halpha_t$ incentivizes $S_t \cup \{n+1\}$, note that the three linear function in $\alpha$, $u_a(\alpha, S_{t-1})$, $u_a(\alpha, S_t)$ and $u_a(\alpha, S_t \cup \{n+1\})$, intersect at the same point $\halpha_t$.
By definition, the intersection between $u_a(\alpha, S_t)$ and $u_a(\alpha, S_{t-1})$ is
$$
\halpha_t = 
\frac{\hc(S_t) - \hc(S_{t-1})}{\hf(S_t) - \hf(S_{t-1})}
= \frac{\tc(S_t) - \tc(S_{t-1})}{f(S_t) - f(S_{t-1})}
= \talpha_t
$$
Similarly for $u_a(\halpha, S_t \cup \{n+1\})$ and $u_a(\halpha, S_t)$,
$$
\halpha_t = \frac{\hc(S_t \cup \{n+1\}) - \hc(S_t)}{\hf(S_t  \cup \{n+1\}) - \hf(S_t)} = \frac{\talpha_t \cdot z/4}{z/4} = \talpha_t.
$$
As the agent's utility is identical to the one in the instance $\InstTildeI$ for any contract $\alpha < \halpha_t$, and
since $\hf(S_t \cup \{n+1\}) > \hf(S_t) > \hf(S_{t-1})$, the agent's best-response for $\halpha_t$ is $S_t \cup \{n+1\}$.

We now bound the principal's utility from the contract $\halpha_t$.
First, observe that since $|S_t|=\frac{n}{2}$ and $|S_{2^n-1}|=n$, it must be that $\talpha_t \le \talpha_{2^n-2} < \talpha_{2^n-1} \le \alpha_{2^n-1}$, where the last inequality follows from the choice of $\delta$, as can be seen in the proof of \Cref{prop:exp_many_cvs_maintained}.
Now we bound the principal's utility from $\halpha_t$:
\begin{align*}
    (1-\halpha_t)\hf(S_t \cup \{n+1\}) 
    &=
    (1-\halpha_t)f(S_t) + (1-\halpha_t)\frac{z}{4} \\
    &>
    (1-\talpha_t)f(S_t) + (1-\alpha_{2^n-1})\frac{z}{4} \\
    &\ge
    1 - \frac{z(1-\alpha_{2^n-1})}{16} + (1-\alpha_{2^n-1})\frac{z}{4} && (\Cref{prop:eps_for_revenue_bound})\\
    &>
    1 + \frac{z(1-\alpha_{2^n-1})}{16}
\end{align*}
Fix any $t' \in \{1,\dots,2^n-1\}$. In order to incentivize a set $S_{t'} \subseteq [n]$ with contract $\alpha$, it must yield more utility to the agent than $S_{t'-1}$ for this contract, and thus it must be that
$\alpha \ge 
\frac{\hc(S_{t'})-\hc(S_{t'-1})}{\hf(S_{t'})-\hf(S_{t'-1})}
=
\frac{\tc(S_{t'})-\tc(S_{t'-1})}{f(S_{t'})-f(S_{t'-1})} = \talpha_{t'}$.
So the principal's utility from incentivizing $S_{t'}$ is at most
\begin{align*}
    (1-\talpha_{t'})\hf(S_{t'}) = 
    (1-\talpha_{t'})f(S_{t'}) \le 1 + \frac{z(1-\alpha_{2^n-1})}{16} < (1-\halpha_t)f(S_t \cup \{n+1\}),
\end{align*}
where the first inequality follows from \Cref{prop:eps_for_revenue_bound}. This concludes the proof.
\end{proof}
\begin{corollary}\label{cor:opt_set_if_non_empty_intersection}
    If $x_f \cap x_c$ is non-empty, then the optimal contract incentivizes some set $S_t \cup \{n+1\}$, where $S_t \subseteq [n]$ and $S_t \in x_f \cap x_c$. 
\end{corollary}

Below we show that if $S_\alpha \subseteq [n+1]$ is the agent's best response for some contract $\alpha$ in the instance $\InstHatI$, then $S_\alpha \setminus \{n+1\}$ is the best response for $\alpha$ in $\InstTildeI$, up to an additive factor of $\sigma/2$, where $\sigma$ is the sparseness parameter for the equal-revenue instance $\instanceI$.

\begin{proposition}\label{prop:BR_in_approxDemand_f_c_submod}
    Let $S_\alpha$ be the agent's best response for contract $\alpha$ in the instance $\InstHatI$, then 
    $S_\alpha \setminus\{n+1\} \in \BR{\sigma/2}{\alpha}{\tI}$.
\end{proposition}
\begin{proof}
Fix a contract $\alpha \in [0,1]$.
    If $S_\alpha \subseteq [n]$, the claim holds trivially, as for sets which do not contain $n+1$, $\hf$ is identical to $f$ and $\hc$ is identical to $\tc$.

     If $S_\alpha = S_t \cup \{n+1\}$ for some $S_t \subseteq [n]$, then by \Cref{obs:n+1Incentivizable}, $S_t \in x_f \cap x_c$.
     Thus, for any $S \subseteq [n]$,
     \begin{align*}
     \alpha f(S) - \tc(S) 
     &=
     \alpha \hf(S) - \hc(S) \\
     &\le
     \alpha \hf(S_\alpha) - \hc(S_\alpha) \\
     &=
     \alpha f(S_t) - \tc(S_t) + \alpha \hf(n+1 \mid S_t) - \hc(n+1 \mid S_t)\\
     &=
     \alpha f(S_t) - \tc(S_t) + \frac{\alpha z}{4} && (c \text{ is non-negative}) \\
     &\le
     \alpha f(S_t) - \tc(S_t) + z/4 \\
     &\le 
     \alpha f(S_t) - \tc(S_t) + \sigma/4,
     \end{align*}
     where the last inequality follows from our choice of $z$.
\end{proof}

\begin{corollary}\label{cor:simulate_BR_f_c_submod}
    Consider an instance $\InstHatI$ defined using $x_f, x_c \in \{0,1\}^{\binom{n}{n/2}}$.
    If Alice and Bob are given full knowledge of $f$, $\tc$ and $\sigma$, they can compute a best response query for $\InstHatI$ using $poly(n)$ communication.
\end{corollary}
\begin{proof}
    To compute a best-response for $\alpha$, Alice can compute $\BR{\sigma/2}{\alpha}{\tI}$ with no communication.
    Then, for any $S \in \BR{\sigma/2}{\alpha}{\tI}$, Alice can compute $\alpha \hf(S)-\hc(S)$ and $\alpha \hf(S\cup\{n+1\})-\hc(S\cup\{n+1\})$, with polynomial communication (in particular, she can use Bob to answer a value query to $\hc$).
    By \Cref{prop:BR_in_approxDemand_f_c_submod}, the best set among the above is the agent's best response for $\alpha$. 
    Also, by \Cref{prop:sparse_demand_submod_f_c}, there are only poly-many such sets, so $poly(n)$ communication suffices.
\end{proof}

\begin{corollary}\label{cor:opt_exp_CC_f_c_submod}
    If Alice and Bob are given full knowledge of $f$, $\tc$ and $\sigma$, then finding the optimal contract in $\InstHatI$ requires exponential communication.
\end{corollary}

\begin{proof}
    We proceed by reducing $\DISJ{\binom{n}{n/2}}$ to the optimal contract problem.
    Given $x_c, x_f \in \{0,1\}^{\binom{n}{n/2}}$, we construct an instance $\InstHatI$ as described above. 
    Let $\alpha^*$ be the optimal contract for $\hI$ and let $S^* \subseteq [n+1]$ be the set which it incentivizes. As Alice and Bob have access to a best-response oracle, we can assume that the protocol returns $(\alpha^*,S^*)$. We show that $n+1 \in S^*$ if and only if $x_f \cap x_f \ne \emptyset$.
    If $n+1 \in S^*$, then by \Cref{obs:n+1Incentivizable}, $S^* \setminus \{n+1\} \in x_f \cap x_c$, and in particular $x_f \cap x_c \ne \emptyset$.
    If $n+1 \notin S^*$, then by \Cref{cor:opt_set_if_non_empty_intersection}, it must be that $x_c \cap x_f = \emptyset$.
    This concludes the proof. \qedhere
\end{proof}

\begin{proof}[Proof of \Cref{thm:cc_submod_f_c}]
Aiming for contradiction, assume there exists a protocol in which Alice and Bob use poly-many best response oracles and polynomial communication to find the optimal contract.
By \Cref{cor:simulate_BR_f_c_submod}, there exists an efficient protocol for the optimal contract if Alice and Bob receive $f,\tc$ and $\sigma$ as input. However this contradicts \Cref{cor:opt_exp_CC_f_c_submod}, which concludes the proof.
\end{proof}

\bibliographystyle{alpha}
\bibliography{refs.bib}

\appendix

\section{Inapproximability Without a Best-Response Oracle}\label{sec:CC_inapprox}
In this section we show the inapproximability of the optimal contract when $f$ and $c$ are both submodular or both supermodular, in the vanilla communication complexity model (see \Cref{sec:CC_model}), that is, without access to a best-response oracle.
In particular, we show that if the reward function $f$ is held by one party and the cost function $c$ is held by another, finding a contract for which the principal's utility is positive requires communicating an exponential number of bits, which implies multiplicative inapproximability. 

While these results are not too surprising, they do generalizes previous hardness results for this problem.
Specifically, a slight variation on the hardness result of \cite{iyer2012algorithms} (Theorem 5.2) implies that $\max_S f(S)-c(S)$ does not admit a constant-factor approximation in poly-time using value queries when both $f$ and $c$ are submodular. 
Our results in this section strengthen the above in two ways: (i) we establish inapproximability for any multiplicative approximation, even if it depends on $n$, 
and (ii) it is proved in a communication complexity model which generalizes value queries.

\subsection{Submodular $f$ and $c$} 
\begin{theorem}\label{thm:cc_inapprox_sub_sub}
    When $f$ and $c$ are submodular, any protocol which computes any multiplicative approximation to the optimal contract requires $\Omega(2^n/\sqrt{n})$ bits of communication.
\end{theorem}

We construct a family of instances with submodular $f$ and $c$ such that $\DISJ{\binom{n}{n/2}}$ can be reduced to finding a set of actions which yields positive utility for the principal.
\Cref{thm:cc_inapprox_sub_sub} follows directly from the following proposition.

\begin{proposition}
    There exists two monotone submodular functions $f$, parametrized by the vector $x_f \in \{0,1\}^{\binom{n}{n/2}}$, and $c$, parametrized by $x_c \in \{0,1\}^{\binom{n}{n/2}}$ such that $f(S) - c(S) > 0$ if and only if $|S|=n/2$ and $S \in x_f \cap x_c$.
\end{proposition}
\begin{proof}
Consider the following functions which define a family of instances parameterized by $x_f,x_c \in \{0,1\}^{\binom{n}{n/2}}$.

\begin{equation}\label{eq:sub-sub}
    f(S) = \begin{cases}
    8|S| & |S|<n/2 \\
    4n-4 & |S|=n/2 \land S \notin x_f\\
    4n-3 & |S|=n/2 \land S \in x_f \\
    2|S| + 3n - 3 & |S|>n/2
\end{cases}
\quad
c(S) = \begin{cases}
    8|S| & |S|<n/2 \\
    4n-2 & |S|=n/2 \land S \notin x_c\\
    4n-4 & |S|=n/2 \land S \in x_c \\
    2|S|+3n-2 & |S|>n/2
\end{cases}
\end{equation}

The fact that $f(S) - c(S) > 0$ if and only if $|S|=n/2$ and $S \in x_f \cap x_c$ can be easily verified. Also, the monotonicity and submodularity of $f$ and $c$ are trivial for $|S| < n/2-1$ and $|S|>n/2+1$. For $|S|\in \{n/2-1,n/2,n/2+1\}$, the possible values of $f(S)$ and $c(S)$ are depicted in \Cref{fig:sub-sub}, alongside the marginals which appear in parentheses.
\end{proof}
\begin{figure}[H]
    \centering
\begin{tikzpicture}[node distance=1cm, every node/.style={draw, circle}, >=Stealth]
\begin{scope}[xshift=10cm]
  \node (top) [draw=none] {\Large $f(S)$};
  \node (A) [below=of top] {$4n-8$};
  \node (B) [below left=of A] {$4n-4$};
  \node (C) [below right=of A, pattern=north east lines, pattern color=gray] {$4n-3$};
  \node (D) [below=2.2cm of A] {$4n-1$};
  \node (E) [draw=none, below=of D] {};

  \draw[->] (top) -- node[draw=none, right] {$(8)$} (A);
  \draw[->] (A) -- node[draw=none, left] {$(4)$} (B);
  \draw[->] (A) -- node[draw=none, right] {$(5)$} (C);
  \draw[->] (B) -- node[draw=none, left] {$(3)$} (D);
  \draw[->] (C) -- node[draw=none, right] {$(2)$} (D);
  \draw[->] (D) -- node[draw=none, right] {$(2)$} (E);
\end{scope}

\begin{scope}[xshift=16cm]
  \node (top2) [draw=none] {\Large $c(S)$};
  \node (A2) [below=of top2] {$4n-8$};
  \node (B2) [below left=of A2] {$4n-2$};
  \node (C2) [below right=of A2, pattern=north east lines, pattern color=gray] {$4n-4$};
  \node (D2) [below=2.2cm of A2, minimum size=1.48cm] {$4n$};
  \node (E2) [draw=none, below=of D2] {};

  \draw[->] (top2) -- node[draw=none, right] {$(8)$} (A2);
  \draw[->] (A2) -- node[draw=none, left] {$(6)$} (B2);
  \draw[->] (A2) -- node[draw=none, right] {$(4)$} (C2);
  \draw[->] (B2) -- node[draw=none, left] {$(2)$} (D2);
  \draw[->] (C2) -- node[draw=none, right] {$(4)$} (D2);
  \draw[->] (D2) -- node[draw=none, right] {$(2)$} (E2);
\end{scope}

\begin{scope}
  \node (minus1) [draw=none, left=3cm of A] {$|S|=\frac{n}{2}-1$};  
  \node (zero) [draw=none, left=1.95cm of B] {$|S|=\frac{n}{2}$};  
  \node (plus1) [draw=none, left=3cm of D] {$|S|=\frac{n}{2}+1$};  
\end{scope}
\end{tikzpicture}
\caption{The values and marginal values of $f$ and $c$ in \Cref{eq:sub-sub} when $|S|\in \{n/2-1,n/2,n/2+1\}$. The possible values of $f(S)$ (left) and $c(S)$ (right) appear inside the circles, the marginals appear in parentheses. The marked circles correspond to sets $S$ such that $S \in x_f$ or $S \in x_c$.}
\label{fig:sub-sub}
\end{figure}

\subsection{Supermodular $f$ and $c$}
We also give an analogous result when $f$ and $c$ are supermodular.

\begin{theorem}\label{thm:cc_inapprox_sup_sup}
    When $f$ and $c$ are supermodular, any protocol which computes any multiplicative approximation to the optimal contract requires $\Omega(2^n/\sqrt{n})$ bits of communication.
\end{theorem}

\begin{proposition}
    There exists two monotone supermodular functions $f$, parametrized by $x_f \in \{0,1\}^{\binom{n}{n/2}}$, and $c$, parametrized by $x_c \in \{0,1\}^{\binom{n}{n/2}}$ such that  $f(S) - c(S) > 0$ if and only if $|S|=n/2$ and $S \in x_f \cap x_c$.
\end{proposition}

\begin{proof}
    Consider the following family of instances parameterized by $x_f,x_c \in \{0,1\}^{\binom{n}{n/2}}$.
\begin{equation}\label{eq:sup-sup}
f(S) = \begin{cases}
    2|S| & |S|<n/2 \\
    n & |S|=n/2 \land S \notin x_f\\
    n+1 & |S|=n/2 \land S \in x_f \\
    6|S|-2n-1 & |S|>n/2
\end{cases}
\quad
c(S) = \begin{cases}
    2|S| & |S|<n/2 \\
    n+2 & |S|=n/2 \land S \notin x_c\\
    n & |S|=n/2 \land S \in x_c \\
    6|S|-2n & |S|>n/2
\end{cases}    
\end{equation}

The fact that $f(S) - c(S) > 0$ if and only if $|S|=n/2$ and $S \in x_f \cap x_c$ can be easily verified. Also, the monotonicity and supermodularity of $f$ and $c$ are trivial For $|S| < n/2-1$ and $|S|>n/2+1$. For $|S|\in \{n/2-1,n/2,n/2+1\}$, the possible values of $f(S)$ and $c(S)$ are depicted in \Cref{fig:sup-sup}, alongside the marginals which appear in parentheses.
\end{proof}
\begin{figure}[H]
    \centering
\begin{tikzpicture}[node distance=1cm, every node/.style={draw, circle}, >=Stealth]
\begin{scope}[xshift=10cm]
  \node (top) [draw=none] {\Large $f(S)$};
  \node (A) [below=of top] {$n-2$};
  \node (B) [below left=of A, minimum size=1.25cm] {$n$};
  \node (C) [below right=of A, pattern=north east lines, pattern color=gray] {$n+1$};
  \node (D) [below=2.2cm of A] {$n+5$};
  \node (E) [draw=none, below=of D] {};

  \draw[->] (top) -- node[draw=none, right] {$(2)$} (A);
  \draw[->] (A) -- node[draw=none, left] {$(2)$} (B);
  \draw[->] (A) -- node[draw=none, right] {$(3)$} (C);
  \draw[->] (B) -- node[draw=none, left] {$(5)$} (D);
  \draw[->] (C) -- node[draw=none, right] {$(4)$} (D);
  \draw[->] (D) -- node[draw=none, right] {$(6)$} (E);
\end{scope}

\begin{scope}[xshift=16cm]
  \node (top2) [draw=none] {\Large $c(S)$};
  \node (A2) [below=of top2] {$n-2$};
  \node (B2) [below left=of A2] {$n+2$};
  \node (C2) [below right=of A2, minimum size=1.25cm, pattern=north east lines, pattern color=gray] {$n$};
  \node (D2) [below=2.2cm of A2] {$n+6$};
  \node (E2) [draw=none, below=of D2] {};

  \draw[->] (top2) -- node[draw=none, right] {$(2)$} (A2);
  \draw[->] (A2) -- node[draw=none, left] {$(4)$} (B2);
  \draw[->] (A2) -- node[draw=none, right] {$(2)$} (C2);
  \draw[->] (B2) -- node[draw=none, left] {$(4)$} (D2);
  \draw[->] (C2) -- node[draw=none, right] {$(6)$} (D2);
  \draw[->] (D2) -- node[draw=none, right] {$(6)$} (E2);
\end{scope}

\begin{scope}
  \node (minus1) [draw=none, left=3cm of A] {$|S|=\frac{n}{2}-1$};  
  \node (zero) [draw=none, left=1.95cm of B] {$|S|=\frac{n}{2}$};  
  \node (plus1) [draw=none, left=3cm of D] {$|S|=\frac{n}{2}+1$};  
\end{scope}
\end{tikzpicture}
\caption{The values and marginal values of $f$ and $c$ in \Cref{eq:sup-sup} when $|S|\in \{n/2-1,n/2,n/2+1\}$. The possible values of $f(S)$ (left) and $c(S)$ (right) appear inside the circles, the marginals appear in parentheses. The marked circles correspond to sets $S$ such that $S \in x_f$ or $S \in x_c$.}
\label{fig:sup-sup}
\end{figure}

\section{FPTAS with Value and Best-Response Queries}\label{sec:FPTAS_subadditive_c}

\newcommand{\actions}[0]{[n]}
\newcommand{\opt}[0]{\text{OPT}}
In this appendix, we observe that the FPTAS of \cite{multimulti} for monotone $f$ and additive $c$ can also be used for the case where $c$ is subadditive,
when instead of a demand query we have access to a best-response query. We present the full argument here for completeness.

\begin{claim}[FPTAS]\label{cla:FPTAS}
For monotone $f$ and monotone and subadditive $c$, \Cref{alg:fptas} gives a $(1-\epsilon)$-approximation to the optimal principal utility with $O\left(\frac{n^2}{\epsilon} \right)$ many value and best-response queries.
\end{claim}

Denote the optimal welfare by $\opt = \max_{\hat{S} \subseteq [n]} (f(\hat{S}) - c(\hat{S})).$ Denote the optimal contract and the agent's choice of actions under this contract by $\alpha^\star$ and $S^\star$.

Our key lemma (\Cref{lem:4m}) shows upper and lower bounds on  $\alpha^\star$.

To bound $\alpha^\star$, we use that, by a reduction to the non-combinatorial model of \cite{dutting2019simple}, the gap between the optimal welfare and the principal's utility under the optimal contract is at most $2^n$.

\begin{observation}[\cite{dutting2019simple}]\label{obs:welfare-single}
    For general $f$ and general $c$, there exists a contract $\alpha$ that guarantees a utility for the principal of at least  $\frac{\opt}{2^n}$. 
\end{observation}

Using this observation, we now show our key lemma, which sets lower and upper bounds on $\alpha^\star$.
\begin{lemma}\label{lem:4m}
    Consider monotone $f$ and monotone and subadditive $c$. Let $j^\star \in \arg\max_{j \in S^\star} c(\{j\})$ and assume $\opt > 0$.
    Then we have
    $\alpha_{\min} \leq \alpha^\star \leq  \alpha_{\max}$, where $\alpha_{\min} =1- \frac{\opt}{ c(\{j^\star\}) + \opt}$ and $\alpha_{\max} = 1- \frac{\opt}{n\cdot 2^n (c(\{j^\star\}) + \opt )}$.
\end{lemma}

\begin{proof}
The set that maximizes $f(S)-c(S)$ is also the best response of the agent under contract $\alpha=1$.
We note that since $\opt>0$, it must be that $\alpha^\star <1 $.
By \Cref{obs:welfare-single}, it holds that 
\begin{equation}
    \frac{\opt}{2^n}\leq (1-\alpha^\star)  f(S^\star)   \leq   f(S^\star) -c(S^\star) \leq \opt,\label{eq:welfare3}
\end{equation}
where the second inequality holds since the agent's expected payment $\alpha^\star \cdot f(S^\star)$ covers the agent's cost $c(S^\star)$,  as the agent's expected utility is non-negative.

We first show that $\alpha^\star \geq \alpha_{\min}$. It holds that, 
$$
\alpha^\star \cdot (  c(S^\star) + \opt  )-c(S^\star) \geq \alpha^\star \cdot (  c(S^\star) + f(S^\star)-c(S^\star) )-c(S^\star) =  \alpha^\star \cdot f(S^\star) -c(S^\star) \geq 0, 
$$
where the first holds by the last inequality in  Equation~\eqref{eq:welfare3}, and the last inequality is since the utility of the agent under $\alpha^\star$ and $S^\star$ is non-negative.

By rearranging, we get that 
$$
\alpha^\star \geq \frac{c(S^\star)}{ c(S^\star) + \opt} \geq \frac{c(\{j^\star\})}{c(\{j^\star\}) +\opt } = 1- \frac{\opt}{ c(\{j^\star\}) + \opt} = \alpha_{\min},
$$
where we used that $c$ is monotone and $z/(z+\opt)$ is non-decreasing in $z$. 

To show that $\alpha^\star \leq \alpha_{\max}$, observe that 
$$ (1-\alpha^\star) (c(S^\star) + \opt) \geq   (1-\alpha^\star) (c(S^\star) + f(S^\star) - c(S^\star)) =  (1-\alpha^\star) f(S^\star) \geq  \frac{\opt}{2^n},  $$
where the first and last inequality hold by the last and first inequality in Equation~\ref{eq:welfare3}, respectively.

Rearranging this, we obtain
$$  \alpha^\star \leq 1- \frac{\opt}{2^n (c(S^\star) + \opt )} \leq 1- \frac{\opt}{n\cdot 2^n (c(\{j^\star\}) + \opt )}  = \alpha_{\max},$$
where we used that $c$ is subadditive so that $c(S^\star) \leq n \cdot c(\{j^\star\})$. 
\end{proof}

\begin{algorithm}
\caption{FPTAS for a single agent using value and demand oracles}\label{alg:fptas}
   \hspace*{\algorithmicindent} \textbf{Parameter:}  $\epsilon \in (0,1) $ \\
   \hspace*{\algorithmicindent} \textbf{Input:} Reward function $f:2^{[n]} \rightarrow \reals_{\geq 0}$, cost function $c:2^{[n]} \rightarrow \reals_{\geq 0}$ \\
    \hspace*{\algorithmicindent} \textbf{Output:}  A contract $\alpha$ and a best response set $S$ 
\begin{algorithmic}[1]
\State Let $\alpha =0$ and $S = \arg\max_{\hat{S}: c(\hat{S}) = 0} f(\hat{S})$ \label{st:init}
\State Let $\opt = \max_{\hat{S} \subseteq [n]} (f(\hat{S}) - c(\hat{S}))$ \label{st:opt}
\For{$j\in \actions$ with $c(\{j\})>0$}
\For{$k=0,\ldots,\lceil \log_{1/(1-\epsilon)} n \cdot 2^n \rceil $}
\State Let $\alpha_{j, k} = 1 - (1-\epsilon)^{k+1} \cdot  \frac{\opt}{ c(\{j\}) + \opt}$
\State Let $S_{j,k} = \arg\max_{\hat{S}\subseteq [n]}  \left(\alpha_{j,k} \cdot  f(\hat{S}) - c(\hat{S})\right)$ \label{st:br}
\If{$(1-\alpha_{j, k}) f(S_{j,k}) > (1-\alpha)f(S)$}
\State $\alpha=\alpha_{j, k}$
\State $S= S_{j,k}$
\EndIf
\EndFor
\EndFor
\State \Return $\alpha,S$
\end{algorithmic}
\end{algorithm}

We are now ready to prove \Cref{cla:FPTAS}.

\begin{proof}[Proof of \Cref{cla:FPTAS}]
It is clear that \Cref{alg:fptas} can be implemented with access to  value queries and best-response queries (we need best-response queries in Steps~\ref{st:init}, \ref{st:opt}, and \ref{st:br}). The query complexity of \Cref{alg:fptas} is upper bounded by $O(n \cdot \log_{1/(1-\epsilon)}n\cdot 2^n)$, which is $O(\frac{n^2}{\epsilon})$. 

It remains to show that \Cref{alg:fptas} achieves the claimed approximation guarantee.
Recall that we use $\alpha^\star,S^\star$ to denote the optimal contract and the agent's choice of actions under this contract. Note that if $u_p(S^\star,\alpha^\star) \leq 0$, the claim holds trivially because Algorithm~\ref{alg:fptas} ensures that $u_p(S, \alpha) \geq 0$. Otherwise, $u_p(S^\star,\alpha^\star) = (1-\alpha^\star)f(S^\star) > 0$. It must then hold that $\alpha^\star < 1$ and $f(S^\star) > 0$. 
So, in particular, $S^\star \neq \emptyset$, so that $j^\star \in \arg\max_{j \in S^\star} c(\{j\})$ is well defined.  Also note that $\opt = \max_{\hat{S} \subseteq [n]} (f(\hat{S}) - c(\hat{S})) \geq f(S^\star) - c(S^\star) \geq (1-\alpha^\star) f(S^\star) >0$, because the agent's utility from $S^\star$ is non-negative. 

If $c(\{j^\star\})=0$, then, by subadditivity of $c$, we must have $c(S^\star) = 0$ and the optimal contract is $\alpha=0$, which is the contract considered in Step~\ref{st:init} of the algorithm. 
Otherwise, $c(\{j^\star\}) > 0$. Consider the iteration of Algorithm~\ref{alg:fptas} in which $j = j^\star$.
Let us denote $\alpha_{\min} =1- \frac{\opt}{ c(\{j^\star\}) + \opt}$ and $\alpha_{\max} = 1- \frac{\opt}{n\cdot 2^n (c(\{j^\star\}) + \opt )}$.
We claim that then there must be a value of  $k \in \{0,\ldots,\lceil \log_{1/(1-\epsilon)} n\cdot 2^n \rceil\}$ such that $1 - \alpha_{j, k}  \leq 1 - \alpha^\star \leq  \frac{1 - \alpha_{j, k}}{1-\epsilon}$.
Indeed, for $k = 0$ we have $\frac{1-\alpha_{j^\star,0}}{1-\epsilon} =  1- \alpha_{\min} \geq  1-\alpha^\star$, where the inequality follows by Lemma~\ref{lem:4m}, while for $k =\lceil \log_{1/(1-\epsilon)}m\cdot 2^m \rceil$ we have $1-\alpha^\star \geq 1- \alpha_{\max} \geq 1-\alpha_{j^\star,\lceil \log_{1/(1-\epsilon)}n\cdot 2^n \rceil}$, where the first inequality follows again by Lemma~\ref{lem:4m}. So there must be a $k \in \{0,\ldots,\lceil \log_{1/(1-\epsilon)}n \cdot 2^n \rceil\}$ with the desired properties.

We claim that for this choice of $j,k$, contract $\alpha_{j, k}$ provides a $(1-\epsilon)$-approximation to the optimal contract.
To see this, let $S$ be the choice of the agent under $\alpha_{j, k}$.
By \Cref{obs:critVals}, since $\alpha_{j, k} \geq \alpha^\star$, it must hold that $f(S) \geq f(S^\star)$. We thus obtain,
$(1-\alpha_{j, k}) f(S) \geq (1-\alpha_{j, k}) f(S^\star) \geq  (1-\epsilon)(1-\alpha^\star) f(S^\star)$,
which completes the proof.
\end{proof}

\section{Omitted Proofs}
\subsection{Missing Proofs from \Cref{sec:eqRev}} \label{apx:missingProof} 

\begin{proof}[Proof of \Cref{lem:alphaMon}]
    One can easily show that $\frac{1}{2}\left(\sqrt{4\alpha_t^2 - 8\alpha_t + 5} - 1\right)$ is always positive for $\alpha_t \in [0,1)$, which implies that $\alpha_t$ is monotonically increasing in $t$. 

    The fact that $\alpha_{t+1} < 1$ can be  derived from the recurrence relation inductively.
    Aiming for contradiction, assume $\alpha_t < 1$ and  $\alpha_{t+1} > 1$, it implies that
    $$
    1 < \alpha_t + \frac{1}{2}\left(\sqrt{4\alpha_t^2 - 8\alpha_t + 5} - 1\right). 
    $$
    By subtracting $\alpha_t - \frac{1}{2}$ and then squaring  both sides of the equation, one can show that $\alpha_t > 1$, contradicting the induction hypothesis.
\end{proof}

\begin{proof}[Proof of \Cref{lem:fMarginalDecrease}]
For $t=0$ we have $f_1 - f_0 = f_1 \approx 1.618$ and for $t=1$, $f_2 - f_1 \approx 1.1338$. 
For $t>1$, let $z_t = \frac{1}{2}\left( \sqrt{4\alpha_t^2-8\alpha_t+5} - 1 \right)$, so that $\alpha_{t+1} = \alpha_t + z_t$.
    Rearranging we get
    \begin{eqnarray*}
        f_{t+1} - f_t &=&
        \frac{z_t}{(1-\alpha_t - z_t)(1-\alpha_t)},
    \end{eqnarray*}
    we will show that is a decreasing function of $t$.
    Substituting $\alpha_t$ with $x$, and $z_t$ with $z(x)$, we'll show that $g(x) = \frac{z(x)}{(1-x - z(x))(1-x)}$ 
    is monotonically decreasing for $x \in [0,1)$. The claim follows from continuity.
    It is enough to show that the numerator of $\frac{\partial g(x)}{\partial x}$, which we denote by $\hat{g}(x)$,  is negative for $x \in [0,1)$.
    \begin{eqnarray*}
        \hat{g}(x) &=& 
        z'(x)(1-x-z(x))(1-x) - z(x)[(-1-z'(x))(1-x) - (1-x-z(x))] \\
        &=& 
        z'(x)(1-x)^2 - z'(x)z(x)(1-x) 
        +z(x)(1+z'(x))(1-x)+z(x)(1-x) -z^2(x) \\
        &=& 
        z'(x)(1-x)^2 + z(x)(1-x)[-z'(x)+1+z'(x)+1] -z^2(x) \\
        &=& 
        z'(x)(1-x)^2 + 2z(x)(1-x) -z^2(x).
    \end{eqnarray*}
    Replacing $z$ with the relevant expressions:
    \begin{eqnarray*}
        \hat{g}(x)&=&
        \frac{-2(1-x)^3}{\sqrt{4x^2-8x+5}} + (\sqrt{4x^2-8x+5}-1)(1-x)-\frac{1}{4}(4x^2-8x+6 -2\sqrt{4x^2-8x+5}) \\
        &=&
        \frac{-2(1-x)^3}{\sqrt{4x^2-8x+5}} + 
        \sqrt{4x^2-8x+5}\left(\frac{3}{2}-x\right)-(1-x)-\left(x^2-2x+\frac{3}{2}\right) \\
        &=&
        \frac{-2(1-x)^3}{\sqrt{4x^2-8x+5}} + 
        \sqrt{4x^2-8x+5}\left(\frac{3}{2}-x\right) -x^2+3x-\frac{5}{2}.
    \end{eqnarray*}
    Since in the range $x \in [0,1)$ it holds that $\sqrt{4x^2-8x+5}>0$, $\hat{g}(x) < 0$ follows if $h(x) < 0$, where $h(x) := \sqrt{4x^2-8x+5} \cdot \hat{g}(x)$.
    \begin{eqnarray*}
        h(x) &=&
        -2(1-x)^3 + (4x^2-8x+5)\left(\frac{3}{2}-x\right)-\left(x^2-3x+\frac{5}{2}\right)\sqrt{4x^2-8x+5} \\
        &=&
        -2x^3+8x^2-11x+\frac{11}{2}-\left(x^2-3x+\frac{5}{2}\right)\sqrt{4x^2-8x+5}, 
    \end{eqnarray*}
   $h(x)$ is negative if and only if the following holds
   \begin{eqnarray*}
       -2x^3+8x^2-11x+\frac{11}{2}
       &<&
       \left(x^2-3x+\frac{5}{2}\right)\sqrt{4x^2-8x+5} \\
       \left(-2x^3+8x^2-11x+\frac{11}{2}\right)^2
       &<&
       \left(x^2-3x+\frac{5}{2}\right)^2(4x^2-8x+5) \\
      4x^6 - 32x^5 + 108x^4 -198x^3 + 209x^2 - 121x + \frac{121}{4}
      &<&
      \left(x^4-6x^3+14x^2-15x+\frac{25}{4}\right)(4x^2-8x+5) \\
      108x^4 -198x^3 + 209x^2 - 121x + \frac{121}{4}
      &<&
      109x^4 -202x^3 + 215x^2 - 125x + \frac{125}{4} \\
      0
      &<&
      x^4 -3x^3 + 6x^2 - 4x + 1.
   \end{eqnarray*}
   One can easily verify that this is true for any $x \ge 0$.
\end{proof}

\subsection{Missing Proofs from \Cref{sec:eqToVal}}\label{apx:proofs_family_equalRev}

\begin{claim}\label{cla:eps_0_positive}
    The right-hand side of equation \Cref{eq:eps_0} is strictly positive. Namely, the following three inequalities hold for any equal revenue instance $\eqRevInst$:
    
\begin{equation}\label{eq:apx_eps_1}
0< \min_{t \in \{1,\dots,2^n-1\}} \feq(S_{t})-\feq(S_{t-1}) 
\end{equation}
\begin{equation}\label{eq:apx_eps_2}
    0 < \min_{t \in \{2,\dots,2^n-1\}} \frac{\ceq(S_t) - \ceq(S_{t-1})}{\alpha_{t-1}} - (\feq(S_t) - \feq(S_{t-1}))
\end{equation}
\begin{equation}\label{eq:apx_eps_3}
    0 < \min_{\substack{S \subseteq T \subseteq [n], \\ i \in [n]\setminus T}} \feq(i \mid S) - \feq(i \mid T)
\end{equation}
\end{claim}

\begin{proof}
Inequality \ref{eq:apx_eps_1} follows from the fact that there $2^n-1$ different incentivizable set and by \Cref{obs:critValStructure}, their rewards are strictly increasing.
Inequality \ref{eq:apx_eps_2} follows from from the fact that there are $2^n-1$ distinct critical values, which implies that $\alpha_{t-1} < \alpha_t = \frac{\ceq(S_t) - \ceq(S_{t-1})}{\feq(S_t) - \feq(S_{t-1})}$. 
Inequality \ref{eq:apx_eps_3} follows immediately from \Cref{lem:f_strictly_submodolar} below, which establishes that $\feq$ is strictly submodular.
\end{proof}

\begin{lemma}\label{lem:f_strictly_submodolar}
    Let $\eqRevInst$ be an equal revenue instance with submodular $\feq$ and additive $\ceq$, then $\feq$ is strictly submodular.
\end{lemma}
\begin{proof}
\newcommand{\ST}{S_t}
\newcommand{\STP}{S_{t'}}
\newcommand{\Sti}{S_t \cup \{i\}}
\newcommand{\Stpi}{S_{t'} \cup \{i\}}
Fix any two sets $\ST \subsetneq \STP \subseteq [n]$, and $i \notin \STP$. We will show that
$$
\feq(i \mid \ST) - \feq(i \mid \STP) > 0.
$$
Observe that $t < t'$. 
Let $\alpha_{ti}, \alpha_{t'}$ the critical value that incentivize the sets $\Sti$ and $S_{t'}$, respectively.
Aiming for contradiction, assume $\feq(\Sti) - \feq(\ST) = \feq(\Stpi) - \feq(\STP)$. Together with the additivity of $\ceq$ we get,
$$
\beta_{t'i} := \frac{\ceq(i \mid S_{t'})}{\feq(i \mid S_{t'})} =
\frac{\ceq(i \mid S_t)}{\feq(i \mid S_t)} =:\beta_{ti}.
$$
Observe that $\Stpi$ dominates $\STP$ for any contract $\alpha \ge \beta_{t'i}$, thus $\alpha_{t'} < \beta_{t'i}$. 
Also, $\ST$ dominates $\Sti$ for any contract $\alpha < \beta_{ti}$, thus $ \beta_{ti} \le \alpha_{ti}$. Together we get $\alpha_{t'} < \beta_{t'i} = \beta_{ti} \le \alpha_{ti}$.

Our assumption also implies that
$\feq(\STP) - \feq(\ST) = \feq(\Stpi) - \feq(\Sti)$, and thus, together with the additivity of $\ceq$, and a similar argument to the above,
\[
 \alpha_{ti} < \frac{\ceq(\Stpi) - \ceq(\Sti)}{ \feq(\Stpi) - \feq(\Sti)} =\frac{\ceq(\STP) - \ceq(\ST)}{\feq(\STP) - \feq(\ST)} \le \alpha_{t'},
\]
which concludes the proof.
\end{proof}

\begin{proof}[Proof of \Cref{prop:feq_kSubmod}]
    Fix $k$. Non-negativity follows as $\eps >0$. For submodularity, observe that for any two sets $S \subseteq T \subseteq [n]$ and any action $ i \in [n]\setminus T$, it holds that
    $$
    \feq_k(i \mid S) - \feq_k(i \mid T) \ge \feq(i \mid S) - \feq(i \mid T) - \eps > 0,$$
    where the first inequality follows from the definition of $\feq_k$ and last  follows from the definition of $\eps$.
    For monotonicity, observe that for any two sets of actions, $S \subsetneq T$ we have
    $$\feq_k(T) - \feq_k(S) \ge \feq(T) - \feq(S) - \eps > 0,$$
    where the last inequality follows from \Cref{eq:apx_eps_1}.
\end{proof}

\begin{proof}[Proof of \Cref{prop:feq_k_CVs}]
    First observe that for any $t < k$ and any $t>k+1$, $\alpha'_t = \alpha_t$, as
    $$
    \alpha'_t 
    =
    \frac{\ceq(S_t) - \ceq(S_{t-1})}{\feq_k(S_t) - \feq_k(S_{t-1})} 
    =
    \frac{\ceq(S_t) - \ceq(S_{k-1})}{\feq(S_t) - \feq(S_{t-1})} 
    =
    \alpha_t.
    $$  
    For $t=k+1$ we have for any $\eps >0$,
    \begin{align*}
        \alpha'_{k+1}  - \alpha'_k
        &=
        \frac{\ceq(S_{k+1}) - \ceq(S_{k})}{\feq_k(S_{k+1}) - \feq_k(S_{k})} 
        - 
        \frac{\ceq(S_{k}) - \ceq(S_{k-1})}{\feq_k(S_{k}) - \feq_k(S_{k-1})} \\
        &=
        \frac{\ceq(S_{k+1}) - \ceq(S_{k})}{\feq(S_{k+1}) - \feq(S_{k}) - \eps} 
        - 
        \frac{\ceq(S_{k}) - \ceq(S_{k-1})}{\feq(S_{k}) + \eps - \feq(S_{k-1})} \\
        &>
        \frac{\ceq(S_{k+1}) - \ceq(S_{k})}{\feq(S_{k+1}) - \feq(S_{k})} 
        - 
        \frac{\ceq(S_{k}) - \ceq(S_{k-1})}{\feq(S_{k}) - \feq(S_{k-1})} \\
        &=
        \alpha_{k+1}  - \alpha_k \\
        &>
        0.
    \end{align*}
    For $t=k$ we have,
    \begin{align*}
        \alpha'_k  - \alpha'_{k-1}
        &=
        \frac{\ceq(S_k) - \ceq(S_{k-1})}{\feq_k(S_k) - \feq_k(S_{k-1})} 
        - 
        \frac{\ceq(S_{k-1}) - \ceq(S_{k-2})}{\feq_k(S_{k-1}) - \feq_k(S_{k-2})} \\
        &=
        \frac{\ceq(S_k) - \ceq(S_{k-1})}{\feq(S_k) + \eps - \feq(S_{k-1})} 
        - 
        \frac{\ceq(S_{k-1}) - \ceq(S_{k-2})}{\feq(S_{k-1}) - \feq(S_{k-2})} \\
        &=
        \frac{\ceq(S_k) - \ceq(S_{k-1})}{\feq(S_k) - \feq(S_{k-1}) + \eps} 
        - 
        \alpha_{k-1}.\\
    \end{align*}
    Using the fact that $\eps < \feq_k(S_k) - \feq(S_{k-1})$, one can easily show that the above expression is greater than zero if and only if 
    $$
    \frac{\ceq(S_k) - \ceq(S_{k-1})}{ \alpha_{k-1}} - (\feq(S_k) - \feq(S_{k-1})) > \eps,
    $$
    which is implied by \Cref{eq:apx_eps_2}.
    We conclude that there are $2^n-1$ distinct critical values, and $\alpha'_t$ incentivizes $S_t$.

    We now proceed to show that $\alpha'_k$ is the unique optimal contract for the instance $\langle n, \feq_k, \ceq \rangle$.
    As for any $t \notin \{k,k+1\}$, as $\feq_k(S_t) = \feq(S_t)$ and $\alpha_t = \alpha'_{t}$, we get that 
    $$
    u_p(\alpha'_t) = (1-\alpha'_t)\feq(S_t) =(1-\alpha_t)\feq(S_t) = 1.
    $$
    For $t = k+1$ we have
    \begin{align*}
    u_p(\alpha'_{k+1})
    &=
    \left(1-\alpha'_{k+1})\feq_k(S_{k+1}\right) \\
    &=
    \left(1-\frac{\ceq(S_{k+1}) - \ceq(S_{k})}{\feq_k(S_{k+1}) - \feq_k(S_{k})}\right)\feq_k(S_{k+1}) \\
    &=
    \left(1-\frac{\ceq(S_{k+1}) - \ceq(S_{k})}{\feq(S_{k+1}) - \feq(S_{k}) - \eps}  \right)\feq(S_{k+1}) && (\text{definition of } \feq_k) \\
    &<
    \left(1-\frac{\ceq(S_{k+1}) - \ceq(S_{k})}{\feq(S_{k+1}) - \feq(S_{k})}  \right)\feq(S_{k+1}) \\
    &=
    (1-\alpha_{k+1})\feq(S_{k+1}) \\
    &=
    1.
    \end{align*}
    For $t = k$,
    \begin{align*}
    u_p(\alpha'_{k})
    &=
    \left(1-\alpha'_{k})\feq_k(S_{k}\right) \\
    &=
    \left(1-\frac{\ceq(S_{k}) - \ceq(S_{k-1})}{\feq_k(S_{k}) - \feq_k(S_{k-1})}\right)\feq_k(S_{k}) \\
    &=
    \left(1-\frac{\ceq(S_{k}) - \ceq(S_{k-1})}{\feq(S_{k}) + \eps - \feq(S_{k-1})}\right)(\feq(S_{k}) + \eps) && (\text{definition of } \feq_k) \\
    &>
    \left(1-\frac{\ceq(S_{k}) - \ceq(S_{k-1})}{\feq(S_{k}) - \feq(S_{k-1})}\right)(\feq(S_{k}) + \eps) \\
    &>
    \left(1-\alpha_{k})\feq(S_{k}\right) \\
    &=
    1.
    \end{align*}   
    This completes the proof.
\end{proof}

\subsection{Missing Proofs from \Cref{sec:CC_submod_fc}}\label{sec:Omitted_CC_submod_fc}

\begin{proof}[Proof of \Cref{prop:ctilde_basix_props}]
\textbf{Non-negativity:}
For any $S \subseteq n $ and $S \ne \emptyset$,
$$
\tc(S) = c(S) - \delta|S|^2 > c(S) - \phi_c \ge c(S_1) - \phi_c = c(S_1) - c(S_0) - \phi_c > 0
$$
\textbf{Monotonicity:}
For any $S \subsetneq T$,
$$
\tc(T) - \tc(S) = c(T)-c(S) + \delta(|S|^2-|T|^2) \ge \phi_c - \delta n^2 > 0
$$
\textbf{Submodularity:}
For any $S \subsetneq T$ and $i \notin T$,
\begin{align*}
\tc(i \mid S) - \tc(i \mid T) 
&=
c(i \mid S) - c(i \mid T) - \delta((|S|+1)^2-|S|^2) + \delta((|T|+1)^2-|T|^2) \\
&= 
\delta 
[(|T|+1)^2-|T|^2) - (|S|+1)^2-|S|^2)] \\
&>0,  
\end{align*}
where the last inequality follows from the strict convexity of $x^2$.
\end{proof}

\begin{proof}[Proof of \Cref{prop:exp_many_cvs_maintained}]
Fix $t\in \{2,\dots,2^n-2\}$. Denote $f_t = f(S_t)$ and similarly for $f_{t+1}$ and $f_{t-1}$.
To show that $\talpha_{t+1} > \talpha_t$, it suffices to show that 
$$
\alpha_{t+1} - \frac{\delta\cdot n^2}{f_{t+1}-f_t}
>
\alpha_{t} + \frac{\delta\cdot n^2}{f_{t}-f_{t-1}}
$$
which is equivalent to
\begin{align*}
     \alpha_{t+1} - \alpha_t
     &> \delta \cdot n^2 \left(\frac{1}{f_t-f_{t-1}} + \frac{1}{f_{t+1}-f_t}\right)
    = \delta \cdot n^2 \underbrace{\left(\frac{f_{t+1}-f_{t-1}}{(f_t-f_{t-1})(f_{t+1}-f_t)}\right)}_{> 0}.
\end{align*}
So setting,
$$
\delta < (\alpha_{t+1}-\alpha_t)\cdot \frac{1}{n^2} \cdot \left(\frac{(f_t-f_{t-1})(f_{t+1}-f_t)}{(f_{t+1}-f_{t-1})}\right),
$$
The last critical value must satisfy, $\talpha_{2^n-1} <1$:
\begin{align*}
\talpha_{2^n-1}
&=
\frac{\tc_{2^n-1} -\tc_{2^n-2}}{f_{2^n-1} - f_{2^n-2}}
<
\frac{c_{2^n-1} - c_{2^n-2} + \delta n^2}{f_{2^n-1} - f_{2^n-2}} 
=
\alpha_{2^n-1} + \frac{\delta n^2}{f_{2^n-1} - f_{2^n-2}}.
\end{align*}
Equivalently,
$$
\delta < \frac{(1-\alpha_{2^n-1})(f_{2^n-1}-f_{2^n-2})}{n^2}.
$$
Lastly, $\talpha_1 > 0$ if and only if
\begin{align*}
\talpha_{1}
&=
\frac{\tc_{1} - 0}{f_1 - 0}
>
\frac{c_1 - \delta n^2}{f_1} = \alpha_1 - \frac{\delta n^2}{f_1}.
\end{align*}
So setting,
$$
0 < \delta < \alpha_1 \frac{f_1}{n^2}
$$
suffices.
As each $\talpha_t = \frac{\tc(S_t)-\tc(S_{t-1})}{f(S_t)-f(S_{t-1})}$ corresponds to a contract for which the agent's best response change in the instance $(n,f,\tc)$, there are $2^n-1$ distinct critical values.
\end{proof}

\begin{proof}[Proof of \Cref{prop:monotone_submod_hf_hc}]
    \textbf{Monotonicity of $\hc$:} 
    Fix $S \subsetneq T \subseteq [n+1]$. 
    \begin{align*}
            \hc(T) - \hc(S) 
            &\ge 
            \hc(T \cap [n]) - \hc(S \cup \{n+1\}) 
            = \tc(T\cap [n]) - \tc(S\cap [n]) - \tc(n+1\mid S\cap[n]) \\
            &\ge \phi_{\tc} - \tc(n+1\mid S\cap[n]) 
            \ge z - z/2 
            > 0.
    \end{align*}
    The first inequality follows as the marginal cost of $n+1$ is non-negative.  
    
    \noindent\textbf{Submodularity of $\hc$:} 
    We will show that for any $S \subseteq [n+1]$ and any distinct $i,j \notin S$,
    \begin{equation}\label{eq:tcSubMod}
        \hc(S \cup \{i\}) + \hc(S \cup \{j\}) > \hc(S \cup \{i,j\}) + \hc(S)  
    \end{equation}

    \begin{itemize}
    \setlength\itemsep{-0.25em}
    \item If $n+1 \notin S \cup \{i,j\}$, then the claim follows from \Cref{prop:ctilde_basix_props}.

    \item If $i=n+1$ (and similarly for the case where $j=n+1$), by rearranging \Cref{eq:tcSubMod} one can see that it is equivalent to,
    $\hc(n+1 \mid S) \ge \hc(n+1 \mid S \cup \{j\})$. The claim follows by the fact that $\hc(n+1 \mid T)$ is weakly decreasing with the cardinality of $T$.

    \item If $n+1 \in S$, let $S' = S \setminus \{n+1\}$. 
    Then be rewriting $\hc(S \cup \{i\})$ as $\hc(n+1 \mid S' \cup \{i\}) + \hc(S'\cup\{i\})$, and similarly to the other expressions, we get that \Cref{eq:tcSubMod} is equivalent to
    \begin{align*}
    \hc(j \mid S') - \hc(j \mid S' \cup \{i\})
    = \tc(j \cup S') - \tc(S')
    - \tc(j \cup S' \cup \{i\}) &+ \tc(S' \cup \{i\}) \\
    \ge 
    \underbrace{\hc(n+1 \mid S'\cup\{i,j\}) - \hc(n+1 \mid S' \cup \{i\}) - \hc(n+1 \mid S' \cup \{j\}) }_{\le 0} &+
    \hc(n+1 \mid S')
    \end{align*}
    The left-hand side of the equation is lower bounded by $\psi_{\tc}$. 
    Since $\hc(n+1 \mid T)$ weakly decreases with the cardinality of $T$, $\hc(n+1 \mid S'\cup\{i,j\}) - \hc(n+1 \mid S' \cup \{i\}) \le 0$. Also, the marginal cost of $n+1$ is non-negative, so $\hc(n+1 \mid S' \cup \{j\}) \ge 0$. Thus, we can upper bound the right-hand side with
    $\hc(n+1 \mid S') \le z/2 < \psi_{\tc}$, which concludes the proof.
    \end{itemize}
    
    \noindent\textbf{Monotonicity of $\hf$:} 
    Fix $S \subsetneq T \subseteq [n+1]$. Observe that regardless of whether $n+1 \in S$, it holds that
    \begin{align*}
            \hf(T) - \hf(S) 
            &\ge \hf(T \cap [n]) - \hf(S \cup \{n+1\}) \\
            &= f(T\cap [n]) - f(S\cap [n]) - \hf(n+1\mid S\cap[n]) \\
            &\ge \phi_{f} - \hf(n+1\mid S\cap[n]) \\ 
            &\ge z - z/4 \\
            &> 0
    \end{align*}
    
    \textbf{Submodularity of $\hf$:}
    We will show that for any $S \subseteq [n+1]$ and any distinct $i,j \notin S$,
    \begin{equation}\label{eq:fSubMod}
        \hf(S \cup \{i\}) + \hf(S \cup \{j\}) \ge \hf(S \cup \{i,j\}) + \hf(S)  
    \end{equation}

    \begin{itemize}
    \setlength\itemsep{-0.25em}
    \item If $n+1 \notin S \cup \{i,j\}$, then the claim follows from \Cref{obs:equalrev_subModf_prop}.
    \item If $i=n+1$ (and similarly for the case where $j=n+1$), by rearranging \Cref{eq:fSubMod} one can see that it is equivalent to,
    $\hf(n+1 \mid S) \ge \hf(n+1 \mid S \cup \{j\})$. The claim follows by the fact that $\hf(n+1 \mid T)$ is weakly decreasing with the cardinality of $T$.
    \item If $n+1 \in S$, let $S' = S \setminus \{n+1\}$. 
    Then be rewriting $\hf(S \cup \{i\})$ as $\hf(n+1 \mid S' \cup \{i\}) + \hf(S'\cup\{i\})$, and similarly to the other expressions, we get that 
    \Cref{eq:fSubMod} is equivalent to
    $$
    f(j \mid S')
    - f(j \mid S' \cup \{i\})
    \ge 
     \hf(n+1 \mid S'\cup\{i,j\}) - \hf(n+1 \mid S' \cup \{i\}) + \hf(n+1 \mid S') - \hf(n+1 \mid S' \cup \{j\}) 
    $$
    The left-hand side of the equation is lower bounded by $\psi_{f}$. 
    Since $\hf(n+1 \mid T)$ weakly decreases with the cardinality of $T$, $\hf(n+1 \mid S'\cup\{i,j\}) - \hf(n+1 \mid S' \cup \{i\}) \le 0$, and we can upper bound the left hand side with
    \begin{align*}
    \hf(n+1 \mid S') - \hf(n+1 \mid S' \cup \{j\}) 
    \le \hf(n+1 \mid S') \le z/4 < \psi_{f},
    \end{align*}
    which concludes the proof. \qedhere
    \end{itemize}
\end{proof}

\section{An Equal Revenue Instance for Additive $f$, Supermodular $c$}\label{sec:eqRev_SupMod_c}

In this section we construct an equal revenue instance as per \Cref{thm:eqRev_SupMod_c}.
\begin{theorem}\label{thm:eqRev_SupMod_c}
For every $n \in \mathbb{N}$, there exists an equal-revenue optimal contract instance, $\langle n, \feq, \ceq \rangle$, such that $\feq:2^n \to \reals_+$ is additive and $\ceq:2^n \to \reals_+$ is supermodular.
\end{theorem}

In order to specify our construction, we associate every subset of actions $S \subseteq [n]$ with an index $t \in \{0,1,...,2^n-1\}$, which is the decimal value of the $n$-bit characteristic vector of $S$, that is, $S_0 = \emptyset, S_1=\{1\}, S_2=\{2\}, S_3=\{1,2\},S_4=
\{3\}, S_5=\{1,3\},\ldots,S_{2^n-1} = \{1,...,n\}$. 
We set the reward for each action $i \in [n]$, to be $f_i = 2^{i-1}$, and $f(\emptyset)=0$. Note that the reward of the set $S_t$ is exactly its index, $t$. Namely,
$$
f(S_t) = \sum_{i \in S_t} 2^{i-1} = \sum_{i=1}^{n} I[i\in S_t] \cdot 2^{i-1} = t.
$$
Moreover, the difference in rewards between two consecutive sets $S_t$ and $S_{t-1}$ is exactly $f(S_t) - f(S_{t-1}) = t - (t-1) = 1$.
Let $\alpha_t$ be the critical value that incentivizes $S_t$. 
An equal revenue instance must satisfy the following requirements for any $t \in \{1,\dots,2^n-1\}$:
\begin{itemize}
    \setlength\itemsep{-0.25em}
    \item The critical value $\alpha_t$ yields a revenue of 1, i.e., $(1-\alpha_t)f(S_t) = 1$, or 
    equivalently, $\alpha_{t} = \frac{t-1}{t}$.   
    \item By \Cref{obs:critValStructure}, we have $\alpha_{t} = \frac{c(S_{t})-c(S_{t-1})}{f(S_{t})-f(S_{t-1})}=c(S_{t})-c(S_{t-1})$. 
\end{itemize}
By combining the two requirements and denote $c_t=c(S_t)$, we get the following recurrence relation, with the initial condition that $c_0 = 0$:
\begin{equation}\label{eq:costs_RR}
c_{t} = c_{t-1} + \frac{t-1}{t}.
\end{equation}

\begin{observation}\label{obs:alphaMon_supermodular_c}
    For every $t \in \{0,...,2^n-2\}$, $0 \le \alpha_t < \alpha_{t+1} < 1$.
\end{observation}

\begin{lemma}\label{lem:cMonoSuperMod}
     The cost function $c$ is strictly monotone and strictly supermodular.
\end{lemma}
\begin{proof}
    Clearly $c$ is monotone as $\frac{t}{t+1} > 0$ for any $t>0$.
    Supermodularity is implied by the fact that the discrete second derivative is positive:
    $
    c_{t+1}-c_t - (c_t-c_{t-1})
    =
    \frac{t}{t+1} - \frac{t-1}{t} > 0
    $
\end{proof}

Next we show that all $\alpha_t$ are critical values, and yield the same principal utility of $1$.
\begin{proposition}\label{prop:CVs_supc_c}
    The set of critical values is $0 = \alpha_1 < \ldots < \alpha_{2^n-1} < 1$ and for any $t \in \{0, \ldots, 2^n-2\}$, $S_t$ is the agent's best response for any contract $\alpha \in [\alpha_t, \alpha_{t+1})$. Additionally, each of the critical values yields a utility of $1$ for the principal.
\end{proposition}
\begin{proof}
Observe that for $\alpha_1=0$, $S_1=\{1\}$ is the agent's best response, as any other set has positive costs or zero reward.
Assume that $S_t$ is the agent's best response at $\alpha_t$, for some $t\ge 0$. By \Cref{obs:critVals}, the agent's best response can only change to a set with a higher reward, i.e., $S_{t+k}$ for some $k\ge1$. 
By \Cref{obs:critValStructure}, the next critical value satisfies
\begin{align*}
    \alpha 
    &=  
    \min_{k \in \{1,\ldots,2^n-t-1\}} \frac{c_{t+k} - c_{t}}{f_{t+k}-f_{t}} 
    =
    \min_{k \in \{1,\ldots,2^n-t-1\}}
    \frac{c_{t+k} - c_{t}}{k} 
    =
    \min_{k \in \{1,\ldots,2^n-t-1\}}
    \frac{\sum_{i=0}^{k-1}c_{t+i+1}-c_{t+i}}{k} \\
    &=
    \min_{k \in \{1,\ldots,2^n-t-1\}}
    \frac{1}{k}\sum_{i=0}^{k-1} \frac{t+i}{t+i+1} 
    \ge
    \min_{k \in \{1,\ldots,2^n-t-1\}}
    \frac{1}{k}\cdot k \cdot \frac{t}{t+1} 
    =
    \alpha_{t+1}.
\end{align*}
Since by definition $\alpha_{t+1}=c(S_{t+1})-c(S_t)$ is the contract for which the agent is indifferent between taking $S_t$ and $S_{t+1}$, it is indeed a critical value. Thus, inductively, every $\alpha_t$ is a critical value.

The fact that for $t \ge 0$, every contract $\alpha_t$ yields the same utility for the principal follows immediately from the definition of $f_t$,
$u_p(\alpha_t) = (1-\alpha_t)f_t = 1$, which concludes the proof.
\end{proof}

\section{Query Complexity Hardness for Additive $f$, Supermodular $c$}\label{sec:demandHardness_SupMod_c}

In this section we prove a query complexity hardness result when $f$ is additive and $c$ is supermodular.
This result is analogous to the hardness result for additive $c$ and submodular $f$, given in \Cref{thm:DemandHardnessForPerturbed}.

\begin{theorem}\label{thm:expDemandQueries_SupMod_c}
    When $f$ is additive and $c$ is supermodular, any algorithm that computes the optimal contract requires exponentially-many supply queries to $c$.
\end{theorem}
Our proof technique is identical to the one given in \Cref{sec:demandHardness}, where we also provide more detail.

\subsection{From Equal Revenue to Value Query Hardness}\label{sec:eqToVal_SupMod_c}
\newcommand{\perturbedInstSupModC}{\langle n,\feq,\ceq_k \rangle}
\newcommand{\familyISupModC}{\I = \{ \perturbedInstSupModC \}_{k=1}^{2^n}}

Given an equal-revenue instance $\eqRevInst$, denote the collection of incentivizable sets with $\IS = \{S_t\}_{t=1}^{2^n-1}$, let $\eps>0$ be such that,
\begin{align}\label{eq:eps_0_SupMod_c}
0 < 
\eps < \min \Bigg\{ 
& \min_{\substack{S \subseteq T \subseteq [n], \\ i \in [n]\setminus T}} (\ceq(i \mid T) - \ceq(i \mid S)), \nonumber \\
& \min_{t \in \{1,\dots,2^n-1\}} \left(\ceq(S_{t})-\ceq(S_{t-1})\right),  \\
& \min_{t \in \{2,\dots,2^n-1\}} (\alpha_t-\alpha_{t-1})(\feq(S_t) - \feq(S_{t-1})) \nonumber
\Bigg\}.
\nonumber
\end{align}

\begin{claim}
    The right-hand side of \Cref{eq:eps_0_SupMod_c} is strictly positive.
\end{claim}
\begin{proof}
    The fact that the first term is positive is proved in \Cref{cla:eps_0_positive}. 
    The two other terms are positive by the fact that there $2^n-1$ distinct critical values, each incentivizes a different set.
    By \Cref{obs:critValStructure}, the costs and rewards of the incentivizable sets are strictly increasing.
\end{proof}

The following lemma is analogous to \Cref{lem:f_strictly_submodolar} and its proof is identical, except for a few necessary changes
\begin{lemma}\label{lem:c_eqrev_strictly_supermodolar}
    Let $\eqRevInst$ be an equal revenue instance with additive $\feq$ and supermodular $\ceq$, then $\ceq$ is strictly supermodular. 
\end{lemma}

We define a family of perturbed cost functions, where $\ceq_k$ is identical to $\ceq$ everywhere, except it gives an additive ``bonus" of $\eps$ to the set $S_k$. Namely,
$$
\ceq_k(S_t) = 
\begin{cases}
\ceq(S_t) - \eps & t=k, \\
\ceq(S_t) & t \ne k.
\end{cases}
$$
This completely defines the collection of instances $\familyISupModC$, that all share the same reward function $\feq$.

The following propositions are analogous to Propositions \ref{prop:feq_kSubmod} and \ref{prop:feq_k_CVs}. 
They establish that the perturbed instances keep two essential properties of the equal revenue instance. Namely, that (i) $\ceq_k$ is monotone and supermodular, and (ii) $\langle n, \feq, \ceq_k\rangle$ has $2^n-1$ critical values and $\alpha_k$ is the unique optimum. 
For readability, we defer their proofs to the end of this subsection.

\begin{proposition}\label{prop:ceq_kSupMod}
    Let $\ceq$ be non-negative, monotone and supermodular function, and let $\eps$ satisfy \Cref{eq:eps_0_SupMod_c}.
    Then for any $k$, $\ceq_k$ is non-negative monotone and supermodular.
\end{proposition}

\begin{proposition}\label{prop:ceq_k_CVs}
    For any $k$, let $\langle n, \feq, \ceq_k \rangle$ be a perturbed instance with $\eps$ satisfying \Cref{eq:eps_0_SupMod_c}.
    The it has $2^n-1$ distinct critical values $\{\alpha'_t\}_{t=1}^{2^n-1}$. Moreover, $\alpha'_k$ is the unique optimal contract.
\end{proposition}

\begin{proposition}\label{prop:valueQueryBound_SupMod_c}
Let $I^* = \langle n, \feq, \ceq^* \rangle \in \I = \{\langle n, \feq, \ceq_k \}_{k=1}^{2^n-1}$.
Any algorithm that is only given access to a value oracle for $c^*$, must perform an exponential number of value queries (in expectation) to find the optimal contract.
\end{proposition}
\begin{proof}
    By Yao's principle, it is enough to show that when we draw $\langle n, \feq, \ceq^* \rangle \in \I$ uniformly at random, any deterministic algorithm requires an exponential number of value queries in expectation over the choice of the instance.

    First, observe that every value query $\ceq^*(S_t)$ can lead to two different answers. If $\ceq^* = \ceq_{t}$, then $\ceq^*(S_t) = \ceq(S_t) - \eps$, otherwise $\ceq^*(S_t) = \ceq(S_t)$.
    
    Thus, without loss of generality, every deterministic algorithm makes a \emph{fixed} series of value queries and stops whenever the answer is $\ceq^*(S_t) = \ceq(S_t) - \eps$. If it stops beforehand, and makes less than $2^n-2$ value queries, there exist two sets $S_t,S_{t'}$ which the algorithm did not query and are both consistent with the oracle answers given. Thus, the algorithm cannot distinguish between the case in which $\alpha_t$ is optimal and the case in which $\alpha_{t'}$ is.
    It follows that any deterministic algorithm that makes $k$ value queries only stops for $k$ different choices of $\ceq^*$. 
    
    Fix a deterministic algorithm for the optimal contract. Since $\ceq^*$ is drawn uniformly at random, with probability $1/2$, the first $2^{n-1}$ queries, $S_{t_1}, \ldots ,S_{t_{2^{n}-1}}$, satisfy 
    $\ceq^*(S_{t_i}) = \ceq(S_{t_i})$, and the algorithm does not stop.
    Thus, the expected number of queries of any deterministic algorithm is at least $\frac{1}{2} \cdot 2^{n-1} = 2^{n-2}$.
\end{proof}

\begin{proof}[Proof of \Cref{prop:ceq_kSupMod}]
    Fix $k$. Non-negativity follows as $\eps > 0$. For supermodularity, observe that for any two sets $S \subseteq T \subseteq [n]$ and any action $ i \in [n]\setminus T$, it holds that
    \begin{align*}
    \ceq_k(i \mid T) - \ceq_k(i \mid S) \ge \ceq(i \mid T) - \ceq(i \mid S) - \eps > 0,
    \end{align*}
    where the first inequality follows from the definition of $\feq_k$ and last follows from \Cref{eq:eps_0_SupMod_c}.
    For monotonicity, observe that for every two sets of actions, $S \subsetneq T$, we have
    $$\ceq_k(T) - \ceq_k(S) \ge \ceq(T) - \ceq(S) - \eps > 0,$$
    where the last inequality follows from \Cref{eq:eps_0_SupMod_c}.
\end{proof}

\begin{proof}[Proof of \Cref{prop:ceq_k_CVs}]
    First observe that for any $t < k$ and any $t>k+1$, $\alpha'_t = \alpha_t$, as
    $$
    \alpha'_t 
    =
    \frac{\ceq(S_t) - \ceq(S_{t-1})}{\feq_k(S_t) - \feq_k(S_{t-1})} 
    =
    \frac{\ceq(S_t) - \ceq(S_{k-1})}{\feq(S_t) - \feq(S_{t-1})} 
    =
    \alpha_t.
    $$  
    For $t=k+1$ we have for any $\eps >0$,
    \begin{align*}
        \alpha'_{k+1}  - \alpha'_k
        &=
        \frac{\ceq_k(S_{k+1}) - \ceq_k(S_{k})}{\feq(S_{k+1}) - \feq(S_{k})} 
        - 
        \frac{\ceq_k(S_{k}) - \ceq_k(S_{k-1})}{\feq(S_{k}) - \feq(S_{k-1})} \\
        &=
        \frac{\ceq(S_{k+1}) - \ceq(S_{k}) + \eps}{\feq(S_{k+1}) - \feq(S_{k})} 
        - 
        \frac{\ceq(S_{k}) - \eps - \ceq(S_{k-1})}{\feq(S_{k}) - \feq(S_{k-1})} \\
        &>
        \frac{\ceq(S_{k+1}) - \ceq(S_{k})}{\feq(S_{k+1}) - \feq(S_{k})} 
        - 
        \frac{\ceq(S_{k}) - \ceq(S_{k-1})}{\feq(S_{k}) - \feq(S_{k-1})} \\
        &=
        \alpha_{k+1}  - \alpha_k \\
        &>
        0.
    \end{align*}
    For $t=k$ we have,
    \begin{align*}
        \alpha'_k  - \alpha'_{k-1}
        &=
        \frac{\ceq_k(S_k) - \ceq_k(S_{k-1})}{\feq(S_k) - \feq(S_{k-1})} 
        - 
        \frac{\ceq_k(S_{k-1}) - \ceq_k(S_{k-2})}{\feq(S_{k-1}) - \feq(S_{k-2})} \\
        &=
        \frac{\ceq(S_k) - \eps - \ceq(S_{k-1})}{\feq(S_k) - \feq(S_{k-1})} 
        - 
        \frac{\ceq(S_{k-1}) - \ceq(S_{k-2})}{\feq(S_{k-1}) - \feq(S_{k-2})} \\
        &=
        \alpha_k - \alpha_{k-1}
        -\frac{\eps}{\feq(S_k) - \feq(S_{k-1})} \\
        &> 0
    \end{align*}
    Where the last inequality follows from \Cref{eq:eps_0_SupMod_c}
    We conclude that there are $2^n-1$ distinct critical values, and $\alpha'_t$ incentivizes $S_t$.

    We now proceed to show that $\alpha'_k$ is the unique optimal contract for the instance $\langle n, \feq, \ceq_k \rangle$.
    As for any $t \ne k, k+1$, it holds that $\alpha_t = \alpha'_{t}$, we get that 
    $$
    u_p(\alpha'_t) = (1-\alpha'_t)\feq(S_t) =(1-\alpha_t)\feq(S_t) = 1.
    $$
    For $t = k+1$ we have,
    \begin{align*}
    u_p(\alpha'_{k+1})
    &=
    \left(1-\alpha'_{k+1})\feq(S_{k+1}\right) \\
    &=
    \left(1-\frac{\ceq_k(S_{k+1}) - \ceq_k(S_{k})}{\feq(S_{k+1}) - \feq(S_{k})}\right)\feq(S_{k+1}) \\
    &=
    \left(1-\frac{\ceq(S_{k+1}) - \ceq(S_{k}) + \eps}{\feq(S_{k+1}) - \feq(S_{k})}\right)\feq(S_{k+1}) \\
    &=
    \left(1-\alpha_{k+1} - \frac{\eps}{\feq(S_{k+1}) - \feq(S_{k})}\right)\feq(S_{k+1}) \\
    &<
    (1-\alpha_{k+1})\feq(S_{k+1}) \\
    &=
    1.
    \end{align*}
    For $t = k$,
    \begin{align*}
    u_p(\alpha'_{k})
    &=
    \left(1-\alpha'_{k})\feq(S_{k}\right) \\
    &=
    \left(1-\frac{\ceq_k(S_{k}) - \ceq_k(S_{k-1})}{\feq(S_{k}) - \feq(S_{k-1})}\right)\feq(S_{k}) \\
    &=
    \left(1-\frac{\ceq(S_{k}) - \eps - \ceq(S_{k-1})}{\feq(S_{k}) - \feq(S_{k-1})}\right)\feq(S_{k}) \\
    &=
    (1-\alpha_k)\feq(S_k) + \frac{\eps \cdot \feq(S_{k})}{\feq(S_{k}) - \feq(S_{k-1})} \\
    &>
    1.
    \end{align*}   
    This completes the proof.
\end{proof}

\subsection{Extending Hardness to Supply Queries}\label{sec:SparsetoDemand_SupMod_c}

We show how to extend our hardness result to the supply query model. The argument follows the footsteps of \Cref{sec:SparsetoDemand} and goes through the notion of \textit{sparse supply}, analogous to sparse demand (see \Cref{def:sparseDemand}).

\begin{definition}\label{def:sparseSupply}[$\sigma$-Sparse Supply]
For $\sigma >0$, we say that $c: 2^{[n]} \rightarrow \reals_+$ has $\sigma$-sparse supply for every price vector $p$, if $|\supSigP|$ is polynomial in $n$, where $\supSigP$ is the collection of sets of actions that maximize the quasi-linear utility, w.r.t. $f$ and $p$, up to an additive factor of $\sigma$:
$$
\supSigP = \left\{ S\subseteq [n] ~\middle|~ \max_{T\subseteq [n]} \left(\sum_{i\in T} p_i - c(T)\right) - \left(\sum_{i\in S} p_i - c(S) \right) \le \sigma \right\}.
$$
We refer to $\supSigP$ as the $\sigma$-approximate supply. 
Observe that for any $0 < \sigma' \le \sigma$, $\sigma$-sparse supply implies $\sigma'$-sparse supply, as $\supply{\sigma'}{p} \subseteq \supSigP$.
\end{definition}

The theorem below is the analog for \Cref{thm:DemandHardnessForPerturbed}, and the proof is essentially identical.
\begin{theorem}\label{thm:SupplyHardnessForPerturbed}
    Let $\eqRevInst$ be an equal revenue instance such that $\ceq$ has sparse supply. Let $\familyISupModC$ be the perturbed family of instances as defined in \Cref{sec:eqToVal_SupMod_c}.
    Any algorithm that finds the optimal contract for any instance in $\familyISupModC$, must make  exponentially-many supply queries in expectation.
\end{theorem}

\begin{proof}
    Let $\eqRevInst$ be an equal revenue instance where $\ceq$ is supermodular and has sparse supply.
    Recall the family of perturbed instances, $\I = \{\langle n,\feq,\ceq_k\rangle\}_{k=1}^{2^n-1}$ which is defined in \ref{sec:eqToVal_SupMod_c}. 
    Namely, $\ceq_k(S_t) = \ceq(S_t)$ for any $t \ne k$ and for $t=k$, $\ceq_k(S_k) = \ceq(S_k) - \eps$, for some $\eps>0$ which satisfies \Cref{eq:eps_0_SupMod_c} and also $\eps \le \sigma$.
    
    Consider a specialized algorithm $\A$, which finds the optimal contract for all instances in $\I$ and is augmented with complete knowledge of $\ceq$, the reward function of the ``original" equal revenue instance.
    Clearly, showing that $\A$ requires exponentially many supply queries in expectation suffice.
    
    Let $\langle n,\feq,\ceq^\star\rangle \in \I$, be a perturbed instance drawn uniformly at random.
    First, observe that $\A$ can compute a supply query for $\ceq^\star$ using poly-many value queries (and perhaps exponentially-many computational steps):
    For any price vector $p$, if $S_p$ should be returned by the supply oracle, then 
    $\sum_{i \in S^\star}p_i - \ceq^\star(S_p) \ge \sum_{i \in S} p_i - \ceq^\star(S)$ for any $S \subseteq [n]$, which implies that,
    $$
    \sum_{i \in S^\star} p_i - \ceq(S_p) + \eps
    \ge
    \sum_{i \in S^\star} p_i - \ceq^\star(S_p)
    \ge
    \sum_{i \in S} p_i - \ceq^\star(S)
    \ge
    \sum_{i \in S} p_i - \ceq(S)
    $$
    and as such, $S_p \in \supply{\eps}{p}$. 
    Thus, for any price vector $p$, every set in the which maximize the supply with respect to $\ceq^\star$ must also belong to $\supply{\eps}{p}$ with respect to $\ceq$.
    
    Because the approximate supply $\supply{\eps}{p}$ does not depend on the identity of $\ceq^\star$, it can be computed without any value queries -- but perhaps using a super-polynomial number of computational steps.
    Thus, to solve the supply query we can pick the best set in the collection $\supply{\eps}{p}$ using $|\supply{\eps}{p}|\le |\supSigP|$ value queries. As $\ceq$ has $\sigma$-sparse demand, the number of queries is polynomial.
    Thus, if $\A$ is able to find the optimal contract with poly-many supply queries, it can also do it with poly-many value queries, due to the above argument. This will contradict \Cref{prop:valueQueryBound_SupMod_c}, and the claim follows.
\end{proof}

\subsection{An Equal Revenue Instance with Sparse Supply}\label{sec:eqRevWithSparse_SupMod_c}

In this section we complete the proof of \Cref{thm:expDemandQueries_SupMod_c}, by showing that the equal-revenue instance described has sparse supply. 
The argument is mostly identical to the one given in \Cref{sec:eqRevWithSparse}, but slight changes are required in the proof of the following lemma, which is analogous to \Cref{lem:RLintervals}.
 \begin{lemma}\label{lem:RLintervals_SupMod_c}
    For any price vector $p \in \reals_{+}^n$, any 
    $0 < \sigma < \min_{l < h} \frac{1}{2}(\alpha_h - \alpha_l)$,
     any action $i \in [n]$ and any $S_t \in \supSigP$, if $t > r_{i,p}$, then $i \notin S_t$. If $t < l_{i,p}$, then $i \in S_t$.
\end{lemma}
\begin{proof}
    The first claim is immediate from the definition of $r_{i,p}$.
    To prove the second claim we show that there exists $\sigma > 0$ small enough such that for any price vector $p \in \reals_{+}^n$ and any two sets $S_t,S_{t'} \subseteq [n]$ such that $i \in S_{t'} \setminus S_t$ and $t < t' - 2^i$, it holds that $S_t$ and $S_{t'}$ cannot simultaneously be approximately optimal (i.e., $\{S_t,S_{t'}\} \not\subseteq \supSigP$).
    Since by definition $S_{r_{i,p}} \in \supSigP$, replacing $t'$ with $r_{i,p}$, proves the claim for any $t < r_{i,p} - 2^i = l_{i,p}$.

    Consider some $\sigma > 0$ and fix $p, S_t, S_{t'}$, and $i$ as above. 
    Since $S_{t'} \setminus \{i\}$ is the agent's best response for contract $\alpha_{S_{t'}\setminus \{i\}}$, it holds that
    $$
    \alpha_{S_{t'}\setminus\{i\}} \cdot 2^{i-1}
    =
    \alpha_{S_{t'}\setminus\{i\}} \cdot f(i \mid S_{t'} \setminus\{i\}) < c_i
    $$
    So if $p_i < \alpha_{S_{t'}\setminus\{i\}} \cdot 2^{i-1}  - \sigma$ we get that $S_{t'} \notin \supSigP$ and the claim holds.
    
    Similarly, $\alpha_{S_t \cup \{i\}} \cdot 2^{i-1}=\alpha_{S_t \cup \{i\}} f(i \mid S_t) \ge c_i$, and if $p_i > \alpha_{S_t \cup \{i\}} \cdot 2^{i-1} + \sigma$ we have that $S_t \notin \supSigP$ and the claim holds.

    Thus, it suffices to show that
    $$
    \alpha_{S_{t'}\setminus\{i\}} \cdot 2^{i-1}  - \sigma
    \ge
    \alpha_{S_t \cup \{i\}} \cdot 2^{i-1} + \sigma
    $$
    First, observe that by assumption $t' - 2^{i-1} > t + 2^{i-1}$, so the index of the set $S_{t'} \setminus \{i\}$ is greater than the index of the set $S_t \cup \{i\}$.    
    Then, by \Cref{lem:alphaMon} it holds that $\alpha_{S_{t'}\setminus \{i\}} > \alpha_{S_t\cup \{i\}}$. 
    By our choice of $\sigma$,
    \begin{align*}
        \sigma 
        &<
        \min_{l < h} \frac{1}{2}(\alpha_h - \alpha_l)
        \le 
        \frac{1}{2}(\alpha_{S_{t'}\setminus \{i\}} - \alpha_{S_{t}\cup \{i\}}) \le 2^{i-2}(\alpha_{S_{t'}\setminus \{i\}} - \alpha_{S_{t}\cup \{i\}}),
    \end{align*}
    and the claim follows.
\end{proof}

The rest of the argument is identical to the one presented in \Cref{sec:eqRevWithSparse}: using \Cref{lem:RLintervals_SupMod_c} one can use the notion of minimal ambiguous action (see \Cref{def:minAmgAction}) to establish an upper bound of $O(n^2)$ on the size of the approximate-supply.

\section{Communication Complexity Hardness for Submodular $f$ and Supermodular $c$}\label{sec:cc_lower_bound_submod_f_supmod_c}
In this section we show the communication complexity lower bound for submodular $f$ and supermodular $c$, analogous to the result of \Cref{sec:CC_submod_fc} for submodular $f$ and $c$.

\begin{theorem}[CC lower bound, submodular $f$ and supermodular $c$]\label{thm:cc_submod_f_supmod_c}
    When $f$ is submodular and $c$ is supermodular, any protocol which computes the optimal contract and makes $poly(n)$ best-response queries, requires $\Omega(2^n/\sqrt{n})$ bits of communication.
\end{theorem}

The argument is very similar to the one given in \Cref{subsec:cc_lower_bound_submod_fc}. We again start from an equal revenue instance with submodular $f$ and additive $c$, per \Cref{sec:eqRev_SupMod_c}.
We perturb the costs to create a \emph{supermodular} cost function $\tc$, and then adding another action whose marginal rewards and costs are parametrized by two binary vectors of length $\binom{n/2}{n}$, indicating ``special'' sets. 
Although the proof is similar, some details must be adjusted. In particular, the marginal reward and cost of action $n+1$ require a more careful handling than in the previous construction.

We describe how to perturb the costs of the equal revenue instance so that the new instance $\InstFTildeI$ has submodular rewards and supermodular costs, $2^n-1$ critical values and sparse best-response.
As in the analogous sections, we augment the instance with an additional action whose marginal depend on two vectors $x_f,x_c \in \{0,1\}^{\binom{n}{n/2}}$ such that (i) rewards and costs are supermodular, and (ii) if Alice and Bob know $\tI$, they can compute a best-response query using $poly(n)$ communication.

\subsection{Making $c$ Strictly Supermodular}
The perturbed cost function is defined as follows, $\tc(S)=c(S) + \delta \cdot |S|^2$, where $\delta$ satisfies
\begin{align}\label{eq:eps_for_tc_supermod}
0 < 
\delta < \min \Bigg\{ 
& 
\frac{1}{n^2}(1-\alpha_{2^n-1})(f(S_{2^n-1})-f(S_{2^n-2})), \\
& \min_{t \in \{1,\dots,2^n-2\}} (\alpha_{t+1}-\alpha_t)\cdot \frac{1}{n^2} \cdot \left(\frac{(f(S_t)-f(S_{t-1}))(f(S_{t+1})-f(S_t))}{(f(S_{t+1})-f(S_{t-1}))}\right) \nonumber \bigg\}
\end{align}

The next two propositions establish the basic properties of the cost function $\tc$. Namely, that for a proper choice of $\delta$, $\tc$ is non-negative, monotone, supermodular, and additionally, the instance $\inst{n}{f}{\tc}$ has $2^n-1$ critical values.

\begin{proposition}\label{prop:tc_nonNeg_monotone_supMod}
For any $\delta > 0$, $\tc$ is non-negative, strictly monotone and strictly supermodular.
\end{proposition}

\begin{proof}
\textbf{Non-negativity:} Follows immediately from the non-negativity of $c$ and the fact that $\delta >0$.
\textbf{Monotonicity:}
For any $S \subsetneq T$,
$$
\tc(T) - \tc(S) = c(T)-c(S) + \delta(|T|^2-|S|^2) > 0
$$
\textbf{Supermodularity:}
For any $S \subsetneq T$ and $i \notin T$,
\begin{align*}
\tc(i \mid T) - \tc(i \mid S)
&=
c(i \mid T) + \delta((|T|+1)^2-|T|^2) - c(i \mid S) - \delta((|S|+1)^2-|S|^2) +  \\
&\ge 
\delta 
[(|T|+1)^2-|T|^2) - (|S|+1)^2-|S|^2)] \\
&>0,  
\end{align*}
where the last inequality follows from the strict convexity of $x^2$.
\end{proof}

\begin{proposition}\label{prop:exp_many_cvs_maintained_SupMod_c}
Let $\delta > 0$ be such that \Cref{eq:eps_for_tc_supermod} holds, then, $\inst{n}{f}{\tc}$ has $2^n-1$ critical values, $\{\talpha_t\}_{t=1}^{2^n-1}$.
\end{proposition}

\begin{proof}
First, observe that
$\talpha_{1} = \frac{\tc_{1} - 0}{f_1 - 0} =\frac{c_1 + \delta n^2}{f_1} > \alpha_1 > 0$.

The same analysis as in the proof of \Cref{prop:exp_many_cvs_maintained} gives that 
picking
$$
0 < \delta < (\alpha_{t+1}-\alpha_t)\cdot \frac{1}{n^2} \cdot \left(\frac{(f_t-f_{t-1})(f_{t+1}-f_t)}{(f_{t+1}-f_{t-1})}\right)
$$
is sufficient to guarantee that $\talpha_{t+1} > \talpha_t$, for any $t\in \{2,\dots,2^n-2\}$.
Similarly, picking 
$$
0 < \delta < \frac{(1-\alpha_{2^n-1})(f_{2^n-1}-f_{2^n-2})}{n^2}
$$
ensures that $\talpha_{2^n-1} <1$.
The claim follows as each $\talpha_t$ corresponds to a contract for which the agent change its preferred set of actions.
\end{proof}

\subsubsection{Sparse Best-Response}
\begin{proposition}\label{prop:sparse_demand_submod_f_sup_c}
    If the instance $\instanceI$ has $\sigma$-sparse demand for some $\sigma > 0$, then for any $0 < \delta < \frac{\sigma}{2n^2}$, the perturbed instance $\InstTildeI$ has $(\sigma/2)$-sparse best-response.
\end{proposition}
\begin{proof}
\newcommand{\sa}{S_\alpha}
\newcommand{\tsa}{\tilde{S}_\alpha}
    
    Fix a contract $\alpha$ and some $0 < \delta < \frac{\sigma}{2n^2}$. We show that 
    $\BR{\sigma/2}{\alpha}{\tI} \subseteq \BRI{\sigma/2+n^2\delta}{\alpha} \subseteq \BRI{\sigma}{\alpha}$, which implies the claim.
    Let $\sa$ be the set of actions which maximizes $\alpha f(S) - c(S)$ and let $\tsa$ be the set of which maximizes $\alpha f(S) - \tc(S)$.
    Observe that for any $S \in \BR{\sigma/2}{\alpha}{\tI}$
    it holds that,
    \begin{align*}
        \alpha f(S) - \tc(S) 
        &\ge \alpha f(\tsa) - \tc(\tsa) - \sigma/2 \\
        &\ge \alpha f(\sa) - \tc(\sa) - \sigma/2 \\
        &\ge \alpha f(\sa) - c(\sa) -\delta n^2 - \sigma/2,
    \end{align*}
    where the last inequality follows from the fact that $\tc(\sa) = c(\sa) + \delta|\sa|^2$.
    On the other hand, 
    $$\alpha f(S) - \tc(S) = \alpha f(S) - c(S) - \delta|S|^2 \le \alpha f(S) - c(S),$$ 
    we conclude that 
    \[
    \alpha f(S) - c(S) \ge \alpha f(\sa) - c(\sa) - (\sigma/2 + \delta n^2),
    \]
    which implies that $S \in \BRI{\sigma/2+n^2\delta}{\alpha}$.
    By our choice of $\delta$, it holds that $\sigma/2+n^2\delta \le \sigma$, and thus $\BRI{\sigma/2+n^2\delta}{\alpha} \subseteq \BRI{\sigma}{\alpha}$, which concludes the proof.
\end{proof}

\subsection{Adding the ($n+1$)th action}

As in \Cref{subsec:cc_lower_bound_submod_fc}, we 
define the marginal rewards and costs of the ($n+1$)th action, to construct a hard instance.
We define $z = \min \{\phi_{\tc},\psi_{\tc},\phi_f,\psi_f, \zeta, \sigma \}$, where the definitions are as in \Cref{subsec:cc_lower_bound_submod_fc}, except for 
$\psi_{\tc} = \min_{S\subseteq T, i \notin T} \tc(i \mid T)-\tc(i \mid S) > 0$, which is adjusted to the supermodular case.

Observe that the same analysis as in \Cref{prop:eps_for_revenue_bound} still holds, thus we can use the same bounds on the revenue generated by any critical value in the instance $\inst{n}{f}{\tc}$

\begin{proposition}\label{prop:eps_for_revenue_bound_supMod_c}
For any $0 < z < \zeta = \frac{\delta \cdot (1-\alpha_{2^n-1}) \phi_f}{16\cdot n^2 \cdot f_{2^n-1}}$, for any $t \in [2^n-1]$, 
$$
(1-\talpha_t)f(S_t) \in \left[1 - \frac{z(1-\alpha_{2^n-1})}{16}, 1 + \frac{z(1-\alpha_{2^n-1})}{16} \right]
$$
\end{proposition}

For any $t$, if $|S_t| < \frac{n}{2}$, let $h(t)$ be the minimal index such that $S_t \subseteq S_{h(t)}$ and $|S_{h(t)}|=\frac{n}{2}$. Otherwise, $h(t)=t$.

The marginal rewards and costs of the $n+1$ action are as follows. Note that in the definition we use $\talpha_t$'s --  the critical values of the instance $\inst{n}{f}{\tc}$, as in \Cref{prop:exp_many_cvs_maintained_SupMod_c}.

$$
\hf(n+1 \mid S_t) = 
\begin{cases}
    z/4 & |S_t| < \frac{n}{2} \\
    z/4 & |S_t| = \frac{n}{2} \land S_t \in x_f \\
    0 & |S_t| = \frac{n}{2} \land S_t \notin x_f \\
    0 & |S_t| > \frac{n}{2}
\end{cases}
\qquad 
\hc(n+1 \mid S_t) = 
\begin{cases}
    \talpha_{h(t)} \cdot z/4 & |S_t| < \frac{n}{2} \\
    \talpha_{t} \cdot z/4 & |S_t| = \frac{n}{2} \land S_t \in x_c \\
    z/2 & |S_t| = \frac{n}{2} \land S_t \notin x_c \\
    z/2 & |S_t| > \frac{n}{2}
\end{cases}.
$$

\begin{claim}\label{cla:n+1_supermod}
    Let $S \subseteq T \subseteq [n]$, then $\hc(n+1 \mid S) \le \hc (n+1 \mid T)$.
\end{claim}
\begin{proof}
    Denote $S_t = S$ and $S_{t'} = T$, and observe that $t' >t$.
    \begin{itemize}
        \setlength\itemsep{-0.25em}
        \item If $|S_{t'}| > n/2$ or $|S_{t'}| = n/2$ and $S_{t'} \notin x_c$, then $\hc(n+1 \mid S_{t'}) = z/2 \ge \hc(n+1 \mid S_t)$.
        \item If $|S_{t'}|=n/2$ and $S_{t'} \in x_c$, 
        observe that  $|S_t| < n/2$, and that by definition $h(t) \le t'$. Thus,
        $\hc(n+1 \mid S_{t'}) = \talpha_{t'}\cdot z/4 \ge \talpha_{h(t)} \cdot z/4 \ge \hc(n+1 \mid S_t)$.
        \item If $|S_{t'}|<n/2$,
        because any set that contains $S_{t'}$ also contains $S_t$, we have that $h(t) \le h(t')$, and,
        $\hc(n+1 \mid S_{t'}) = \talpha_{h(t')}\cdot z/4 \ge \talpha_{h(t)} \cdot z/4 \ge \hc(n+1 \mid S_t)$. \qedhere
    \end{itemize}
\end{proof}

\begin{proposition}\label{prop:hc_monotone_supMod}
    $\hc$ is monotone and supermodular.
\end{proposition}

\begin{proof}
Observe that $\hf$ is identical to the one presented in \Cref{subsec:cc_lower_bound_submod_fc}, thus its monotonicity and submodularity follow from \Cref{prop:monotone_submod_hf_hc}. It remains to show that $\hc$ satisfies the above properties.
    \textbf{Monotonicity:} 
    Fix $S \subsetneq T \subseteq [n+1]$. 
    \begin{align*}
            \hc(T) - \hc(S) 
            &\ge \hc(T \cap [n]) - \hc(S \cup \{n+1\}) 
            = \tc(T\cap [n]) - \tc(S\cap [n]) - \hc(n+1\mid S\cap[n]) \\
            &\ge \phi_{\tc} - \hc(n+1\mid S\cap[n]) 
            \ge z - z/2 
            > 0
    \end{align*}
    First inequality follows as the marginal cost of $n+1$ is non-negative.
    
    \textbf{Supermodularity:} 
    We will show that for any $S \subseteq [n+1]$ and any $i,j \notin S$, where $i \ne j$, it holds that 
    \begin{equation}\label{eq:tcSupMod}
        \hc(S \cup \{i\}) + \hc(S \cup \{j\}) \le \hc(S \cup \{i,j\}) + \hc(S)  
    \end{equation}

    \begin{itemize}
    \setlength\itemsep{-0.25em}
    \item If $n+1 \notin S \cup \{i,j\}$, then the claim follows from \Cref{prop:tc_nonNeg_monotone_supMod}.

    \item If $i=n+1$ (and similarly for the case where $j=n+1$), by rearranging \Cref{eq:tcSupMod} one can see that it is equivalent to,
    $\hc(n+1 \mid S) \le \hc(n+1 \mid S \cup \{j\})$, and the claim follows from \Cref{cla:n+1_supermod}.

    \item If $n+1 \in S$, let $S' = S \setminus \{n+1\}$. 
    Then be rewriting $\hc(S \cup \{i\})$ as $\hc(n+1 \mid S' \cup \{i\}) + \hc(S'\cup\{i\})$, and similarly to the other expressions, we get that 
    \Cref{eq:tcSupMod} is equivalent to
    \begin{align*}
    \hc(j \mid S' \cup \{i\}) - \hc(j \mid S') 
    = \tc(j \cup S') - \tc(S')
    - \tc(j \cup S' \cup \{i\}) &+ \tc(S' \cup \{i\}) \\
    \ge 
    \underbrace{\hc(n+1 \mid S' \cup \{i\})-\hc(n+1 \mid S'\cup\{i,j\})}_{\le 0 \;(\Cref{cla:n+1_supermod})} +\hc(n+1 \mid S' \cup \{j\}) &-
    \hc(n+1 \mid S')
    \end{align*}
    The left-hand side of the equation is lower bounded by $\psi_{\tc}$. 
    Since $\hc(n+1 \mid T)$ weakly decreases with the cardinality of $T$, $\hc(n+1 \mid S'\cup\{i,j\}) - \hc(n+1 \mid S' \cup \{i\}) \le 0$. Also, the marginal cost of $n+1$ is non-negative, so $\hc(n+1 \mid S' \cup \{j\}) \le 0$. Thus, we can upper bound the right-hand side with
    $\hc(n+1 \mid S') \le z/2 < \psi_{\tc}$, which concludes the proof.
    \end{itemize}
\end{proof}

\begin{observation}
    If a set of the form $S_t \cup \{n+1\}$, for some $S_t \subseteq [n]$, can be incentivized by the principal, then $|S_t| = n/2$ and $S_t \in x_c \cap x_f$.
\end{observation}
\begin{proof}
    If $|S_t| > n/2$ or $|S_t| = n/2$ and $S_t \notin x_c \cap x_f$, then clearly the $n+1$ action yields negative utility to the agent for any contract $\alpha \in [0,1]$, and cannot be incentivized.

    If $|S_t|<n/2$, then it order for $S_t \cup \{n+1\}$ to be picked by the agent for some contract $\alpha$, it must yield more utility than $S_t$, i.e., it most hols that
    $\alpha \ge \frac{\hc(\{n+1\} \mid S_t)}{\hf(\{n+1\} \mid S_t)} = \talpha_{h(t)}$.

    We next show that for any $\alpha \ge \talpha_{h(t)}$, $S_{h(t)}$ yields more utility to the agent than $S_t \cup \{n+1\}$. Indeed,
    \begin{align*}
    \alpha
    f(S_t \cup \{n+1\})-c(S_t \cup \{n+1\}) 
    &=
    \talpha_{h(t)} 
    f(S_t \cup \{n+1\})-c(S_t \cup \{n+1\}) + (\alpha-\talpha_{h(t)})f(S_t \cup \{n+1\})\\
    &=
    \talpha_{h(t)} f(S_t) + \frac{z \cdot \talpha_{h(t)}}{4} - c(S_t) - \frac{z\cdot \talpha_{h(t)}}{4} + (\alpha-\talpha_{h(t)})f(S_t \cup \{n+1\})\\
    &=
    \talpha_{h(t)} f(S_t)-c(S_t) + (\alpha-\talpha_{h(t)})(f(S_t) + z/4)\\
    &<
    \talpha_{h(t)} f(S_{h(t)})-c(S_{h(t)}) + (\alpha-\talpha_{h(t)})f(S_{t+1}), \\
    &\le
    \talpha_{h(t)} f(S_{h(t)})-c(S_{h(t)}) + (\alpha-\talpha_{h(t)})f(S_{h(t)}), \\
    &=
    \alpha f(S_{h(t)})-c(S_{h(t)}).
    \end{align*}
    The first inequality follows from the fact that $0 < z < \phi_{f}$ and $S_{h(t)}$ is the agent's best-response for $\talpha_{h(t)}$. 
    The second inequality holds by the fact that $h(t) \ge t+1$ whenever $|S_t|<n/2$.
    Since the agent breaks ties in favor of the set with the larger reward, $S_{h(t)}$ is preferred over $S_t$ whenever $\alpha \ge \talpha_{h(t)}$, which implies the claim.
\end{proof}

\begin{proposition}\label{prop:n+1_optimal_sub_f_sup_c}
    Let $S_t$ be the set with the minimal index such that $S_t \in x_c \cap x_f$. 
    The contract $\halpha_t = \frac{\hc(S_t)-\hc(S_{t-1})}{\hf(S_t)-\hf(S_{t-1})}$ incentivizes the set $S_t \cup \{n+1\}$ and
    the principal's revenue from $\halpha_t$, $u_p(\halpha_t,S_t\cup\{n+1\})$, exceeds the revenue from any incentivizable set of action of the form $S_{t'} \subseteq [n]$.
\end{proposition}
The proof of \Cref{prop:n+1_optimal_sub_f_sup_c} is identical to the proof of \Cref{prop:n+1_is_optimal}, except it uses the bounds established in \Cref{prop:exp_many_cvs_maintained_SupMod_c} and \Cref{prop:eps_for_revenue_bound_supMod_c}.

We obtain the following corollary
\begin{corollary}\label{cor:opt_set_if_non_empty_intersection_sub-sup}
    If $x_f \cap x_c$ is non-empty, then the optimal contract incentivizes some set $S_t \cup \{n+1\}$, where $S_t \subseteq [n]$ and $S_t \in x_f \cap x_c$. 
\end{corollary}

Akin to \Cref{prop:BR_in_approxDemand_f_c_submod}, If $S_\alpha$ is a best-response in $\InstHatI$, then $S_\alpha \setminus \{n+1\}$ is approximately a best-response in $\InstTildeI$. The proof is identical.
\begin{proposition}\label{prop:BR_in_approxDemand_sub-sup}
    Let $S_\alpha$ be the agent's best response for contract $\alpha$ in the instance $\InstHatI$, then 
    $S_\alpha \setminus\{n+1\} \in \BR{\sigma/2}{\alpha}{\tI}$.
\end{proposition}
\Cref{cor:opt_set_if_non_empty_intersection_sub-sup} and \Cref{prop:BR_in_approxDemand_sub-sup} imply \Cref{thm:cc_submod_f_supmod_c}, and the arguments are identical to the ones presented in the proof of  \Cref{thm:cc_submod_f_c} in \Cref{sec:CC_submod_fc}.

\section{Communication Complexity Hardness for Supermodular $f$ and Supermodular $c$}\label{sec:cc_lower_bound_supmod_f_c}

In this section, we prove communication complexity lower bound when both $f$ and $c$ are supermodular.
The arguments are analogous to the ones presented in \Cref{sec:CC_submod_fc}, but some of the technical details differ. In particular, our construction uses the equal-revenue instance in \Cref{sec:eqRev_SupMod_c}, for additive $f$ and supermodular $c$, as a starting point. 
For completeness, we fully describe the construction required to achieve this result. Namely, we show how to:
\begin{enumerate}
    \item Perturb the rewards of the equal revenue instance so that the new instance $\InstFTildeI$ has supermodular costs and rewards, $2^n-1$ critical values and sparse best-response.
    \item Augment the instance with an additional action whose marginals depend on two vectors $x_f,x_c \in \{0,1\}^{\binom{n}{n/2}}$ such that (i) rewards and costs are supermodular, and (ii) if Alice and Bob know $\tI$, they can compute a best-response query using $poly(n)$ communication.
\end{enumerate}
Once we establish these properties, the remaining arguments are identical to the ones for the case of submodular costs and rewards, presented in \Cref{sec:eqRev_SupMod_c}.
    
\begin{theorem}[CC lower bound, supermodular $f$ and $c$]\label{thm:cc_supmod_f_c}
    When $f$ and $c$ are supermodular, any protocol which computes the optimal contract and makes $poly(n)$ best-response queries, requires $\Omega(2^n/\sqrt{n})$ bits of communication.
\end{theorem}

The proposition below follows immediately from \Cref{obs:critValStructure} and \Cref{lem:c_eqrev_strictly_supermodolar}.
\begin{observation}\label{obs:equalrev_addf_prop}
Let $\inst{n}{f}{c}$ be an equal-revenue instance with additive $f$ and supermodular $c$. Then $f$ is strictly monotone and $c$ is strictly monotone and strictly supermodular.
\end{observation}
    
Denote the set of incentivizable sets with $\IS = \{S_1,\dots,S_{2^n-1}\}$. Recall that $c(S_{t-1})<c(S_t)$ and $f(S_{t-1})<f(S_t)$, for any $t$.
We denote $S_0 = \emptyset$ and assume $c$ and $f$ are normalized such that $c(\emptyset)=f(\emptyset)=0$.
We perturbed the rewards to create a supermodular function $\tf$. The perturbations are defined with respect to 
$\phi_f = \min_{t} f(S_{t+1})-f(S_{t}) > 0$ and
$\phi_c = \min_{t} c(S_{t+1})-c(S_{t}) > 0$, the discrete derivatives of $f$ and $c$ with respect to $t$. By \Cref{obs:equalrev_addf_prop}, these are well defined.

\subsection{Making $f$ strictly supermodular}
The perturbed reward function is defined as follows, $\tf(S)=f(S) + \delta \cdot |S|^2$, where $\delta$ satisfies
\begin{align}\label{eq:eps_for_tf}
0 < 
\delta < \min \Bigg\{ 
& \min_t \frac{\phi_f}{n^2}(\alpha_{t+1}-\alpha_t), 
\min_t \frac{1}{2n^2}\left(\frac{1}{\alpha_t}-\frac{1}{\alpha_{t+1}}\right)\frac{(c_{t+1}-c_{t})(c_{t}-c_{t-1})}{(c_{t+1}-c_{t-1})}
 \bigg\},
\end{align}
where $c_t = c(S_t)$.
First observe that $\tf$ is defined exactly as the supermodular cost function $\tc$ in \Cref{sec:cc_lower_bound_submod_f_supmod_c} (up to renaming). Thus, the same argument as in the proof of \Cref{prop:tc_nonNeg_monotone_supMod} suggests that $\tf$ is non-negative, monotone and supermodular.
    
\begin{proposition}\label{prop:tf_basic_props}
For $0 < \delta$, $\tf$ is non-negative, strictly monotone and strictly supermodular.
\end{proposition}
Next we show that for a proper choice of $\delta$, the number of critical values in the perturbed instance $\InstFTildeI$ is still $2^n-1$.

\begin{proposition}\label{prop:tf_exp_many_cvs}
If $\delta > 0$ satisfies \Cref{eq:eps_for_tf}, then $\InstFTildeI$ has $2^n-1$ critical values.
\end{proposition}

\begin{proof}
Since $\delta < \frac{\phi_f}{n^2}$, then $\talpha_t > \frac{c_{t+1}-c_t}{f_{t+1} - f_t + \delta n^2} > 0$, and in particular $\talpha_2 > \talpha_1 = 0$.
Fix $t\in \{2,\dots,2^n-2\}$. Denote $c_t = c(S_t)$, $\tf_t=\tf(S_t)$, and $\talpha_t = \frac{c_t-c_{t-1}}{\tf_t-\tf_{t-1}}$.
We wish to show that $\talpha_{t+1} > \talpha_t$, or equivalently, $1/{\talpha_{t+1}} < 1/{\talpha_t}$. 
By definition of $\talpha_t$, the following expression is equivalent
$$
\frac{f_{t+1} - f_t + \delta(|S_{t+1}|^2-|S_t|^2)}{c_{t+1} - c_t }
<
\frac{f_{t} - f_{t-1} + \delta(|S_{t}|^2-|S_{t-1}|^2)}{c_{t} - c_{t-1}}.
$$
It is sufficient to show that
$$
\frac{1}{\alpha_{t+1}}+ \frac{\delta n^2}{c_{t+1}-c_t} < \frac{1}{\alpha_{t}} - \frac{\delta n^2}{c_{t}-c_t}.
$$
By rearranging we get that this is equivalent to 
$$
\delta < \frac{1}{2n^2}\left(\frac{1}{\alpha_t}-\frac{1}{\alpha_{t+1}}\right)\frac{(c_{t+1}-c_{t})(c_{t}-c_{t-1})}{(c_{t+1}-c_{t-1})}.
$$
Finally, by the monotonicity of $f$ and \Cref{obs:critValStructure}, $S_{2^n-1}=[n]$, and
$$
\talpha_{2^n-1} = \frac{c_{2^n-1}-c_{2^n-2}}{f_{2^n-1}-f_{2^n-1}+\delta(|S_{2^n-1}|^2-|S_{2^n-2}|^2)} < \frac{c_{2^n-1}-c_{2^n-2}}{f_{2^n-1}-f_{2^n-1}} = \alpha_{2^n-1} <1
$$
As each $\talpha_t = \frac{c(S_t)-c(S_{t-1})}{\tf(S_t)-\tf(S_{t-1})}$ corresponds to a contract for which the agent's best response change in the instance $(n,\tf,c)$, there are $2^n-1$ distinct critical values.
\end{proof}
    
\subsection{Establish Sparse Best-Response}

\begin{proposition}\label{prop:sparse_demand_supmod_f_c}
    If the instance $\instanceI$ has $\sigma$-sparse supply for some $\sigma > 0$, then for any $0 < \delta < \frac{\sigma}{2n^2}$, the perturbed instance $\InstTildeI$ has $(\sigma/2)$-sparse best-response.
\end{proposition}
\begin{proof}
\newcommand{\sa}{S_\alpha}
\newcommand{\tsa}{\tilde{S}_\alpha}
    Fix a contract $\alpha$ and some $0 < \delta < \frac{\sigma}{2n}$. We show that 
    $\BR{\sigma/2}{\alpha}{\tI} \subseteq \BRI{\sigma/2+n\delta}{\alpha} \subseteq \BRI{\sigma}{\alpha}$, which implies the claim.
    Let $\sa$ be the set of actions which maximizes $\alpha f(S) - c(S)$ and let $\tsa$ be the set of which maximizes $\alpha \tf(S) - c(S)$.
    Observe that for any $S \in \BR{\sigma/2}{\alpha}{\tI}$
    it holds that,
    \begin{align*}
        \alpha \tf(S) - c(S) 
        &\ge \alpha \tf(\tsa) - c(\tsa) - \sigma/2 \\
        &\ge \alpha \tf(\sa) - c(\sa) - \sigma/2 \\
        &\ge \alpha f(\sa) - c(\sa) - \sigma/2,
    \end{align*}
    where the last inequality follows from the fact that $\tf(\sa) = f(\sa) + \delta|\sa|^2$.
    On the other hand, 
    $$\alpha \tf(S) - c(S) = \alpha f(S) - c(S) + \delta|S|^2 \le \alpha f(S) - c(S) + \delta n^2,$$ 
    we conclude that 
    \[
    \alpha f(S) - c(S) \ge \alpha f(\sa) - c(\sa) - (\sigma/2 + \delta n^2),
    \]
    which implies that $S \in \BRI{\sigma/2+n^2\delta}{\alpha}$.
    By our choice of $\delta$, it holds that $\sigma/2+n^2\delta \le \sigma$, and thus $\BRI{\sigma/2+n^2\delta}{\alpha} \subseteq \BRI{\sigma}{\alpha}$, which concludes the proof.
\end{proof}
    
\subsection{Adding the ($n+1$)th action}
As in \Cref{subsec:cc_lower_bound_submod_fc} and \Cref{sec:cc_lower_bound_submod_f_supmod_c}, we 
define the marginal rewards and costs of the ($n+1$)th action, 
using $z = \min \{\phi_{c},\psi_{c},\phi_{f},\phi_{\tf},\psi_{\tf}, \zeta\}$, where
\begin{itemize}
    \setlength\itemsep{-0.25em}
    \item $\phi_{c} = \min_{c} c(S_{t+1})-c(S_{t}) > 0$. For any $S \subsetneq T \subseteq [n]$, $c(T)-c(S) \ge \phi_{c}$),
    \item $\phi_{\tf} = \min_{t} \tf(S_{t+1})-\tf(S_{t}) > 0$, note that for any $S \subsetneq T \subseteq [n]$, $\tf(T)-\tf(S) \ge \phi_{\tf}$),
    \item $\phi_{f} = \min\{1/2,\min_{t} f(S_{t+1})-f(S_{t})\} > 0$. For any $S \subsetneq T \subseteq [n]$, including $S = \emptyset$, $f(T)-f(S) \ge \phi_{f}$,
    \item $\psi_{c} = \min_{S\subseteq T, i \notin T} c(i \mid S)-c(i \mid T) > 0$,
    \item $\psi_{\tf} = \min_{S\subseteq T, i \notin T} \tf(i \mid S)-\tf(i \mid T) > 0$,
    \item $\zeta = \frac{\delta \cdot (1-\alpha_{2^n-1})\phi_f}{16\cdot n^2 \cdot(f_{2^n-1}+1)}> 0$.
    \item $\sigma>0$, the sparseness parameter for which $\InstFTildeI$ has sparse demand, see \Cref{prop:sparse_demand_supmod_f_c}.
\end{itemize}

\begin{proposition}\label{prop:eps_for_revenue_sup_f}
For any $0 < z < \zeta = \frac{\delta \cdot (1-\alpha_{2^n-1})\phi_f}{16\cdot n^2 \cdot (f_{2^n-1}+1)}$, for any $t \in [2^n-1]$, 
$$
(1-\talpha_t)f(S_t) \in \left[1 - \frac{z(1-\alpha_{2^n-1})}{16}, 1 + \frac{z(1-\alpha_{2^n-1})}{16} \right]
$$
\end{proposition}
\begin{proof}
    Fix $t$. Denote $f_t = f(S_t)$ and $c_t = c(S_t)$.
    First observe that,
    $$
    (1-\talpha_t)\tf_t \in [(1-\talpha_t)f_t - \delta \cdot n^2, (1-\talpha_t)f_t + \delta \cdot n^2].
    $$
    Next, observe that $\talpha_t \le \alpha_t$, and that
    \begin{align*}
        \alpha_t - \talpha_t
        &=
        \frac{c_t-c_{t-1}}{f_t-f_{t-1}} - \frac{c_t-c_{t-1}}{f_t-f_{t-1} + \delta(|S_t|^2-|S_{t-1}|^2)}\\
        &=
        \frac{\delta(c_t-c_{t-1})(|S_t|^2-|S_{t-1}|^2)}{(f_t-f_{t-1})(f_t-f_{t-1} + \delta(|S_t|^2-|S_{t-1}|^2))} \\
        &=\frac{\delta\alpha_t(|S_t|^2-|S_{t-1}|^2)}{f_t-f_{t-1} + \delta(|S_t|^2-|S_{t-1}|^2)} \\
        &\le \frac{\delta \alpha_t n^2}{\phi_f}
    \end{align*}
    By combining the two observations above we get that,
    \begin{align*}
    (1-\talpha_t)\tf_t 
    &\in
    \left[(1-\alpha_t)f_t - \frac{\delta n^2\alpha_t f_t}{\phi_f} - \delta n^2,  (1-\alpha_t)f_t + \frac{\delta n^2\alpha_t f_t}{\phi_f} + \delta n^2\right] \\
    &\subseteq
    \left[1 - \frac{\delta n^2(\alpha_t f_t+1)}{\phi_f},  1 + \frac{\delta n^2(\alpha_t f_t+1)}{\phi_f}\right] \\
    &\subseteq
    \left[1 - \frac{\delta n^2(f_{2^n-1}+1)}{\phi_f},  1 + \frac{\delta n^2(f_{2^n-1}+1)}{\phi_f}\right] \\
    &\subseteq
    \left[1 - \frac{z(1-\alpha_{2^n-1})}{16},  1 + \frac{z(1-\alpha_{2^n-1})}{16}\right],
    \end{align*}
    where the second transition follows from the fact that $\phi_f < 1$, and  last transition follows from the definition of $\zeta$.
\end{proof}
We are ready to define the new instance $\inst{n+1}{\hf}{\hc}$, where 
\begin{align*}
\hf(S)&=\tf(S\setminus\{n+1\})+\indicator{n+1 \in S}\cdot\hf(n+1 \mid S), \text{ and } \\
\hc(S)&=c(S\setminus\{n+1\})+\indicator{n+1 \in S}\cdot\hc(n+1 \mid S).
\end{align*}
The marginal rewards of the $n+1$ action are as follows
$$
\hf(n+1 \mid S_t) = 
\begin{cases}
    0 & |S_t| < \frac{n}{2} \\
    0 & |S_t| = \frac{n}{2} \land S_t \notin x_f(S_t) \\
    z/4 & |S_t| = \frac{n}{2} \land S_t \in x_f(S_t)\\
    z/4 & |S_t| > \frac{n}{2}
\end{cases}
\qquad 
\hc(n+1 \mid S_t) = 
\begin{cases}
    (\talpha_1\cdot z)/8 & |S_t| < \frac{n}{2} \\
    (\talpha_t\cdot z)/4 & |S_t| = \frac{n}{2} \land S_t \in x_c \\
    z/2 & |S_t| = \frac{n}{2} \land S_t \notin x_c \\
    z/2 & |S_t| > \frac{n}{2}
\end{cases}.
$$
    
\begin{proposition}\label{prop:monotone_supmod_hf_hc}
    Both $\hf$ and $\hc$ are monotone and supermodular.
\end{proposition}

\begin{proof}
    \textbf{Monotonicity of $\hc$:} 
    Fix $S \subsetneq T \subseteq [n+1]$. 
    \begin{align*}
            \hc(T) - \hc(S) 
            &\ge 
            \hc(T \cap [n]) - \hc(S \cup \{n+1\}) 
            = c(T\cap [n]) - c(S\cap [n]) - c(n+1\mid S\cap[n]) \\
            &\ge \phi_{c} - c(n+1\mid S\cap[n]) 
            \ge z - z/2 
            > 0
    \end{align*}
    First inequality follows as the marginal cost of $n+1$ is non-negative.  
    
    \textbf{Submodularity of $\hc$:} 
    We will show that for any $S \subseteq [n+1]$ and any $i,j \notin S$, where $i \ne j$, it holds that 
    \begin{equation}\label{eq:hcSupMod}
        \hc(S \cup \{i\}) + \hc(S \cup \{j\}) \le \hc(S \cup \{i,j\}) + \hc(S)
    \end{equation}
    
    \begin{itemize}
    \setlength\itemsep{-0.25em}
    \item If $n+1 \notin S \cup \{i,j\}$, then the claim follows from \Cref{obs:equalrev_addf_prop}.

    \item If $i=n+1$ (and similarly for the case where $j=n+1$), by rearranging \Cref{eq:hcSupMod} one can see that it is equivalent to,
    $\hc(n+1 \mid S) \le \hc(n+1 \mid S \cup \{j\})$. The claim follows immediately as by the fact that $\hc(n+1 \mid T)$ is weakly increasing with the cardinality of $T$.

    \item If $n+1 \in S$, let $S' = S \setminus \{n+1\}$. 
    Then be rewriting $\hc(S \cup \{i\})$ as $\hc(n+1 \mid S' \cup \{i\}) + \hc(S'\cup\{i\})$, and similarly to the other expressions, we get that 
    \Cref{eq:tcSubMod} is equivalent to
    \begin{align*}
    \hc(j \mid S' \cup \{i\}) - \hc(j \mid S')
    = c(j \cup S' \cup \{i\}) - c(S' \cup \{i\}) &- (c(j \cup S') - c(S')) \\
    \ge 
    \underbrace{\hc(n+1 \mid S' \cup \{i\}) - \hc(n+1 \mid S'\cup\{i,j\})   - \hc(n+1 \mid S')}_{\le 0} &+ \hc(n+1 \mid S' \cup \{j\})
    \end{align*}
    The left-hand side of the equation is lower bounded by $\psi_{c}$. 
    Since $\hc(n+1 \mid T)$ weakly decreases with the cardinality of $T$, $\hc(n+1 \mid S' \cup \{i\}) - \hc(n+1 \mid S'\cup\{i,j\}) \le 0$. Also, the marginal cost of $n+1$ is non-negative, so $\hc(n+1 \mid S') \ge 0$. Thus, we can upper bound the right-hand side with
    $\hc(n+1 \mid S' \cup \{j\}) \le z/2 < \psi_{c}$, which concludes the proof.
    \end{itemize}
    
    \textbf{Monotonicity of $\hf$:} 
    Fix $S \subsetneq T \subseteq [n+1]$. Observe that regardless of whether $n+1 \in S$, it holds that
    \begin{align*}
            \hf(T) - \hf(S) 
            &\ge \hf(T \cap [n]) - \hf(S \cup \{n+1\}) \\
            &= \tf(T\cap [n]) - \tf(S\cap [n]) - \hf(n+1\mid S\cap[n]) \\
            &\ge \phi_{\tf} - \hf(n+1\mid S\cap[n]) \\ 
            &\ge z - z/4 \\
            &> 0
    \end{align*}
    
    \textbf{Supermodularity of $\hf$:}
    We will show that for any $S \subseteq [n+1]$ and any $i,j \notin S$, where $i \ne j$, it holds that 
    \begin{equation}\label{eq:fSupMod}
        \hf(S \cup \{i\}) + \hf(S \cup \{j\}) \le \hf(S \cup \{i,j\}) + \hf(S)  
    \end{equation}
    
    \begin{itemize}
    \setlength\itemsep{-0.25em}
    \item If $n+1 \notin S \cup \{i,j\}$, then the claim follows from \Cref{prop:tf_basic_props}.
    \item If $i=n+1$ (and similarly for the case where $j=n+1$), by rearranging \Cref{eq:fSupMod} one can see that it is equivalent to,
    $\hf(n+1 \mid S) \le \hf(n+1 \mid S \cup \{j\})$. The claim follows immediately as by the fact that $\hf(n+1 \mid T)$ is weakly decreasing with the cardinality of $T$.
    \item If $n+1 \in S$, let $S' = S \setminus \{n+1\}$. 
    Then be rewriting $\hf(S \cup \{i\})$ as $\hf(n+1 \mid S' \cup \{i\}) + \hf(S'\cup\{i\})$, and similarly to the other expressions, we get that 
    \Cref{eq:fSupMod} is equivalent to
    \begin{align*}
    \hf(j \mid S' \cup \{i\}) - \hf(j \mid S')
    = \tf(j \cup S' \cup \{i\}) - \tf(S' \cup \{i\}) &- (\tf(j \cup S') - \tf(S')) \\
    \ge 
    \underbrace{\hf(n+1 \mid S' \cup \{i\}) - \hf(n+1 \mid S'\cup\{i,j\})   - \hf(n+1 \mid S')}_{\le 0} &+ \hf(n+1 \mid S' \cup \{j\})
    \end{align*}
    The left-hand side of the equation is lower bounded by $\psi_{\tf}$. 
    As in the argument for the supermodularity of $\hc$, the right-hand side can be upper bounded with $\hf(n+1 \mid S' \cup \{j\}) \le z/4 < \psi_{c}$, which concludes the proof. \qedhere
    \end{itemize}
\end{proof}

\begin{observation}\label{obs:n+1Incentivizable_supmod_hf_hc}   
    If a set of the form $S_t \cup \{n+1\}$, for some $S_t \subseteq [n]$, can be incentivized by the principal, then $|S_t| = n/2$ and $S_t \in x_c \cap x_f$.
\end{observation}
\begin{proof}
We divide into cases: 

    If $|S_t| < \frac{n}{2}$, the marginal contribution of $n+1$ is $0 - \frac{\talpha_1\cdot z}{8} < 0$.
    
    If $|S_t|>n/2$, the marginal contribution of $n+1$ is $z/4 - z/2 < 0$.
    
    If $|S_t|=n/2$ and $x_f(S_t)=0$, the marginal contribution of $n+1$ is at most $0 - \frac{\talpha_t\cdot z}{4} < 0$.
    
    If $|S_t|=n/2$ and $x_c(S_t)=0$, the marginal contribution of $n+1$ is at most $z/4 - z/2 < 0$.
\end{proof}

\begin{proposition}\label{prop:n+1_is_optimal_supmod_hf_hc}
    Let $S_t$ be the set with the minimal index such that $S_t \in x_c \cap x_f$. 
    The contract $\halpha_t = \frac{\hc(S_t)-\hc(S_{t-1})}{\hf(S_t)-\hf(S_{t-1})}$ incentivizes the set $S_t \cup \{n+1\}$ and
    the principal's revenue from $\halpha_t$, $u_p(\halpha_t,S_t\cup\{n+1\})$, exceeds the revenue from any incentivizable set of action of the form $S_{t'} \subseteq [n]$.
\end{proposition}
The proof of \Cref{prop:n+1_is_optimal_supmod_hf_hc} is identical to the proof of \Cref{prop:n+1_is_optimal}, except it uses the bounds established in \Cref{prop:eps_for_revenue_sup_f} and \Cref{eq:eps_for_tf}.
We obtain the following corollary
\begin{corollary}\label{cor:opt_set_if_non_empty_intersection_sup-sup}
    If $x_f \cap x_c$ is non-empty, then the optimal contract incentivizes some set $S_t \cup \{n+1\}$, where $S_t \subseteq [n]$ and $S_t \in x_f \cap x_c$. 
\end{corollary}
    
Akin to \Cref{prop:BR_in_approxDemand_f_c_submod}, If $S_\alpha$ is a best-response in $\InstHatI$, then $S_\alpha \setminus \{n+1\}$ is approximately a best-response in $\InstFTildeI$.
\begin{proposition}\label{prop:BR_in_approxDemand_sup-sup}
    Let $S_\alpha$ be the agent's best response for contract $\alpha$ in the instance $\InstHatI$, then 
    $S_\alpha \setminus\{n+1\} \in \BR{\sigma/2}{\alpha}{\tI}$.
\end{proposition}
The proof is identical to the analogous \Cref{prop:BR_in_approxDemand_f_c_submod}.
\Cref{cor:opt_set_if_non_empty_intersection_sup-sup} and \Cref{prop:BR_in_approxDemand_sup-sup} imply \Cref{thm:cc_supmod_f_c}, and the arguments are identical to the ones presented in the proof of  \Cref{thm:cc_submod_f_c} in \Cref{sec:CC_submod_fc}.

\section{Polynomial Representation of the Equal Revenue Instance for Submodular $f$ and Additive $c$}\label{apx:rounded}

\newcommand{\tn}{\theta(n^2)}
\DeclarePairedDelimiter{\floor}{\lfloor}{\rfloor}
\newcommand{\famI}{\I = \{ \inst{n}{f_{S_t}}{c} \}_{S_t \subseteq [n],S \ne \emptyset}}

In this section we round our equal-revenue reward function for submodular $f$ and additive $c$, presented in \Cref{sec:eqRev}, so it can be represented with $poly(n)$ bits.
We show that this new reward function, $\tf$, can be extended to a family of optimal-contract problems to show hardness in the demand query model, as in \Cref{sec:demandHardness}.

\begin{theorem}\label{thm:PolyRepr}
    When $f$ is submodular any algorithm that computes the optimal contract requires an exponential number of demand queries in the instance representation size. 
\end{theorem} 

Recall the family of perturbed instances defined in \Cref{sec:eqToVal}, which we denote in this section by $\famI$.
Each reward function $f_{S_t}$ was defined to be identical to the original $f$, except for a small fixed bonus $\eps > 0 $ for the set $S_t$, namely 
\begin{equation}\label{eq:fS}
f_{S_t}(S_{t'}) = \begin{cases}
    f(S_{t'}) & t' \ne t, \\
    f(S_{t'}) + \eps & t' = t.
\end{cases}    
\end{equation}

The purpose of the bonus was to guarantee that
when the reward function is $f_{S_t}$, there are $2^n-1$ critical values, but only one of these critical values, $\alpha_t$, is optimal. 
Our argument also relied on the fact that $\eps$ is small enough to ensure that a demand oracle for $f_{S_t}$ can be simulated with poly-many value queries, and that the monotonicity and submodularity of $f$ are maintained.

Thus, to prove \Cref{thm:PolyRepr} it is enough to show that there exists a rounded reward function $\tf$ and  $\eps > 0$, whose representation is polynomial in $n$, such that the following properties hold for any optimal contract instance defined with respect to the costs $c_i = 2^{i-1}$ and a reward function $\tf_{S_t}$:
\begin{enumerate}
    \setlength\itemsep{-0.25em}
    \item \label[prop]{pr:monSubMod} $\tf_{S_t}$ is monotone and submodular.
    \item \label[prop]{pr:demandOracle} A demand oracle for $\tf_{S_t}$ can be simulated with poly-many value queries.
    \item \label[prop]{pr:uniqOpt} There are $2^n-1$ critical values $0 = \gamma_0 < \dots < \gamma_{2^n-1} < 1$,  each yielding an expected utility of at most $1 + 2^{-\tn}$ to the principal, except for one which yields $1 + 2^{-\theta(n)}$.
\end{enumerate}

In \Cref{subApx:Bounds} we show that in the original construction any two critical values are significantly spaced from each other, and in \Cref{subApx:construction} we prove that rounding our construction with precision $2^{-\tn}$ is sufficient to maintain the above properties.

\subsection{Bounds for the Original Construction}\label{subApx:Bounds}
The goal of this section is to give a lower bound on the distance between any two critical values (and between the last critical value and $1$), which is of order $2^{-\theta (n)}$.

\begin{lemma}\label{lem:ztBounds}
    For any $t \in \{1,\dots,2^n-1\}$,
    $$
    (1-\alpha_t)^3 < \alpha_{t+1} - \alpha_t < (1-\alpha_t)^{3/2}. 
    $$
\end{lemma}
\begin{proof}
    Fix $t$. Recall that 
    $\alpha_{t+1} - \alpha_t = \frac{1}{2}(\sqrt{4\alpha_t^2 - 8\alpha_t +5} - 1)$.
    First, observe that the inequality holds for $\alpha_1 = \frac{\sqrt{5}-1}{2} \approx 0.618$, as $(1-\alpha_1)^{3} \approx 0.0055$, $\alpha_2 - \alpha_1 \approx 0.129$ and $(1-\alpha_1)^{3/2} \approx 0.236$.

    To prove the inequality on the right, observe that for any $x \in [\alpha_1,1)$:
    \begin{eqnarray*}
        \left(\frac{\frac{1}{2}(\sqrt{4x^2 - 8x +5} - 1)}{(1-x)^{3/2}}\right)'
        &=&
        \frac{1}{2(1-x)^3}\cdot \left(\frac{8(x-1)(1-x)^{3/2}}{2\sqrt{4x^2-8x+5}} + \frac{3}{2}(1-x)^{1/2}\cdot(\sqrt{4x^2-8x+5}-1)\right) \\
        &=&
        \frac{1}{2(1-x)^{5/2}}\cdot \left(\frac{-4(1-x)^2}{\sqrt{4x^2-8x+5}} + \frac{3}{2}\cdot(\sqrt{4x^2-8x+5}-1)\right).
    \end{eqnarray*}
    The above is negative if and only if
    \begin{eqnarray*}
        \frac{3}{2}(\sqrt{4x^2-8x+5}-1) < \frac{4(1-x)^{2}}{\sqrt{4x^2-8x+5}}.
    \end{eqnarray*}
    One can verify that this is equivalent to
    \begin{eqnarray*}
        2x^2 - 4x + \frac{7}{2} - \frac{3}{2}\sqrt{4x^2-8x+5} < 0,
    \end{eqnarray*}
    and that this inequality holds for any $x \in [\alpha_1,1)$ and in particular for any $\alpha_t$ with $t \ge 1$. 
    For the left inequality,  
    \begin{eqnarray*}
        \left(\frac{\frac{1}{2}(\sqrt{4x^2 - 8x +5} - 1)}{(1-x)^{3}}\right)'
        &=&
        \frac{1}{2(1-x)^6}\cdot \left(\frac{8(x-1)(1-x)^{3}}{2\sqrt{4x^2-8x+5}} +3(1-x)^2\cdot(\sqrt{4x^2-8x+5}-1)\right) \\
        &=&
        \frac{1}{2(1-x)^4}\cdot \left(\frac{-4(1-x)^{2}}{\sqrt{4x^2-8x+5}} + 3(\sqrt{4x^2-8x+5}-1)\right), \\
    \end{eqnarray*}
    which is positive if and only if
    $$
    3(\sqrt{4x^2-8x+5}-1) > \frac{4(1-x)^{2}}{\sqrt{4x^2-8x+5}}.
    $$
    Equivalently,
    $$
    8x^2-12x+11-\sqrt{4x^2-8x+5} > 0,
    $$
    which holds for any $x \in \reals$.
\end{proof}

\begin{lemma}\label{lem:distFrom1}
    For any $t \in [2^n-1]$,  $1-\alpha_t \ge 2^{-6n}$. 
\end{lemma}
\begin{proof}
    Observe that $1-\alpha_t$ is a monotonically decreasing series. 
    Aiming for contradiction, assume there exists $t$ such that $1-\alpha_t <2^{-6n}$, and in particular $\alpha_{2^n-1} > 1- 2^{-6n}$.
    Let $t^*$ be such that
    $1-\alpha_{t^*} < 2^{-4n}$ and $1-\alpha_{t^*-1} \ge 2^{-4n}$.
    Using telescoping sum  and the upper bound from \Cref{lem:ztBounds} we get 
    \begin{eqnarray*}
        \alpha_{t^*} - \alpha_1
        &=&
        \sum_{t=1}^{t^*-1}\alpha_{t+1} - \alpha_t 
        =
        \sum_{t=1}^{t^*-1}\frac{1}{2}\left( \sqrt{4\alpha_t^2 - 8\alpha_t + 5} -1 \right) \\
        &\le&
        \int_{\alpha_1}^{\alpha_{t^*-1}}\frac{1}{2}\left( \sqrt{4x - 8x + 5} -1 \right) dx \\
        &<&
        \int_{\alpha_1}^{\alpha_{t^*-1}} (1-x)^{3/2} dx \\
        &=&
        \frac{-2}{5}\left[ (1-x)^{5/2} \right]_{\alpha_1}^{\alpha_{t^*-1}} \\
        &\le&
        \frac{2}{5}\left[ (1-\alpha_1) -(1-\alpha_{t^*-1})^{5/2} \right] \\
        &\le&
        \frac{2}{5}\left[ (1-\alpha_1) -(2^{-4n})^{5/2} \right] 
        =
        \frac{2}{5}\left[ (1-\alpha_1) -(2^{-10n}) \right],
    \end{eqnarray*}
    where the last inequality follows from the definition of $t^*$. We get 
    \begin{equation}\label{eq:t*Up}
        \alpha_{t^*} \le \frac{2}{5} + \frac{3}{5}\alpha_1 -\frac{2}{5}\cdot 2^{-10n} \approx 0.7708 - \frac{2}{5}\cdot 2^{-10n}
    \end{equation}
    \Cref{lem:alphaMon} implies that $\alpha_{t+1} - \alpha_t$ is a decreasing function of $t$,
    so for any $t > t^*$, 
    $$
    \alpha_{t+1} - \alpha_t \le \alpha_{t^*+1} - \alpha_{t^*} < (1-\alpha_{t^*})^{3/2} < 2^{-6n},
    $$
    where the second to last inequality is due to \Cref{lem:ztBounds}.
    Thus, $\alpha_{2^n-1}$ satisfies:
    $$
    \alpha_{2^n-1} 
    = \alpha_{t^*} + \sum_{t=t^*+1}^{2^n-2} \alpha_{t+1} - \alpha_{t} 
    \le \alpha_{t^*} + 2^n \cdot 2^{-6n}
    = \alpha_{t^*} + 2^{-5n}.
    $$
    According to our assumption, $\alpha_{2^n} > 1-2^{-6n}$ and so 
    \begin{equation}\label{eq:t*Low}
        \alpha_{t^*} \ge 1-2^{-6n}-2^{-5n}.
    \end{equation}
    Combining \Cref{eq:t*Up} and \Cref{eq:t*Low} we get
    $$
    2^{-6n} + 2^{-5n} - \frac{2}{5}2^{-10n} \ge 0.2292,
    $$
    a contradiction for any $n >1$.
\end{proof}

\begin{corollary}\label{cor:alphaDiffBound}
    For any $t$, $\alpha_{t+1} - \alpha_t \ge 2^{-18n}$
\end{corollary}
\begin{proof}
    By Lemmas \ref{lem:ztBounds} and \ref{lem:distFrom1} we have
    $$
    \alpha_{t+1} - \alpha_t \ge (1-\alpha_t)^3 \ge 2^{-18n},
    $$
    as claimed.
\end{proof}

\subsection{Poly-Size Construction}\label{subApx:construction}
We define a rounded version of our original critical values $\talpha_t = \floor{\alpha_t}$, where $\floor{x} = \max\{k\cdot 2^{-\tn} \mid k \in \mathbb{Z}, x \ge k\cdot 2^{-\tn} \}$, 
and a new reward function:
\begin{align}\label{eq:roundingscheme}
\tf(S_t) &= \tf_t = \frac{1}{1-\talpha_t}.
\end{align}

We show that $\tf$ possesses the three required properties defined at the beginning of \Cref{apx:rounded}. 
We establish \Cref{pr:monSubMod} in \Cref{cor:ftlSubmod}, \Cref{pr:demandOracle} in \Cref{prop:ftlDemand},
and \Cref{pr:uniqOpt} in \Cref{prop:ftlUniqOpt}. 

We begin by bounding the distance between the rounded instance and the original one, proving that the series of differences $\tf_{t+1} - \tf_t$ is equal to 
$f_{t+1} - f_t$  up to an additive factor of $2^{-\tn}$.
\begin{lemma}\label{lem:ftlProximity}
For any $t \in [2^n-1]$,
$$
\tf_{t+1} - \tf_t = (f_{t+1} - f_t) \pm 2^{-\tn}.
$$
\end{lemma}

\begin{proof}
Denote $\alpha_t - \talpha_t = \alpha_t - \floor{\alpha_t} = e_t$. Note that $0 \le e_t < 2^{-\tn}$. With this definition at hand, we can rewrite $\tilde{f}_{t+1}-\tilde{f}_t$ as follows:
\begin{eqnarray*}
    \tf_{t+1} - \tf_t 
    &=&
    \frac{1}{1-\talpha_{t+1}} - \frac{1}{1-\talpha_t} \\
    &=&
    \frac{\talpha_{t+1} - \talpha_t}{(1-\talpha_t)(1-\talpha_{t+1})} \\
    &=&
    \frac{\alpha_{t+1} - \alpha_t + e_t - e_{t+1}}{(1-\alpha_t + e_t)(1-\alpha_{t+1} + e_{t+1})} \\
    &=&
    \frac{\alpha_{t+1} - \alpha_t + e_t - e_{t+1}}{(1-\alpha_t)(1-\alpha_{t+1}) + e_t(1-\alpha_{t+1}) + e_{t+1}(1-\alpha_t) + e_{t+1}e_{t}}. 
\end{eqnarray*}
We upper bound
\begin{eqnarray*}
     \tf_{t+1} - \tf_t 
    &\le&
    \frac{\alpha_{t+1} - \alpha_t + 2^{-\tn}}{(1-\alpha_t)(1-\alpha_{t+1})} \\
    &=&
    f_{t+1} - f_t + \frac{2^{-\tn}}{(1-\alpha_t)(1-\alpha_{t+1})} \\
    &\le&
    f_{t+1} - f_t + \frac{2^{-\kappa n}}{2^{-6n}\cdot 2^{-6n}} \\
    &=&
    f_{t+1} - f_t + 2^{-\tn}.
\end{eqnarray*}
We lower bound
\begin{eqnarray*}
    \tf_{t+1} - \tf_t 
    &=&
    \frac{\alpha_{t+1} - \alpha_t + e_t - e_{t+1}}{(1-\alpha_t)(1-\alpha_{t+1}) + e_t(1-\alpha_{t+1}) + e_{t+1}(1-\alpha_t) + e_{t+1}e_{t}} \\
    &>&
    \frac{\alpha_{t+1} - \alpha_t - 2^{-\tn}}{(1-\alpha_t)(1-\alpha_{t+1}) + 3\cdot 2^{-\tn}} \\
    &=& 
    \frac{\alpha_{t+1} - \alpha_t}{(1-\alpha_t)(1-\alpha_{t+1}) + 3\cdot2^{-\tn}} -\frac{2^{-\tn}}{(1-\alpha_t)(1-\alpha_{t+1}) + 3\cdot 2^{-\tn}} \\
    &=&
    \frac{\alpha_{t+1} - \alpha_t}{(1-\alpha_t)(1-\alpha_{t+1})}
    \cdot
    \frac{(1-\alpha_t)(1-\alpha_{t+1})}{(1-\alpha_t)(1-\alpha_{t+1}) + 3\cdot2^{-\tn}} 
    -\frac{2^{-\tn}}{(1-\alpha_t)(1-\alpha_{t+1})+ 3\cdot 2^{-\tn}}\\
    &=&
    (f_{t+1} - f_t)
    \cdot
    \left(1 - \frac{3\cdot 2^{-\tn}}{(1-\alpha_t)(1-\alpha_{t+1}) + 3\cdot 2^{-\tn}} \right)
    -\frac{2^{-\tn}}{(1-\alpha_t)(1-\alpha_{t+1})+ 3\cdot 2^{-\tn}} \\
    &\ge&
    (f_{t+1} - f_t)
    -
    \frac{3\cdot 2^{-\tn} (f_{t+1} - f_t) + 2^{-\tn}}{2^{-12n}}\\
    &\ge&
    (f_{t+1} - f_t)
    -
    \frac{6\cdot 2^{-\tn} + 2^{-\tn}}{2^{-12n}} \\
    &\ge&
    (f_{t+1} - f_t)
    -
     2^{-\tn},
\end{eqnarray*}
where the second to last inequality follows from \Cref{lem:fMarginalDecrease} and the fact that $f_1 - f_0 \approx 1.618 < 2$.
\end{proof}

As in the original instance, the discrete derivative  $\tf_{t+1}-\tf_{t}$ is negative, as shown in the following lemma.

\begin{lemma}\label{lem:ftlDecreasing}
$\tf_{t+1} - \tf_t$ is a decreasing function of $t$, for any $t$.
\end{lemma}

\begin{proof}
It is enough to show that for any $t$, 
$f_{t+1} - f_{t} - (f_t - f_{t-1}) < - 2^{-\theta(n)}$, as this implies
\begin{eqnarray*}
    \tf_{t+1} - \tf_{t} - (\tf_t - \tf_{t-1}) &\le& 
    f_{t+1} - f_{t} + 2^{-\tn} - (f_t - f_{t-1} - 2^{-\tn}) \\
    &=& 
    f_{t+1} - f_{t} - (f_t - f_{t-1}) + 2^{-\tn} \\ 
    &=& 
    - 2^{-\theta(n)} + 2^{-\tn} \\ 
    &<&
    0.
\end{eqnarray*}
Recall from \Cref{prop:CVs}, that $\alpha_{t+1} = \frac{1}{f_{t+1} - f_t}$. Thus,
\begin{eqnarray*}
    f_{t+1} - f_{t} - (f_t - f_{t-1}) =
    \frac{1}{\alpha_{t+1}} - \frac{1}{\alpha_{t}} 
    =
    \frac{\alpha_t - \alpha_{t+1}}{\alpha_t \cdot \alpha_{t+1}} 
    <
    \frac{-2^{-18}}{\alpha_0^2}
    =
    -2^{-\theta (n)},
\end{eqnarray*}
where the inequality follows from \Cref{cor:alphaDiffBound} and the monotonicity of $\alpha_t$.
\end{proof}

As an immediate corollary from the negative discrete derivative of $\tf$, together with the bounds in \Cref{lem:ztBounds}, we show that choosing $\eps = 2^{-\theta(n)}$ is enough to guarantee \Cref{pr:monSubMod}.
\begin{corollary}\label{cor:ftlSubmod}
     There exists an $\eps = 2^{-\theta(n)}$ such that for every $S_t \subseteq [n]$, $\tf_{S_t}$ is monotone and submodular.
\end{corollary}
\begin{proof}
    The monotonicity of $\tf$ follows directly from \Cref{cor:alphaDiffBound} and the rounding scheme in \Cref{eq:roundingscheme} which has precision $2^{-\tn}$. To guarantee the monotonicity of $\tf_{S_t}$ it is enough to pick $\eps < \tf_{t+1}-\tf_t$, for any $t \in [2^n-2]$. 

    To maintain the decreasing second derivative (which implies submodularity, as in the original $f$), it is enough to satisfy $2\eps < \tf_{t+1}-\tf_t - (\tf_{t+2}-\tf_{t+1})$, for any $t \in [2^n-3]$.

    One can easily show that $\eps = 2^{-\theta(n)}$ satisfies both sets of constraints by using the proximity between the discrete derivatives of $\tf$ and $f$ (\Cref{lem:ftlProximity}) and the bounds on the original construction (\Cref{lem:ztBounds}), similarly to \Cref{lem:ftlDecreasing}.
\end{proof}

Next we show that when $\tf$ is the reward function (and the additive cost is $c_i = 2^{i-1}$ for any $i \in [n]$), there are $2^n-1$ distinct critical values (different from $\talpha_t$), each of which yields the same utility to the principal, up to an additive factor of $2^{-\tn}$.

\begin{lemma}\label{lem:TildeCVs}
Consider the optimal contract problem defined with respect to $\tf$. There exists a series of critical values $0 = \beta_0 < \beta_1 < \ldots < \beta_{2^n-1} < 1$ such that for any $t \in \{0, \ldots, 2^n-2\}$, $S_t$ is the agent's best response for contract $\alpha \in [\beta_t, \beta_{t+1})$.
Additionally, each $\beta_t$ yields an expected utility of $1 \pm 2^{-\tn}$ to the principal.
\end{lemma}

\begin{proof}
    As in \Cref{prop:CVs}, for any $t \ge 0$, the critical value for which $S_{t+1}$ is the agent's best response is $\beta_{t+1} = \frac{1}{\tf_{t+1} - \tf_t}$.
    From \Cref{lem:ftlProximity},
    \begin{eqnarray*}
        \frac{1}{\tf_{t+1} - \tf_t}
        &=&
        \frac{1}{f_{t+1} - f_t \pm 2^{-\tn}} \\
        &=&
        \frac{1}{f_{t+1} - f_t} \left(\frac{f_{t+1} - f_t}{f_{t+1} - f_t \pm 2^{-\tn}}\right) \\
        &=&
        \alpha_{t+1} \left(1 \mp \frac{2^{-\tn}}{f_{t+1} - f_t \pm 2^{-\tn}}\right),
    \end{eqnarray*}
    where the last equality follows from \Cref{prop:CVs}.
    Also, from \Cref{prop:CVs} and \Cref{lem:distFrom1}, we have
    $$
    f_{t+1} - f_t = \frac{1}{\alpha_{t+1}} \ge \frac{1}{1-2^{-6n}},
    $$
    which implies
    $$
    \frac{1}{\tf_{t+1} - \tf_t} 
    \in 
    \alpha_{t+1} \mp  \frac{\alpha_{t+1} \cdot 2^{-\tn}}{1-2^{-6n} \pm 2^{-\tn}}
    \subseteq
    \alpha_{t+1} \mp  \frac{2^{-\tn}}{1-2^{-6n}}.
    $$
    By \Cref{cor:alphaDiffBound}, 
    $$
    \beta_{t+1} - \beta_t 
    \ge \alpha_{t+1} - \alpha_t -2^{-\tn} 
    \ge 2^{-18n} -2^{-\tn} 
    > 0,
    $$
    and from \Cref{lem:distFrom1},
    $$
    \beta_{2^n-1} \le 1- 2^{-6n} + 2^{-\tn} < 1.
    $$
    To see that each $\beta_t$ yields approximately the same revenue for the principal, we replace $\beta_t$ in principal's expected revenue according to the first equation in the proof:
\begin{eqnarray*}
    \tf_t(1-\beta_t) 
    =
    \frac{1}{1-\talpha_t}(1-\beta_t) 
    \in
    \frac{1}{1-\alpha_t + e_t}\left(1-\alpha_{t} \mp  \frac{2^{-\tn}}{1-2^{-6n}}\right). 
\end{eqnarray*}
This term is upper bounded as follows
\begin{eqnarray*}
    \tf_t(1-\beta_t) 
    \le
    1 + \frac{2^{-\tn}}{1-2^{-6n}} 
    \le
    1 + 2^{-\tn},
\end{eqnarray*}
and lower bounded as follows
\begin{eqnarray*}
    \tf_t(1-\beta_t) 
    &\ge&
    \frac{1}{1-\alpha_t + 2^{-\tn}}
    \left(1-\alpha_{t} - \frac{2^{-\tn}}{1-2^{-6n}} \right) \\
    &=&
    \frac{1-\alpha_{t}}{1-\alpha_t + 2^{-\tn}}
    -
    \frac{2^{-\tn}}{(1-\alpha_t+2^{-\tn})(1-2^{-6n})} \\
    &\ge&
    1 - 2^{-\tn}.
\end{eqnarray*}
This concludes the proof.
\end{proof}

The following proposition establishes \Cref{pr:uniqOpt}. It shows that picking $\eps = 2^{-\theta(n)}$ is sufficient to guarantee that (i) there are $2^n-1$ different critical values in the instance defined with respect to $\tf_{S_t}$, for any $S_t \subseteq [n]$ (see \Cref{eq:fS}), and (ii) one of these critical values yields significantly more utility to the principal than the others.

\begin{proposition}\label{prop:ftlUniqOpt}
    There exists an $\eps = 2^{-\theta(n)}$ such that there for every $S_t \subseteq [n]$, the instance defined with respect to $\tf_{S_t}$ has $2^n-1$ critical values, each of them yields revenue of at most $1 + 2^{-\tn}$ to the principal except for the optimal contract, $\gamma_t$, which yields $1+2^{-\theta(n)}$.
\end{proposition}
\begin{proof}
    Observe that when moving from the optimal contract instance defined with respect to $\tf$ to the one defined with respect to $\tf_{S_t}$, the only critical values that may change are those defined with respect to $\tf(S_t)$, namely $\beta_{t} = \frac{1}{\tf_{t} - \tf_{t-1}}$ and $\beta_{t+1} = \frac{1}{\tf_{t+1} - \tf_t}$, which turn to $\gamma_{t} = \frac{1}{\tf_{t} + \eps - \tf_{t-1}}$ and $\gamma_{t+1} = \frac{1}{\tf_{t+1} - \tf_t - \eps}$, respectively.

    Thus, in order to maintain the number of critical values to be as in \Cref{lem:TildeCVs}, we need to make sure that $\beta_{t-1} < \gamma_t < \gamma_{t+1} < \beta_{t+2}$.
    One can easily verify that for any $t$, these constraints only bound $\eps$ to be at most $2^{-\theta(n)}$, by bounding the discrete derivative, similarly to \Cref{lem:ftlDecreasing}.

    The fact that all critical values except $\gamma_t$ and $\gamma_{t+1}$ yield at most $1 + 2^{-\tn}$ utility to the principal is immediate from \Cref{lem:TildeCVs}, which also implies this for $\gamma_{t+1}$, as $\tf_{t+1}(1-\gamma_{t+1}) \le \tf_{t+1}(1-\beta_{t+1})$.
    
    It is left to show that $\gamma_t$ yields revenue of at least $1 + 2^{-\theta(n)}$.
    Using simple algebraic manipulations we can rewrite the principal's utility
    $$
    (\tf_t+ \eps)(1-\gamma_t) = 
    \tf_t(1-\beta_t) + \frac{\eps}{\tf_t + \eps - \tf_{t-1}}(1 - \eps + \tf_t \beta_t)
    \ge 
    1-2^{-\tn} + O(\eps)
    =
    1+2^{-\theta(n)},
    $$
    where the inequality follows from \Cref{lem:TildeCVs}. 
\end{proof}

We conclude by showing that $\eps = 2^{-\theta(n)}$ is sufficient for the rounded instance to satisfy \Cref{pr:demandOracle}.
\begin{proposition}\label{prop:ftlDemand}
    There exists an $\eps = 2^{-\theta(n)}$ such that for every $S_t \subseteq [n]$, a demand oracle for $\tf_{S_t}$ can be implemented with poly-many value queries.
\end{proposition}
\begin{proof}
    It is enough to show that $\tf$ satisfies \Cref{lem:RLintervals}. The claim follows from the same argument used for the original $f$.
    To see that \Cref{lem:RLintervals} holds for $\tf$ as well, replace $f$ with $\tf$ and $\beta_t$ with $\alpha_t$ in the proof.
    The same argument holds, only that we need to verify that $\eps$ can be chosen to have a polynomial representation in $n$.
    The constraint imposed on $\eps$ in this lemma is
    $$
    \eps < \frac{2^{i-2}(\beta_{S_{t'}\setminus \{i\}} - \beta_{S_{t} \cup \{i\}})}{\beta_{S_{t'}\setminus \{i\}} \cdot \beta_{S_{t} \cup \{i\}}}
    $$
    for any $i \in [n]$ any index $t \in [2^n-1]$ and any $t' > t+2^i$.

    From \Cref{lem:ztBounds}, we have for any $t$
    $$
    \beta_{S_{t'}\setminus \{i\}} - \beta_{S_{t} \cup \{i\}} 
    \ge
    \alpha_{S_{t'}\setminus \{i\}} - \alpha_{S_{t} \cup \{i\}} 
    - 2 \cdot 2^{-\tn}
    \ge
    2^{-18n} - 2 \cdot 2^{-\tn}.
    $$
    Thus, picking $\eps = 2^{-\theta(n)}$ is enough to satisfy the above constraint for the rounded instance.
\end{proof}

\end{document}